\RequirePackage{xcolor}
\documentclass[]{article}
\usepackage{array}
\usepackage[utf8]{inputenc}
\usepackage[T1]{fontenc}
\usepackage[english]{babel}
\usepackage{amsmath, amsthm, amssymb, mathtools}
\usepackage{shadethm}
\usepackage[nothm]{thmbox}
\usepackage{graphicx}
\usepackage{xcolor}
\usepackage{psfrag, pstool}
\usepackage{footnote}
\usepackage{geometry}
\usepackage{enumerate,enumitem}
\usepackage[affil-it]{authblk}
\usepackage{float}

\usepackage[final]{changes}
\definechangesauthor[color=blue]{josa}
\definechangesauthor[color=green]{mark}
\usepackage{cleveref}
\crefname{algorithm}{algorithm}{algorithms}
\Crefname{algorithm}{Algorithm}{Algorithms}
\makeatletter
\newcommand*{\algotitle}[2]{%
  \stepcounter{algocf}%
  \hypertarget{algocf.title.\theHalgocf}{}%
  \NR@gettitle{#1}%
  \label{#2}%
  \addtocounter{algocf}{-1}%
}
\makeatother

\usepackage{pgfplots}
\pgfplotsset{compat=default}
\graphicspath{{graphics/}}
\usepackage{subcaption}
\usepackage{caption}
\usepackage{array}
\usepackage{thmtools}
\usepackage[vlined, ruled, boxed]{algorithm2e}

\newtheorem{theorem}{Theorem}[section]
\newtheorem{corollary}[theorem]{Corollary}
\newtheorem{lemma}[theorem]{Lemma}
\newtheorem{definition}[theorem]{Definition}
\newtheorem{remark}[theorem]{Remark}

\newcolumntype{L}{>{$}l<{$}}

\newcommand{\0}{\mathbf{0}}
\newcommand{\av}{\mathbf{a}}
\newcommand{\bv}{\mathbf{b}}
\newcommand{\z}{\mathbf{z}}
\newcommand{\p}{\mathbf{p}}
\newcommand{\y}{\mathbf{y}}
\newcommand{\x}{\mathbf{x}}
\newcommand{\aj}{\mathbf{a}}
\newcommand{\g}{\mathbf{g}}
\newcommand{\vv}{\mathbf{v}}
\newcommand{\cj}{\mathbf{c}}
\newcommand{\dist}{\mathrm{dist}}
\newcommand{\vol}{\mathrm{Vol}}

\newcommand{\supp}{\mathrm{supp}}
\newcommand{\sign}{\mathrm{sign}}

\newcommand{\tube}{\mathrm{tube}}
\newcommand{\conv}{\mathrm{conv}}
\newcommand{\id}{\mathrm{Id}}

\newcommand{\Span}{\mathrm{span}}

\newcommand{\reach}{\mathrm{reach}}
\newcommand{\dd}{\mathrm{d}}
\newcommand{\diam}{\mathrm{diam}}

\newcommand{\B}{\mathcal{B}}
\newcommand{\C}{\mathcal{C}}
\newcommand{\Ch}{\widehat{\mathcal{C}}}

\newcommand{\N}{\mathcal{N}}
\newcommand{\M}{\mathcal{M}}
\newcommand{\Mh}{\widehat{\M}}
\newcommand{\Mr}{\mathcal{M}^{rel}}
\renewcommand{\O}{\mathcal{O}}
\renewcommand{\P}{\mathbb{P}}
\newcommand{\Ph}{\widehat{\mathbb{P}}}

\newcommand{\ch}{\widehat{\bf c}}
\newcommand{\Vh}{\widehat{V}}
\newcommand{\PP}{\mathcal{P}}
\newcommand{\St}{\mathbb{S}}
\newcommand{\R}{\mathbb{R}}

\newcommand{\tab}{\;\;\;\;}

\newcommand{\blue}[1]{#1}

\newcommand{\eps}{\varepsilon}
\usepackage{bbm}
\newcommand{\eins}{\mathbbm{1}}
\newcommand{\X}{\mathcal{X}}
\newcommand{\Vol}{\mathrm{Vol}}
\DeclareMathOperator*{\argmin}{arg\,min}
\DeclareMathOperator*{\argmax}{arg\,max}

\newcommand{\E}[2][]{\textnormal{\;$\mathbb{E}$}_{#1}\!\left[#2\right]}
\newcommand{\Ejk}[3][]{\textnormal{\;$\mathbb{E}_{#1}$}_{#2}\!\left[#3\right]}

\renewcommand{\Pr}[2][]{\textnormal{\;$\mathrm{Pr}$}_{#1}\!\left[#2\right]}

\title{On Recovery Guarantees for \\ One-Bit Compressed Sensing on Manifolds}
\author[1]{M. A. Iwen\thanks{supported in part by NSF DMS-1416752.}} 
\author[2]{Felix Krahmer${}^\dagger$
}
\author[3]{Sara Krause-Solberg\thanks{supported by the Deutsche Forschungsgemeinschaft (DFG) under Grant SFB Transregio 109 and Emmy-Noether junior research group KR 4512/1-1}}
\author[4]{Johannes Maly\thanks{supported by the Deutsche Forschungsgemeinschaft (DFG) under Grant SPP 1798.}}
\affil[1]{Department of Mathematics, and Department of Computational Mathematics, Science, and Engineering (CMSE), Michigan State University, East Lansing, MI, 48824, USA }
\affil[2]{Department of Mathematics, Technical University of Munich, 85748 Garching, Germany}
\affil[3]{HIP Helmholtz Imaging Platform, DESY, 22607 Hamburg, Germany}
\affil[4]{Department of Mathematics, RWTH Aachen University, 52062 Aachen, Germany}

\date{\today}

\begin{document}

\maketitle
\begin{abstract}
	This paper studies the problem of recovering a signal from one-bit compressed sensing measurements under a manifold model; that is, assuming that the signal lies on or near a manifold of low intrinsic dimension. We provide a convex recovery method based on the Geometric Multi-Resolution Analysis and prove recovery guarantees with a near-optimal scaling in the intrinsic manifold dimension. Our method is the first tractable algorithm with such guarantees for this setting. The results are complemented by numerical experiments confirming the validity of our approach.
\end{abstract}

\section{Introduction}

Linear inverse problems are ubiquitous in many applications in science and engineering. Starting with the seminal works of Cand\`es, Romberg and Tao \cite{Candes2006} as well as Donoho \cite{Donoho2006}, a new paradigm in their analysis became an active area of research in the last decades. Namely, rather than considering the linear model as entirely given by the application, one seeks to actively choose remaining degrees of freedom, often using a randomized strategy, to make the problem less ill-posed. This approach gave rise to a number of recovery guarantees for random linear measurement models under structural data assumptions. The first works considered the recovery of sparse signals; subsequent works analyzed more general union-of-subspaces models \cite{Eldar2009} and the recovery of low rank matrices \cite{Recht2010}, a model that can also be employed when studying phaseless reconstruction problems \cite{Candes2013} or bilinear inverse problems \cite{Ahmed2014}.

Another line of works following this approach studies manifold models. That is, one assumes that the structural constraints are given by (unions of finitely many)  manifolds. While this model is considerably richer than say sparsity, its rather general formulation makes a unified study, at least in some cases, somewhat more involved. The first work to study random linear projections of smooth manifold was \cite{Baraniuk2006}, the authors show that Gaussian linear dimension reductions typically preserve the geometric structure. In \cite{Iwen2013}, these results are refined and complemented by a recovery algorithm, which is based on the concept of the Geometric Multi-Resolution Analysis as introduced in \cite{Allard2012} (cf.~Section~\ref{subsec:GMRA} below). These results were again substantially improved in \cite{Eftekhari2015}; these latest results no longer explicitly depend on the ambient dimension.

Arguably, working with manifold models is better adapted to real world data than sparsity and hence may allow one to work with smaller embedding dimensions. For that, however, other practical issues need to be considered as well. In particular, to our knowledge there are almost no works to date that study the effects of quantization, i.e., representing the measurements using only a finite number of bits (the only remotely connected work that we are aware of is \cite{Mondal2007}, but this paper does not consider dimension reduction and exclusively focuses on the special case of Grassmann manifolds). 

For sparse signal models, in contrast, quantization of subsampled random measurements is an active area of research. On the one hand, a number of works considered the scenario of memoryless scalar quantization, that is, each of the measurement is quantized independently. In particular, the special case of representing each measurement only by a single bit, its sign,  -- often referred to as {\em one-bit compressed sensing} -- has received considerable attention. In \cite{Jacques2013}, it was shown that one-bit compressed sensing with Gaussian measurements approximately preserves the geometry, and a heuristic recovery scheme was presented. In \cite{Plan2013LP}, recovery guarantees for a linear method, again with Gaussian measurements, were derived. Subsequently, these results were generalized to subgaussian measurements \cite{Ai2014}, and partial random circulant measurements \cite{Dirksen2017}. In \cite{Plan2013}, the authors provided a recovery procedure for noisy one-bit Gaussian measurements which provably works on more general signal sets (essentially arbitrary subsets of the euclidean ball). This procedure, however, becomes NP-hard as soon as the signal set is non-convex, a common property of manifolds.

Another line of works studied so-called feedback quantizers, that is, the bit sequence encoding the measurements is computed using a recursive procedure. These works adapt the Sigma-Delta modulation approach originally introduced in the context of bandlimited signals \cite{G87, NST96} and later generalized to frame expansions \cite{BPY06, BPY06b} to the sparse recovery framework. A first such approach was introduced and analyzed for Gaussian measurements in \cite{GLPSY13}; subsequent works generalize the results to subgaussian random measurements \cite{KSY13,FK14}. Recovery guarantees for a more stable reconstruction scheme based on convex optimization were proved for subgaussian measurements in \cite{SAAB2018123} and extended to partial random circulant matrices in \cite{FKS17}. For more details on the mathematical analysis available for different scenarios, we refer the reader to the overview chapter \cite{BJKS15}. \\
In this paper we focus on the MSQ approach and leave the study of Sigma-Delta quantizers under manifold model assumptions for future work.

\subsection{Contribution}

We provide the first tractable one-bit compressed sensing algorithm for signals which are well approximated by manifold models. It is simple to implement and comes with error bounds that basically match the state-of-the-art recovery guarantees in \cite{Plan2013}. In contrast to the minimization problem introduced in \cite{Plan2013} which does not come with a minimization algorithm, our approach always admits a convex formulation and hence allows for tractable recovery. Our approach is based on the Geometric Multi-Resolution Analysis (GMRA) introduced in \cite{Allard2012}, and hence combines the approaches of \cite{Iwen2013} with the general results for one-bit quantized linear measurements provided in \cite{Plan2013},\cite{Plan2014}.

\subsection{Outline}

\paragraph{} We begin by a detailed description of our problem in Section \ref{ProblemFormulation} and fix notation for the rest of the paper. The section also includes a complete axiomatic definition of GMRA. Section \ref{MainResults} states our main results. The proofs can be found in Section \ref{Proofs}. In Section \ref{Numerics} we present some numerical experiments testing the recovery in practice and conclude with Section \ref{Discussion}. Technical parts of the proofs as well as adaption of the results to GMRAs from random samples are deferred to the Appendix.

\section{Problem Formulation, Notation, and Setup} 
\label{ProblemFormulation}

\paragraph{} The problem we address is the following. We consider a \added[id=josa]{given union of} low-dimensional \replaced[id=josa]{manifolds}{ manifold} (i.e., signal class) $\M$ of intrinsic dimension $d$ \replaced[id=josa]{that is a subset of}{ which lives on} the unit sphere $\St^{D-1}$ of a higher dimensional space $\R^D$, $d \ll D$.  Furthermore, we image that we do not know $\M$ perfectly, and so instead we only have approximate information about $\M$ represented in terms of a structured dictionary model $\mathcal{D}$ for the manifold.  Our goal is now to recover an unknown signal $\x \in \M$ from $m$ one-bit measurements
\begin{align} \label{one-bit-measurements}
\y = \sign(A\x),
\end{align}
where $A \in \R^{m\times D}$ has Gaussian i.i.d.\ entries of variance $1/\sqrt{m}$, using as few measurements, $m$, as possible. \added[id=josa]{Each single measurement $\sign(\langle \aj_i,\x \rangle)$ can be interpreted as the random hyperplane $\{ \z \in \R^D \colon \langle \aj_i,\z \rangle = 0\}$ tessellating the sphere (cf.\ Figure \ref{Z6}).} In order to succeed using only $m \ll D$ such one-bit measurements we will use the fact that \deleted[id=josa]{$\x$ is approximately sparse in} our (highly coherent, but structured) dictionary $\mathcal{D}$ for $\M$ \added[id=josa]{provides structural constraints for the signal $\x$ to be recovered}.  \replaced[id=josa]{Thus the setup connects to recent generalizations of the quantized compressed sensing problem \cite{Plan2013} which we will exploit in our proof.}{ Note that this now becomes a quantized compressive sensing problem (see, e.g., \cite{SAAB2018123}) for signals which are nearly sparse in terms of the manifold dictionary $\mathcal{D}$ for $\M$.}

\begin{figure}
	\centering
	\begin{subfigure}{0.4\textwidth}
		\centering
		\def\svgwidth{0.5\linewidth}
		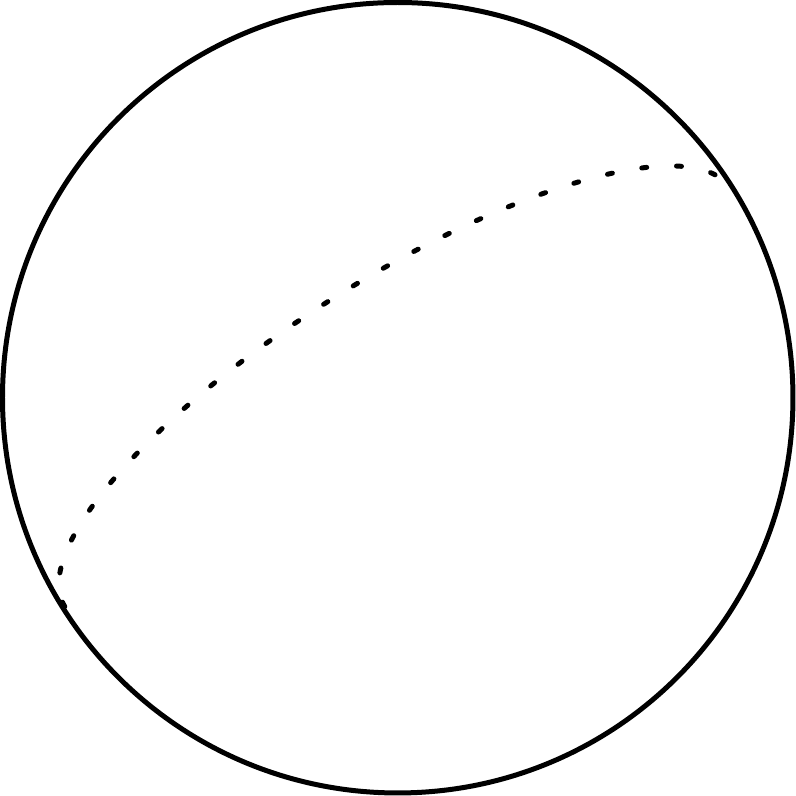
		\subcaption{Tessellation of the sphere by random hyperplanes.}\label{Z6}
	\end{subfigure}
	\begin{subfigure}{0.4\textwidth}
		\centering
		\def\svgwidth{0.5\linewidth}
		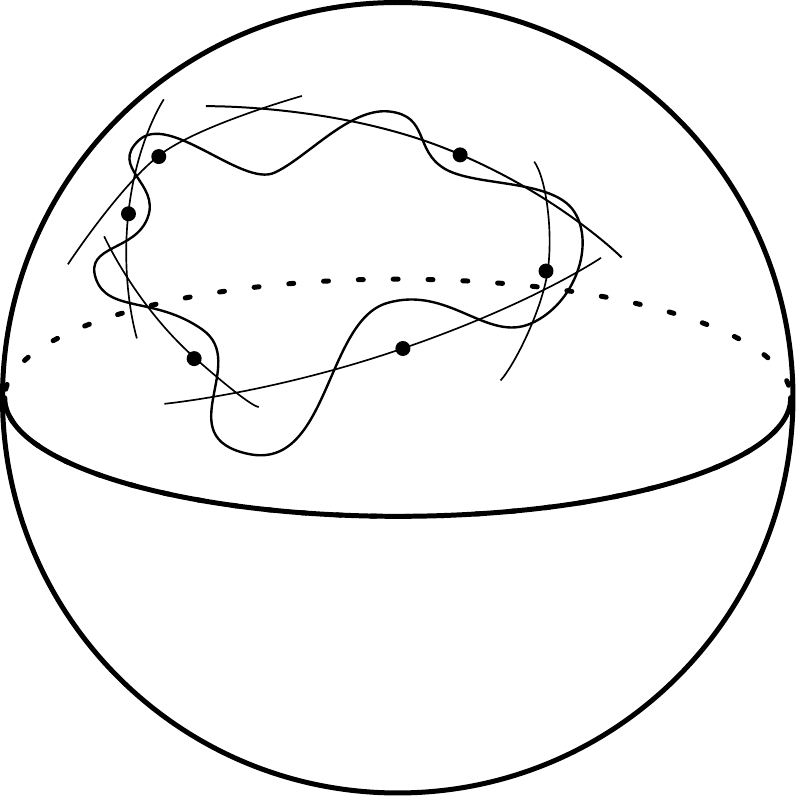
		\subcaption{Submanifold $\M$ of $\St^{D-1}$ and one level of GMRA.}\label{Z5}
	\end{subfigure}
	\caption{One-bit measurements and GMRA.}\label{Z5+Z6}
\end{figure}

\subsection{GMRA Approximations to $\M \subset \R^D$, and Two Notions of Complexity}\label{subsec:GMRA}

Clearly, the solution to this problem depends on what kind of representation, $\mathcal{D}$, of the manifold $\M$ we have access to. In this paper we consider the scenario where the dictionary for the manifold is provided by a Geometric Multi Resolution Analysis (GMRA) approximation to $\M$ \cite{Allard2012} (cf.\ Figure \ref{Z5}). \replaced[id=josa]{We will mainly work with}{ In particular, the type of} GMRA approximations of $\M$ \deleted[id=josa]{ we will utilize herein is } 
characterized by the axiomatic Definition \ref{GMRA} below, \added[id=josa]{but we also consider the case of a GMRA approximation based on random samples from $\M$ (see Section \ref{MainResults} and Appendices~\ref{EmpiricalGMRA}~and~\ref{EmpiricalGMRAproof} for more details).} 

As one might expect, the complexity and structure of the GMRA-based dictionary for $\M$ will depend on the complexity of $\M$ itself. In this paper we will \replaced[id=josa]{work with}{utilize} two different measures of a set's complexity:  $(i)$ the set's {\it Gaussian width}, and $(ii)$ the notion of the {\it reach} of the set \cite{Federer1959}.  The Gaussian width of a set $\M \subset \R^D$ is defined by 
$$w(\M) := \E {\sup_{\z \in \M} \langle \g,\z \rangle}$$ where $\g \sim \mathcal{N}(\0,I_D)$. Properties of this quantity are discussed in Section \ref{Toolbox}.  The notion of reach is, in contrast, more obviously linked to the geometry of $\M$ and requires a couple of additional definitions before it can be defined formally. 

The first of these definitions is the {\it tube of radius $r$} around a given subset $\M \subset \R^D$, which is the $D$-dimensional superset of $\M$ consisting of all the points in $\R^D$ that are within Euclidean distance $r \geq 0$ of $\M \subset \R^D$,
$$\tube_r(\M) := \left\{ \x \in \R^D ~\colon~ \inf_{\y \in \M} \| \x - \y \|_2 \leq r \right\}.$$
The domain of the nearest neighbor projection onto the closure of $\M$ is also needed, and is denoted by 
$$D(\M) := \left\{\x\in\mathbb \R^D ~\colon~\exists! \y \in \overline{\M} \text{ such that }\|\x-\y\|_2=\inf_{\z\in\M}\| \x - \z \|_2 \right\}.$$
Finally, the reach of the set $\M \subset \R^D$ is simply defined to be the \replaced[id=josa]{smallest}{ largest} distance $r$ around $\M$ for which the nearest neighbor projection onto the closure of $\M$ is \added[id=josa]{no longer} well defined.  \replaced[id=josa]{Equivalently,}{ That is,} 
$$\reach(\M):=\sup\{r\geq0 ~\colon~ \tube_r(\M) \subseteq D(\M)\}.$$ 
Given this definition one can see, e.g., that the reach of any $d < D$ dimensional sphere of radius $r$ in $\R^D$ is always $r$, and that the reach of any $d \leq D$ dimensional convex subset of $\R^D$ is always $\infty$.

\begin{definition}[GMRA  Approximation to $\M$, \cite{Iwen2013}]\label{GMRA}
	Let $J \in \mathbb{N}$ and $K_0,K_1,...,K_J \in \mathbb{N}$. \replaced[id=josa]{Then a \emph{Geometric Multi Resolution Analysis (GMRA) Approximation} of $\M$ is a collection $\{ (\C_j,\PP_j) \}$, $j \in [J] := \{0,...,J\}$, of sets }{ For each $j \in [J] := \{0,...,J\}$ we assume to have sets} $\C_j = \{ \cj_{j,k} \}_{k = 1}^{K_j} \subset \R^D$ of centers and
	\begin{align*}
	\PP_j = \left\{ \P_{j,k}: \R^D \rightarrow \R^D ~\big\vert~ k \in [K_j] \right\}
	\end{align*}
	of affine projectors which approximate $\M$ at scale $j$, \replaced[id=josa]{such that }{ These form a \emph{Geometric Multi Resolution Analysis (GMRA) Approximation} of $\M$ if} the following assumptions \eqref{first}-\eqref{third} hold.
	\begin{enumerate}[label=(\arabic*),ref={\arabic*}]
		\item\label{first}\textbf{Affine Projections:} Every $\P_{j,k} \in \PP_j$ has both an associated center $\cj_{j,k} \in \C_j$ and an orthogonal matrix $\Phi_{j,k} \in \R^{d \times D}$, such that
		\begin{align*}
		\P_{j,k}(\z) = \Phi_{j,k}^T \Phi_{j,k} (\z-\cj_{j,k}) + \cj_{j,k},
		\end{align*}
		i.e., $\P_{j,k}$ is the projector onto some affine $d$-dimensional linear subspace $P_{j,k}$ containing $\cj_{j,k}$.
		\item\label{second}\textbf{Dyadic Structure:} The number of centers at each level is bounded by $|\C_j| = K_j \leq C_{\C}2^{dj}$ for \added[id=josa]{an absolute constant} $C_{\C} \geq 1$.  
		There exist $C_1 > 0$ and $C_2 \in (0,1]$, such that following conditions are satisfied:
		\begin{enumerate}[label=(\alph*), ref={\theenumi\alph*}]
			\item \label{GMRA2a} $K_j \le K_{j+1}$, for all $j \in [J-1]$.
			\item \label{GMRA2b} $\Vert \cj_{j,k_1} - \cj_{j,k_2} \Vert_2 > C_1 \cdot 2^{-j}$, for all $j \in [J]$ and $k_1 \neq k_2 \in [K_j]$.
			\item \label{GMRA2c} For each $j \in [J]\backslash \{ 0 \}$ there exists a parent function $p_j : [K_j] \rightarrow [K_{j-1}]$ with
			\begin{align*}
			\Vert &\cj_{j,k} - \cj_{j-1,p_j(k)} \Vert_2 \leq C_2 \cdot \min_{k' \in [K_{j-1}] \backslash \{ p_j(k) \}} \Vert \cj_{j,k} - \cj_{j-1,k'} \Vert_2.
			\end{align*}
		\end{enumerate}
		\item\label{third} \textbf{Multiscale Approximation:} The projectors in $\PP_j$ approximate $\M$ at scale $j$, i.e.,  when $\M$ is sufficiently smooth the affine spaces $P_{j,k}$ locally approximate $\M$ pointwise with error $\O \left( 2^{-2j} \right)$. More precisely:
		\begin{enumerate}[label=(\alph*), ref={\theenumi\alph*}]
			\item \label{tube} There exists $j_0 \in [J-1]$, such that $\cj_{j,k} \in \tube_{C_1 \cdot 2^{-j-2}} (\M)$, for all $j > j_0 \geq 1$ and $k \in [K_j]$. 
			\item \label{3b} For each $j \in [J]$ and $\z \in \R^D$ let $\cj_{j, k_j(\z)}$ be one of the centers closest to $\z$, i.e.,
			\begin{align} \label{copt}
			k_j(\z) \in \argmin_{k \in [K_j]} \Vert \z - \cj_{j,k} \Vert_2.
			\end{align}
			Then, for each $\z \in \M$ there exists a constant $C_\z > 0$ such that
			\begin{align*}
			\Vert \z - \P_{j, k_j(\z)} (\z) \Vert_2 \le C_\z \cdot 2^{-2j},
			\end{align*}
			for all $j \in [J]$. Moreover, for each $\z \in \M$ there exists $\tilde{C}_\z > 0$ such that
			\begin{align*}
			\Vert \z - \P_{j,k'}(\z) \Vert_2 \le \tilde{C}_\z \cdot 2^{-j},
			\end{align*}
			for all $j \in [J]$ and $k' \in [K_j]$ satisfying
			\begin{align*}
			\Vert \z - \cj_{j,k'} \Vert_2 \le 16 \cdot \max \left\{ \Vert \z - \cj_{j,k_j(\z)} \Vert_2, C_1 \cdot 2^{-j-1} \right\}.
			\end{align*}
		\end{enumerate}
	\end{enumerate}
\end{definition}

\begin{remark} 
By property \eqref{first} GMRA approximation represents $\M$ as a combination of several anchor points (the centers $\cj_{j,k}$) and corresponding low dimensional affine spaces $P_{j,k}$. The levels $j$ control the accuracy of the approximation. The centers are organized in a tree-like structure as stated in property \eqref{second}. Property \eqref{third} then characterizes approximation criteria to be fulfilled on different refinement levels. Note that centers do not have to lie on $\M$ (compare Figure \ref{Z5}) but their distance to $\M$ is controlled by property \eqref{tube}. \\
\blue{If the centers form a maximal $2^{-j}$ packing of a smooth manifold $\M$ at each scale $j$ or if the GMRA is constructed from manifold samples as discussed in \cite{Maggioni2016} (cf.~Appendix \ref{EmpiricalGMRAproof}), the constants $C_1$ and $\tilde C_\z$ are in fact bounded by absolute constants which will become important later on, cf.\ Remark \ref{Rem:PlusFuncGone}.}
\end{remark}

\begin{figure}[ht]
	\centering
	\def\svgwidth{0.45\linewidth}
	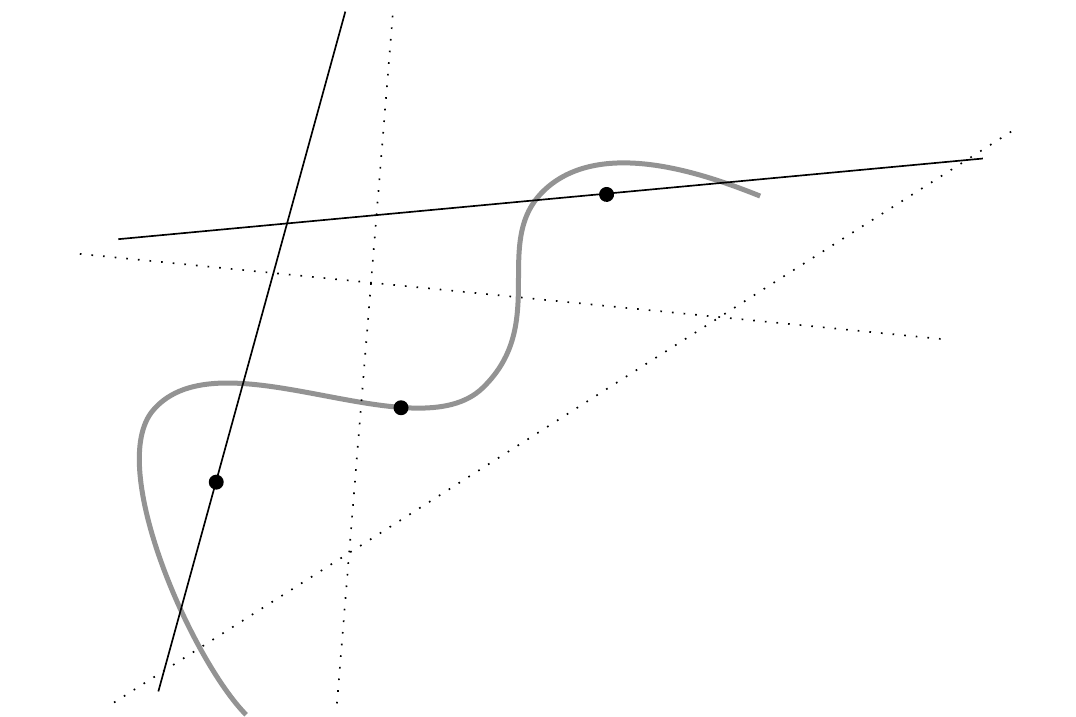
	\caption{The closest center $\cj_{j,k_j(\x)}$ is not identified by measurements. Dotted lines represent one-bit hyperplanes.}\label{Z7}
\end{figure}

\subsection{Additional Notation}

\paragraph{} Let us now fix some additional notation. Throughout the remainder of this paper we will work with several different metrics.  Perhaps most importantly, we will quantify the distance between two points $\z,\z' \in \R^D$ with respect to their one-bit measurements by 
$$d_A(\z,\z') := \frac{d_H(\sign(A\z),\sign(A\z'))}{m},$$ 
where $d_H$ counts the number of differing entries between the two sign patterns (i.e., $d_A(\z,\z')$ is the normalized Hamming distance between the signs of $A\z$ and $A\z'$). 
Furthermore, let $\P_\St$ denote orthogonal projection onto the unit sphere $\St^{D-1}$, and more generally let $\P_K$ denote orthogonal (i.e., nearest neighbor) projection onto the closure of an arbitrary set $K \subset \R^D$ wherever it is defined.  Then, for all $\z, \z' \in \R^D$ we will denote by $d_G(\z,\z') = d_G(\P_\St(\z),\P_\St(\z'))$ the geodesic distance between $\P_\St(\z)$ and $P_\St(\z')$ on $\St^{D-1}$ normalized to fulfill $d_G(\z'',-\z'') = 1$ for all $\z'' \in \R^D$.

Herein the Euclidian ball with center $\z$ and radius $r$ is denoted by $\B(\z,r)$. In addition, the {\it scale-$j$ GMRA approximation to $\M$},  
\begin{equation*}
\M_j := \{ \P_{j,k_j(\z)}(\z) \colon \z \in \B(\0,2) \} \cap \B(\0,2),
\end{equation*}
will refer to the portions of the affine subspaces introduced in Definition \ref{GMRA} for each fixed $j$ which are potentially relevant as approximations to some portion of $\M \subset \St^{D-1}$. 
To prevent the $\M_j$ above from being empty we will further assume in our results that we only use scales $j > j_0$ large enough to guarantee that $\tube_{C_1 2^{-j-2}} (\M) \subset \B(0,2)$.  
Hence we will have $\cj_{j,k} \in \B(\0,2)$ for all $k \in K_j$, and so $\C_j \subset \M_j$.  This further guarantees that no sets $P_{j,k} \cap \B(\0,2)$ are empty, and that $P_{j,k} \cap \B(\0,2) \subset \M_j$ for all $k \in K_j$.

Finally, we write $a \gtrsim b$ if $a \geq Cb$ for some constant $C > 0$. The diameter of a set $K \subset \R^D$ will be denoted by $\diam(K) := \sup_{\z,\z' \in K} \| \z - \z' \|_2$, where $\| \cdot \|_2$ is the Euclidian norm. We use $\dist(A,B) = \inf_{\av \in A, \bv \in B} \| \av - \bv \|_2$ for the distance of two sets $A,B \subset \R^D$ and by abuse of notation $\dist(\0, A) = \inf_{\av \in A} \| \av \|_2$. The operator norm of a matrix $A\in \R^{n_1\times n_2}$ is denoted by  $\| A \|=\sup_{\x\in\mathbb R^{n_2},\|\x\|_2\leq1}\|A\x\|_2$. We will write $\mathcal{N}(K,\eps)$ to denote the Euclidian covering number of a set $K \subset \R^D$ by Euclidean balls of radius $\eps$ (i.e., $\mathcal{N}(K,\eps)$ is the minimum number of $\eps$-balls that are required to cover $K$). And, the operators $\lfloor r \rfloor$ (resp. $\lceil r \rceil$) return the closest integer smaller (resp. larger) than $r \in \R$.

\subsection{The Proposed Computational Approach}

\paragraph{} Combining prior GMRA-based compressed sensing results \cite{Iwen2013} with the one-bit results of Plan and Vershynin in \cite{Plan2013} suggests the following strategy for recovering an unknown $\x \in \M$ from the measurements given in \eqref{one-bit-measurements}:  First, choose a center $\cj_{j,k'}$ whose one-bit measurements agree with as many one-bit measurements of $\x$ as possible. Due to the varying shape of the tessellation cells this is not an optimal choice in general (see Figure \ref{Z7}). Nevertheless, one can expect $P_{j,k'}$ to be a good approximation to $\M$ near $\x$.  Thus, in the second step a modified version of Plan and Vershynin's noisy one-bit recovery method using $P_{j,k'}$ should yield an approximation of $\P_{j,k'}(\x)$ which is close to $\x$.\footnote{Note that in this second step the given measurements $\y$ of $\x$ are interpreted as being noisy measurements of $\P_{j,k'}(\x)$.}  See \nameref{algorithm} for pseudocode.\\

\begin{algorithm}[H]
    \algotitle{OMS-simple}{algorithm}
	\SetAlgoRefName{OMS-simple}
	\begin{minipage}{15.5cm}
		\begin{enumerate}[label=\Roman*.]
			\item \label{I}\Indm Identify a center $\cj_{j,k'}$ close to $\x$ via
			\begin{align} \label{MinC}
			\cj_{j,k'} \in \argmin_{\cj_{j,k} \in \C_j} \; d_H( \sign(A \cj_{j,k}), \y ),
			\end{align}
			where $d_H$ is the Hamming distance, i.e., $d_H(\z,\z') := |\{l: z_l \neq z'_l \}|$. If $d_H( \sign(A \cj_{j,k'}), \y ) = 0$, directly choose $\x^\ast = \cj_{j,k'}$ and omit \ref{II}
			\item \label{II}If there is no center in the same cell as $\x$ (as in Figure \ref{Z7}), solve a noisy one-bit recovery problem as in \cite{Plan2013}, i.e.,
			\begin{align} \label{MinX}
			\begin{split}
			\x^\ast &= \argmin_{\z \in \R^D} \sum_{l=1}^{m} (-y_l) \langle \aj_l, \z \rangle , \quad
			\text{subject to } \z = \P_{j,k'}(\z) \text{ and } \Vert \z \Vert_2 \le R
			\end{split}
			\end{align}
			where $R$ is a suitable parameter.
		\end{enumerate}
	\end{minipage}
	\caption{OnebitManifoldSensing - Simple Version}
\end{algorithm}

\begin{remark} \label{+Remark}
The minimization in (\ref{MinC}) can be efficiently calculated by exploiting tree structures in $\C_j$. Numerical experiments (see Section \ref{Numerics}) suggest this strategy to yield adequate approximation for the center $\cj_{j,k_j(\x)}$ in \eqref{copt}, while being considerably faster (we observed differences in runtime up to a factor of 10).
\end{remark}

Though simple to understand, the constraints in \eqref{MinX} have two issues that we need to address:  First, in some cases the minimization problem \eqref{MinX} empirically exhibits suboptimal recovery performance (see Section \ref{SIMPLEvsCONVEX} for details). Second, the parameter $R$ in \eqref{MinX} is unknown a priori (i.e., \nameref{algorithm} requires parameter tuning, making it less practical than one might like).  
Indeed, our analysis shows that making an optimal choice for $R$ in \nameref{algorithm} requires a priori knowledge about $\| \P_{j,k'}(\x) \|_2$ which is only approximately known in advance.  \deleted[id=josa]{Thankfully, however, the noisy one-bit results of Plan and Vershynin apply to arbitrary subsets of the unit ball $\B(\0,1) \subset \R^D$ which will allow us to adapt our recovery approach.}

To address this issue, we will modify the constraints in \eqref{MinX} and instead minimize over the convex hull of the nearest neighbor projection of $P_{j,k'} \cap \B(\0,2)$ onto \replaced[id=josa]{$\St^{D-1}$}{$\B(\0,1)$}, 
\begin{equation*}
\conv \left( \P_\St (P_{j,k'} \cap \B(\0,2)) \right),
\end{equation*}
to remove the $R$ dependence. 
If $\0 \in P_{j,k'}$ one has $\conv \left( \P_\St (P_{j,k'} \cap \B(\0,2)) \right) = P_{j,k'} \cap \B(\0,1)$. If $\0 \notin P_{j,k'}$ the set $\conv \left( \P_\St (P_{j,k'} \cap \B(\0,2)) \right)$ is described by the following set of convex constraints which are straightforward to implement in practice. Denote by $\P_\cj$ the projection onto the vector $\cj = \P_{j,k'}(\0)$.  Then,
\begin{align} \label{ConvexSet}
\z \in \conv \left( \P_\St (P_{j,k'} \cap \B(\0,2)) \right) \Leftrightarrow \begin{cases}
\| \z \|_2 \le 1, \\
\Phi_{j,k'}^T \Phi_{j,k'} \z + \P_\cj(\z) = \z, \\
\langle \z,\cj \rangle \ge \frac{1}{2} \| \cj \|_2^2,
\end{cases}.
\end{align}

The first two conditions above restrict $\z$ to \deleted[id=josa]{the intersection of} $\B(\0,1)$ and $\Span (P_{j,k'})$\added[id=josa]{, respectively}. The third condition then removes all points that are too close to the origin (see Figure \ref{Z10}). \added[id=josa]{ A rigorous proof of equivalence can be found in Appendix \ref{ConvexHull}.}
\begin{figure}[ht]
	\centering
	\def\svgwidth{0.8\linewidth}
	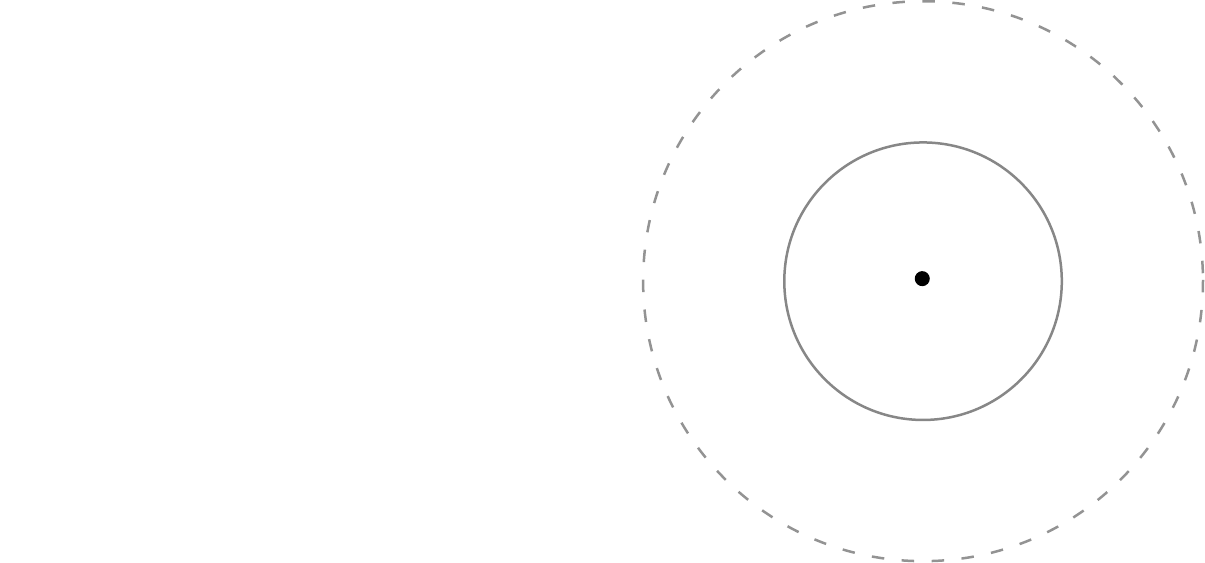
	\caption{Two views of an admissible set $\conv(\P_\St(P_{j,k'}\cap\mathcal B(0,2)))$ from \eqref{ConvexSet} for a case with $\| \cj \|_2 = \| \P_{j,k'}(\0) \|_2 < 1$.} 
	\label{Z10}
\end{figure}
\added[id=josa]{Our analysis uses that the noisy one-bit recovery results of Plan and Vershynin apply to arbitrary subsets of the unit ball $\B(\0,1) \subset \R^D$ which will allow us to adapt our recovery approach.} Replacing the constraints in \eqref{MinX} with those in \eqref{ConvexSet} we obtain the following modified recovery approach, \nameref{algorithm2}.

\begin{algorithm}[H]
    \algotitle{OMS}{algorithm2}
	\SetAlgoRefName{OMS}
	\begin{minipage}{15.5cm}
		\begin{enumerate}[label=\Roman*.]
			\item \label{I2}\Indm Identify a center $\cj_{j,k'}$ close to $\x$ via
			\begin{align} \label{MinCnew}
			\cj_{j,k'} \in \argmin_{\cj_{j,k} \in \C_j} \; d_H( \sign(A\cj_{j,k}),\y ).
			\end{align}
			where $d_H$ is the Hamming distance, i.e., $d_H(\z,\z') := |\{l: z_l \neq z'_l \}|$. If $d_H( \sign(A\cj_{j,k'}),\y ) = 0$, directly choose $\x^\ast = \cj_{j,k'}$ and omit \ref{II2}
			\item \label{II2}If there is no center lying in the same cell as $\x$ (see Figure \ref{Z7}), recover the projection of $\x$ onto $P_{j,k'}$, i.e., $\P_{j,k'}(\x)$. To do so solve the convex optimization
			\begin{align} \label{MinXnew}
			\begin{split}
			\x^\ast &= \argmin_{\z \in \R^D} \sum_{l=1}^{m} (-y_l) \langle \aj_l,\z \rangle, \quad
			\text{subject to } \z \in \conv \left( \P_\St (P_{j,k'} \cap \B(\0,2)) \right).
			\end{split}
			\end{align}
		\end{enumerate}
	\end{minipage}
	\caption{OnebitManifoldSensing}
\end{algorithm}

As we shall see, theoretical error bounds for both \nameref{algorithm} and \nameref{algorithm2} can be obtained by nearly the same analysis despite their differences.

\section{Main Results} 
\label{MainResults}

\paragraph{} In this section we present the main results of our work, namely that both \nameref{algorithm} and \nameref{algorithm2} approximate a signal on $\M$ to arbitrary precision with a near-optimal number of measurements. More precisely, we obtain the following theorem.


%

\blue{
\begin{theorem}[(Uniform) Recovery] \label{GeneralApproximation}
	There exist absolute constants $E,E',c > 0$ such that the following holds. Let $\epsilon \in (0,1)$ and assume the GMRA's maximum refinement level $J \geq j := \lceil \log(1/\eps) \rceil$.  
	Further suppose that one has $\dist (\0,\M_j) \ge 1/2$, $0<C_1 < 2^{j}$, and $\sup_{\x \in \M} \tilde{C}_\x < 2^{j - 2}$.  If     
	\begin{align} \label{mBoundUniformRecoveryCorollary}
	m \ge E C_1^{-6} \eps^{-6} \max \left\{w(\M), \sqrt{d\log(e/\eps)} \right\}^2,
	\end{align}
	then with probability at least $1-12\exp(-c C_1^2 \eps^2 m)$ for all $\x \in \M \subset \St^{D-1}$ the approximations $\x^\ast$ obtained by \nameref{algorithm2} satisfy
	\begin{align}
	\label{equ:GeneralApprox}
	\| \x - \x^\ast \|_2^2 \le E' \left( 1 + \tilde{C}_\x + C_1 \max \left\{ 1,\log(C_1^{-1}) \right\} \right)^2  \eps \log\left( \frac{1}{2\eps} \right). 
	\end{align}
\end{theorem}
\begin{proof}
The proof can be found below Theorem \ref{ApproximationTheorem} in Section \ref{Proofs}.
\end{proof}
}

\begin{remark}   \label{Rem:PlusFuncGone}
	   \blue{Let us briefly comment on the assumptions of Theorem \ref{GeneralApproximation}. Since $\M \subset \St^{D-1}$, requiring $\dist (\0,\M_j) \ge 1/2$ is a mild assumption. Any GMRA not fulfilling it would imply a worst-case reconstruction error of $1/2$ in \eqref{equ:GeneralApprox}. The constant $1/2$ was chosen for simplicity and can be replaced by an arbitrary number in $(0,1)$. This, however, influences the constants $E,E',c$. }\\
	   The restrictions on $C_1$ and $\tilde C_\x$ are easily satisfied, e.g., if the centers form a maximal $2^{-j}$ packing of $\M$ at each scale $j$ or if the GMRA is constructed from manifold samples as discussed in \cite{Maggioni2016} (cf.~Appendix \ref{EmpiricalGMRAproof}). In both these cases $C_1$ and $\tilde C_\x$ are in fact bounded by absolute constants.
\end{remark}

\paragraph{} Note that Theorem \ref{GeneralApproximation} depends on the Gaussian width of $\M$. For general sets this quantity provides a useful measure of the set's complexity. \blue{In the case of compact Riemannian submanifolds of $\R^D$ it might be more convenient to have a dependence on the geometric properties of $\M$ instead (e.g., its volume and reach).}  Indeed, one can show by means of \cite{Eftekhari2015} that $w(\M)$ can be upper bounded in terms of the manifold's intrinsic dimension $d$, its $d$-dimensional volume $\Vol(\M)$, and the inverse of its reach. Intuitively, these dependencies are to be expected as a manifold with fixed intrinsic dimension $d$ can become more complex as either its volume or curvature (which can be bounded by the inverse of its reach) grows.  The following theorem, which is a combination of different results in \cite{Eftekhari2015}, formalizes this intuition by bounding the Gaussian width of a manifold in terms of its geometric properties.

\begin{theorem}
         \label{Riemann}
	Assume $\M \subset \R^D$ is a compact $d$-dimensional Riemannian manifold with $d$-dimensional volume $\Vol(\M)$ where $d \ge 1$. Then one can replace $w(\M)$ in above theorem by
	\begin{equation*}
	w(\M) \leq C\cdot\diam(\M) \cdot \sqrt{d \cdot \max \left\{\log\left( c\frac{\sqrt{d}}{\min\{1,\reach(\M)\}}\right), 1 \right\} + \log(\max\{1,\Vol(\M)\})}.
	\end{equation*}
	where $C,c > 0$ are absolute constants.
\end{theorem}

\begin{proof}
See Appendix~\ref{ProofRieman}.
\end{proof}

\begin{remark}
Note that in our setting $\M \subset \St^{D-1}$ implies that $\diam(\M) \le 2$ and $\reach(\M) \le 1$. As we will see the Gaussian width of the GMRA approximation to $\M$ is also bounded in terms of $w(\M)$.  This additional width bound is crucial to the proof of Theorem \ref{GeneralApproximation} as the complexity of the GMRA approximation to $\M$ also matters whenever one attempts to approximate an $\x \in \M$ using only the available GMRA approximation to $\M$.  See, e.g., Lemmas~\ref{BoundOfGaussianWidth},~\ref{BoundOfGaussianWidthFine}~and~\ref{RiemannianGWidthBound} below for upper bounds on the Gaussian widths of GMRA approximations to manifolds $\M \subset \St^{D-1}$ in various settings.
\end{remark}

\paragraph{} Finally, we point out that Theorem \ref{GeneralApproximation} assumes access to a GMRA approximation to $\M \subset \St^{D-1}$ which satisfies all of the axioms listed in Definition \ref{GMRA}.  Following the work of Maggioni, Minsker, and Strawn \cite{Maggioni2016}, however, one can also ask whether a similar result will still hold if the GMRA approximation one has access to has been learned by randomly sampling points from $\M$ without the assumptions of Definition \ref{GMRA} being guaranteed a priori.  Indeed, such a setting is generally more realistic.  In fact it turns out that a version of Theorem \ref{GeneralApproximation} still holds for such empirical GMRA approximations under suitable conditions; see Theorem \ref{ApproximationTheorem2}. We refer the interested reader to Appendix \ref{EmpiricalGMRA} and Appendix \ref{EmpiricalGMRAproof} for additional details and discussion regarding the use of such empirically learned GMRA approximations.

\section{Proofs} 
\label{Proofs}

This section provides proofs of the main result in both settings described above and establishes several technical lemmas. First, properties of the Gaussian width and the geodesic distance are collected and shown. Then, the main results are proven for a given GMRA approximation fulfilling the axioms. 

\subsection{Toolbox} 
\label{Toolbox}

We start by connecting slightly different definitions of dimensionality measures similar to the Gaussian width and clarify how they relate to each other. This is necessary as the tools we make use of appear in their original versions referring to different definitions of Gaussian width.
\begin{definition}[Gaussian (mean) width]
	Let $g\sim\mathcal N(\0,\id_D)$. For a subset $K \subset \R^D$ define
	\begin{itemize}
		\item[(i)] the \emph{Gaussian width}: $w(K) := \E{\sup_{\x\in K}\langle \g,\x\rangle}$
		\item[(ii)] the \emph{Gaussian mean width} to be the Gaussian width of $K-K$ and
		\item[(iii)] the \emph{Gaussian complexity}: $\gamma(K)=\E{\sup_{\x\in K}|\langle \g,\x\rangle|}$.
	\end{itemize}
	By combining Properties 5. and 6. of  Proposition 2.1  in \cite{Plan2013} on has
	\begin{align} \label{WidthInequality}
	    w(K-K)\leq 2w(K)\leq 2\gamma(K)\leq 2\left(w(K-K)+\sqrt{\frac 2 \pi}\mathrm{dist}(\0,K)\right).
	\end{align}
\end{definition}

\begin{remark} \label{GaussianWidthRemark} One can easily verify that $w(K) \geq 0$ for all $K \subset \R^D$ since $w(K) := \E{\sup_{\x\in K}\langle \g,\x\rangle} \geq \sup_{\x\in K} \E{\langle \g,\x\rangle} = 0$.  The square of the Gaussian width $w(K \cap \B(\0,1))^2$ of $K\subset \R^D$ is also a good measure of intrinsic dimension. For example, if $K$ is a linear subspace with $\mathrm{dim}(K)=d$ then $w(K \cap \B(\0,1)) \leq \sqrt{d}$. In this sense, the Gaussian width extends the concept of dimension to general sets $K$.  Furthermore, for a finite set $K$ the Gaussian width is bounded by $w(K) \le C_f \; \diam(K \cup \{ \0 \}) \sqrt{\log|K|}$. This can be deduced directly from the definition (see, e.g., \S2 of \cite{Plan2013}).  
\end{remark}


\paragraph{} Now that we have introduced the notion of Gaussian width, we can use it to characterize the \replaced[id=josa]{union }{size} of the given manifold and \replaced[id=josa]{a single }{one} level of \added[id=josa]{its} GMRA approximation $\M \cup \M_j$ (recall the definition of $\M_j$ in Section \ref{ProblemFormulation}).

\begin{lemma}[A Bound of the Gaussian Width for Coarse Scales] \label{BoundOfGaussianWidth}
	For $\M_j$, the subspace approximation in the GMRA of level $j > j_0$ (cf.\ end of Section \ref{ProblemFormulation}) for $\M$ of dimension $d \geq 1$, the Gaussian width of $\M\cup\M_j$ can be bounded from above and below by
	\begin{align*}
		\max \{ w(\M),w(\M_j) \} \le w(\M\cup\M_j) \le 2w(\M) + 2w(\M_j) + 3 \le 2w(\M)+ C\sqrt{dj}.
	\end{align*}
\end{lemma}

\begin{remark} \label{rem:BoundOfGaussianWidth}
	Note that the first inequality holds for general sets, not only $\M$ and $\M_j$. Moreover, one only uses $\M_j \subset \B(\0,2)$ to prove the second inequality. It thus holds for $\M_j$ replaced with arbitrary subsets of $\B(\0,2)$. We might use both variations referring to Lemma \ref{BoundOfGaussianWidth}.
\end{remark}

\begin{proof}
    The first inequality follows by noting that
    \begin{align*}
    	\max \{ w(\M),w(\M_j) \} = \max \left\{ \E[]{\sup_{\vv \in \M} \langle \vv, \g \rangle}, \E[]{\sup_{\vv \in \M_j} \langle \vv,\g \rangle} \right\} \le \E[]{\sup_{\vv \in \M \cup \M_j} \langle \vv,\g \rangle} = w(\M \cup \M_j).
    \end{align*}
    To obtain the second inequality observe that
    \begin{align}
    \begin{split} \label{equ:GenWidthUnionBound}
        w(\M \cup \M_j) &\le \gamma (\M \cup \M_j) \le \E{\sup_{\vv\in\M} |\langle \vv,\g\rangle| + \sup_{\vv\in\M_j} |\langle \vv,\g\rangle|} = \gamma (\M)+\gamma (\M_j)  \\
        &\le 2w(\M) + 2w(\M_j) + \sqrt{\frac 2 \pi}\mathrm{dist}(\0,\M) + \sqrt{\frac 2 \pi}\mathrm{dist}(\0,\M_j) \\
        &\le 2\left( w(\M) + w(\M_j) + 1.5\sqrt{\frac 2 \pi} \right) 
    \end{split}
    \end{align}
    where we used \eqref{WidthInequality}, the fact that $\M \subset \St^{D-1}$, and that $\M_j \subset \B(\0,2)$. 

    For the last inequality we bound $w(\M_j)$. First, note that
	\begin{align*}
    	w (\M_j)&= \E{\sup_{\vv\in\M_j}\langle \vv,\g\rangle} = \E{\sup_{\vv\in \{ \P_{j,k_j(\x)}(\x) \colon \x \in \B(\0,2) \} \cap \B(\0,2) } \langle \vv,\g\rangle}\\
    	&\le \E{\sup_{\x\in \bigcup_{k \in [K_j]} P_{j,k} \cap \B(\0,2)}\langle \x,\g\rangle}.
    \end{align*}
    For all $k \in [K_j]$ there exist $d$-dimensional Euclidean balls \replaced[id=josa]{$L_{j,k}\subset P_{j,k}$ }{$L_{j,k}$} of radius $2$ such that $P_{j,k} \cap \B(\0,2) \subset L_{j,k}$. Hence, $\bigcup_{k \in [K_j]} (P_{j,k} \cap \B(\0,2)) \subset L_j := \bigcup_{k \in [K_j]} L_{j,k}$. By definition the $\eps$-covering number of $L_j$ (a union of $K_j$ $d$-dimensional balls) can be bounded by $\mathcal{N}(L_j,\eps) \le K_j (6/\eps)^d$ which implies $\log \mathcal{N}(L_j,\eps) \le dj \log ( 12C_\C/\eps )$ by GMRA property \eqref{second}. By Dudley's inequality (see, e.g., \replaced[id=josa]{\cite{Dudley2006} }{Chapter 8 of \cite{foucart2013mathematical}}) we conclude \added[id=josa]{via Jensen's inequality} that
    \begin{align*}
    w(\M_j) &\le w(L_j) \leq C_\text{Dudley} \int_0^2 \sqrt{\log \mathcal{N}(L_j,\eps)} ~\dd\eps \leq C_\text{Dudley} \sqrt{dj} \int_0^2 \sqrt{\log (12C_\C) - \log(\eps)} ~\dd\eps\\
    &\leq C_\text{Dudley} \sqrt{dj} \sqrt{ 2 \log (12C_\C) - \int_0^2 \log(\eps) ~\dd\eps} \\
    &\leq C' \sqrt{dj}
    \end{align*}
    where $C'$ is a constant depending on $C_\text{Dudley}$ and $C_\C$. Choosing $C = 2C' + 3$ yields the claim as $3\sqrt{2/ \pi} \le 3 \sqrt{dj}$.\\

\end{proof}

The following two lemmas concerning width bounds for fine scales will also be useful.  Their proofs (see Appendix \ref{ProofWidthBounds}), though more technical, use similar ideas to the proof of Lemma~\ref{BoundOfGaussianWidth}.  The first lemma improves on Lemma~\ref{BoundOfGaussianWidth} for large values of $j$ by considering a more geometrically precise approximation to $\M$, $\Mr_j \subset \M_j$.

\begin{lemma}[A Bound of the Gaussian Width for Fine Scales] \label{BoundOfGaussianWidthFine}
If $j \geq \log_2(D)$, $\max \{ 1, \sup_{\z\in\M}C_\z \} =: C_{\M}<\infty$, and $\Mr_j := \{\P_{j,k_j(\z)}(\z)\colon \z\in \M \} \cap B(\0,2)$ we obtain
	\begin{align*}
		\max \{ w(\M),w(\Mr_j) \} \le w(\M\cup\Mr_j) \le 2w(\M) + 2w(\Mr_j) + 3 \le C (w(\M) + 1)\log(D).
	\end{align*}
\end{lemma}

It is not surprising that for general $\M \in \St^{D-1}$ the width bound for $w(\M_j)$ (resp. $w(\Mr_j)$) depends on either $j$ or $\log(D)$. When using the proximity of $\Mr_j$ to $\M$ in Lemma \ref{BoundOfGaussianWidthFine} we only use the information that $\Mr_j \subset \tube_{C_{\M} 2^{-2j}}$ and a large ambient dimension $D$ will lead to a higher complexity of the tube. In the case of Lemma \ref{BoundOfGaussianWidth} we omit the proximity argument by using the maximal number of affine $d$-dimensional spaces in $\M_j$ and hence do not depend on $D$ but on the refinement level $j$. 

The next lemma just below utilizes even more geometric structure by assuming that $\M$ is a Riemannian Manifold.  It improves on both Lemma~\ref{BoundOfGaussianWidth} and~\ref{BoundOfGaussianWidthFine} for such $\M$ by yielding a width bound which is independent of both $j$ and $D$ for all $j$ sufficiently large.

\begin{lemma}[A Bound of the Gaussian Width for Approximations to Riemannian Manifolds] \label{RiemannianGWidthBound}
Let $\M \subset \St^{D-1}$ be a compact $d$-dimensional Riemannian manifold with $d$-dimensional volume $\Vol(\M)$ where $d \ge 1$.  Furthermore, suppose that for $\max \{ 1, \sup_{\z\in\M}C_\z \} =: C_{\M}$, $j > \max \{ j_0, \log_2(8 C_{\M} / C_1) \}$, and set $\Mr_j := \{\P_{j,k_j(\z)}(\z)\colon \z\in \M \} \cap B(\0,2)$.  Then, there exist absolute constants $C,c > 0$ such that
$$\max \{ w(\M),w(\Mr_j) \} \le w(\M\cup\Mr_j) \le C \sqrt{ d\left(1+\log\left( c \frac{\sqrt{d}}{\reach(M)}\right)\right) + \log(\max\{1,\Vol(\M)\})}.$$
Here the constants $C_z$ and $C_1$ are from properties \eqref{3b} and \eqref{tube}, respectively.
\end{lemma}


Finally, the following lemma quantifies the equivalence between Euclidean and normalized geodesic distance on the sphere.

\begin{lemma}\label{Lem:NormalizedGeo}
	For $\z, \z'\in \St^{D-1}$ one has
	\begin{equation*}
	d_G( \z, \z')\leq\| \z- \z'\|_2\leq \pi d_G( \z, \z').
	\end{equation*}
\end{lemma}
\begin{proof}
	First observe that $\langle \z, \z'\rangle =\cos\measuredangle(\z,\z')=\cos(\pi d_G(\z,\z'))$. This yields
	\begin{equation*}
	\| \z- \z'\|_2-d_G( \z, \z')=\sqrt{2-2\cos(\pi d_G(\z,\z'))}-d_G(\z,\z')\geq 0
	\end{equation*}
	as the function $f(x)=\sqrt{2-2\cos(\pi x)}-x$ is non-negative on $\left[0,1\right]$.\\
	For the upper bound note the relation between the geodesic distance $\tilde d_G$ and the normalized geodesic distance $d_G$
	\begin{equation*}
	    \tilde d_G(\z,\z')=\pi d_G(\z,\z')
	\end{equation*}
	which yields
	\begin{equation*}
	    \| \z- \z'\|_2\leq\tilde d_G(\z,\z')=\pi d_G(\z,\z').
	\end{equation*}
\end{proof}

We now have the preliminary results necessary in order to prove Theorem \ref{GeneralApproximation}.

\subsection{Proof of Theorem \ref{GeneralApproximation} with Axiomatic GMRA} \label{ProofAxiomatic}

\paragraph{} Recall that our theoretical result concerns \nameref{algorithm} with recovery performed using \eqref{MinC} and \eqref{MinX}.  The proof is based on following idea. We first control the error $\| \cj_{j,k'} - \x \|_2$ made by \eqref{MinC} in approximating a GMRA center closest to $\x$. To do so we make use of Plan and Vershynin's result on $\delta$-uniform tessellations in \cite{Plan2014}. Recall the equivalence between one-bit measurements and random hyperplanes.

\begin{definition}[Uniform tessellation, {\cite[Definition 1.1]{Plan2014}}]
	Let $K \subset \St^{D-1}$ and an arrangement of $m$ hyperplanes in $\R^D$ be given via a matrix $A$ (i.e., the $j$-th row of $A$ is the normal to the $j$-th hyperplane). Let $d_A(\x,\y) \in [0,1]$ denote the fraction of hyperplanes separating $\x$ and $\y$ in $K$ and let $d_G$ be the normalized geodesic distance on the sphere, i.e.\ opposite poles have distance one. Given $\delta > 0$, the hyperplanes provide a $\delta$-uniform tessellation of $K$ if
	\begin{align*}
	| d_A(\x,\y) - d_G(\x,\y)| \le \delta
	\end{align*}
	holds for all $\x, \y \in K$.
\end{definition}

\begin{theorem}[Random Uniform Tessellation, {\cite[Theorem 3.1]{Plan2014}}] \label{RUT}
	Consider a subset $K\subseteq \St^{D-1}$ and let $\delta>0$. Let
	\begin{equation*}
	m\geq \bar{C}\delta^{-6}\max \{w(K)^2, 2/\pi \}
	\end{equation*}
	and consider an arrangement of $m$ independent random hyperplanes in $\R^D$ uniformly distributed according to the Haar measure. Then with probability at least $1-2\exp(-c\delta^2m)$, these hyperplanes provide a $\delta$-uniform tessellation of $K$. Here and later $\bar C,c$ denote positive absolute constants.
\end{theorem}
\begin{remark}
	In words Theorem \ref{RUT} states that if the number of one-bit measurements scale at least linearly in intrinsic dimension of a set $K \subset \St^{D-1}$ then with high probability the percentage of different measurements of two points $x,y \in K$ is closely related to their distance on the sphere. Implicitly the diameter of all tessellation cells is bounded by $\delta$.
	
	The original version of Theorem \ref{RUT} uses $\gamma(K)$ instead of $w(K)$. However, note that by \eqref{WidthInequality} we get for $K\subseteq \St^{D-1}$ that $\gamma (K) \leq w(K-K)+\sqrt{2/\pi} \leq 3w(K)$ as long as the $w(K) \ge \sqrt{2/\pi}$ which is reasonable to assume. Hence, if $\bar{C}$ is changed by a factor of $9$, Theorem \ref{RUT} can be stated as above.
\end{remark}

\paragraph{} Using these results we will show in Lemma \ref{cjk'Bound} that the center $\cj_{j,k'}$ identified in step \ref{I} of the algorithm \nameref{algorithm} satisfies $\| \x - \cj_{j,k'} \|_2 \le 16 \max \{ \| \x - \cj_{j,k_j(\x)} \|_2, C_1 2^{-j-1} \}$ in Lemma \ref{cjk'Bound}. Therefore, the GMRA property \eqref{3b} provides an upper bound on $\| \x-\P_{j,k'}(\x) \|_2$. What remains is to then bound the gap between $\P_{j,k'}(\x)$ and the approximation $\x^\ast$. This happens in two steps. First, Plan and Vershynin's result on noisy one-bit sensing (see Theorem \ref{NoisyOneBit}) is applied to a scaled version of \eqref{MinX} bounding the distance between $\P_{j,k'}(\x)$ and $\bar{\x}$ (the minimizer of the scaled version). This argument works by interpreting the true measurements $\y$ as a noisy version of the non-accessible one-bit measurements of $\P_{j,k'}(\x)$. The rescaling becomes necessary as Theorem \ref{NoisyOneBit} is restricted to the unit ball in Euclidean norm. Lastly, a geometric argument is used to bound the distance between the minimum points $\bar{\x}$ and $\x^\ast$ in order to conclude the proof.
\begin{theorem}[Noisy One-Bit, {\cite[Theorem 1.3]{Plan2013}}] \label{NoisyOneBit}
	Let $\av_1,...,\av_m$ be i.i.d\ standard Gaussian random vectors in $\R^D$ and let $K$ be a subset of the Euclidean unit ball in $\R^D$. Let $\delta > 0$ and suppose that
	\begin{align*}
	m \ge C'\delta^{-6} w(K)^2.
	\end{align*}
	Then with probability at least $1-8\exp (-c\delta^2 m)$, the following event occurs. Consider a signal $\tilde{\x} \in K$ satisfying $\| \tilde{\x} \|_2 = 1$ and its (unknown) uncorrupted one-bit measurements $\tilde{\y} = (\tilde{y}_1,\dots,\tilde{y}_m)$ given as
	\begin{align*}
	\tilde{y}_i = \sign(\langle \av_i,\tilde{\x} \rangle ), \tab i = 1,2,...,m.
	\end{align*}
	Let $\y=(y_1,...,y_m) \in \{ -1,1 \}^m$ be any (corrupted) measurements satisfying $d_H(\tilde{\y},\y) \le \tau m$. Then the solution $\bar{\x}$ to the optimization problem
	\begin{align*}
	\bar{\x} = \argmax_\z \sum_{i = 1}^{m} y_i \langle \av_i,\z \rangle \tab \text{ subject to } \z \in K
	\end{align*}
	with input $\y$ satisfies
	\begin{align*}
	\| \bar{\x} - \tilde{\x} \|_2^2 \le \delta \sqrt{\log\left( \frac{e}{\delta} \right)} + 11\tau \sqrt{\log \left( \frac{e}{\tau} \right)}.
	\end{align*}
\end{theorem}

\begin{remark}
    Theorem \ref{NoisyOneBit} yields guaranteed recovery of unknown signals $x \in K \subset \B(0,1)$ up to a certain error by the formulation we use in \eqref{MinX} from one-bit measurements if the number of measurements scales linearly with the intrinsic dimension of $K$. The recovery is robust to noise on the measurements. Note that the original version of Theorem \ref{NoisyOneBit} uses $w(K-K)$ instead of $w(K)$. As $w(K-K)\leq 2w(K)$ by \eqref{WidthInequality} the result stated above also holds for a slightly modified constant $C'$.
\end{remark}

\paragraph{} We begin by proving Lemma \ref{cjk'Bound}.

\begin{lemma}\label{cjk'Bound}
	If $m \ge \bar{C} C_1^{-6} 2^{6(j+1)} \max \{ w(\M \cup \P_\St (\C_j))^2, 2/ \pi \}$ the center $\cj_{j,k'}$ chosen in step \ref{I} of Algorithm \nameref{algorithm} fulfills
	\begin{align*}
	\| \x - \cj_{j,k'} \|_2 \le 16 \max \{ \| \x - \cj_{j,k_j(\x)} \|_2, C_1 2^{-j-1} \}.
	\end{align*}
	for all $\x \in \M \subset \St^{D-1}$ with probability at least $1-2\exp(-c (C_1 2^{-j-1})^2 m)$.
\end{lemma}

\begin{proof}
	By definition of $\cj_{j,k'}$ in \eqref{MinC} we have that
	\begin{align*}
	d_H ( \sign(A \cj_{j,k'}),\y ) \le d_H ( \sign(A\cj_{j,k_j(\x)}),\y ).
	\end{align*}
	As, for all $\z,\z' \in \R^D$, $d_H(\sign(A\z),\sign(A\z')) = m \cdot d_A(\z,\z') = m \cdot d_A(\P_\St(\z),\P_\St(\z'))$, this is equivalent to
	\begin{align*}
	d_A ( \P_\St (\cj_{j,k'}),\x ) \le d_A ( \P_\St (\cj_{j,k_j(\x)}),\x ).
	\end{align*}
	Noting that Gaussian random vectors and Haar random vectors yield identically distributed hyperplanes, 
	Theorem \ref{RUT} now transfers this bound to the normalized geodesic distance, namely
	\begin{align*}
	d_G ( \P_\St (\cj_{j,k'}),\x ) \le d_G ( \P_\St (\cj_{j,k_j(\x)}),\x ) + 2\delta
	\end{align*}
	with probability at least $1-2\exp(-c\delta^2m)$ where $\delta = C_1 2^{-j-1}$. Observe $d_G(\z,\z') \le \| \z - \z' \|_2 \le \pi d_G(\z,\z')$ for all $\z,\z' \in \St^{D-1}$ (recall Lemma \ref{Lem:NormalizedGeo}) which leads to
	\begin{align*}
	\| \P_\St (\cj_{j,k'}) - \x \|_2 &\le \pi d_G( \P_\St ( \cj_{j,k_j(\x)}),\x) + 2\pi \delta \\
	&\le \pi \| \P_\St (\cj_{j,k_j(\x)}) - \x \|_2 + 2\pi \delta.
	\end{align*}
	As by property \eqref{tube} the centers are close to the manifold, they are also close to the sphere and we have $\| \P_\St (\cj_{j,k}) - \cj_{j,k} \|_2 < C_1 2^{-j-2}$, for all $\cj_{j,k} \in \C_j$. Hence, we conclude
	\begin{align*}
	\| \cj_{j,k'} - \x \|_2 &\le\|\cj_{j,k'}-\P_\St(\cj_{j,k'})\|_2+\|\P_\St(\cj_{j,k'})-x\|_2\\
	&\le \pi ( \| \cj_{j,k_j(\x)} - \x \|_2 + C_1 2^{-j-2} ) + 2\pi \delta + C_1 2^{-j-2} \\
	&\le \left( \pi + \frac{\pi}{2} + 2\pi + \frac{1}{2} \right) \max \{ \| \cj_{j,k_j(\x)} - \x \|_2, C_1 2^{-j-1} \} \\
	&\le 16 \max \{ \| \cj_{j,k_j(\x)} - \x \|_2, C_1 2^{-j-1} \}.
	\end{align*}
\end{proof}

\paragraph{} We can now prove a detailed version of Theorem \ref{GeneralApproximation} for the given axiomatic GMRA and deduce Theorem \ref{GeneralApproximation} as a corollary.

\begin{theorem}[Uniform Recovery - Axiomatic Case] \label{ApproximationTheorem}

Let $\M \subset \St^{D-1}$ be given by its GMRA for some levels $j_0 < j \le J$, such that $C_1 < 2^{j_0 + 1}$ 
	where $C_1$ is the constant from GMRA properties \eqref{GMRA2b} and \eqref{tube}. 
	Fix $j$ and assume that $\dist (\0,\M_j) \ge 1/2$.  Further, let $d \geq 1$ and
	\begin{align} \label{mBoundUniformRecoveryAxiomaticCase}
		m \ge 16\max\{C',\bar{C}\} C_1^{-6} 2^{6(j+1)} (w(\M) + C\sqrt{dj})^2,
	\end{align}
	where $C'$ is the constant from Theorem \ref{NoisyOneBit}, $\bar{C}$ from Theorem \ref{RUT}, and $C > 3$ from Lemma \ref{BoundOfGaussianWidth}. Then, with probability at least $1 - 12\exp(-c (C_1 2^{-j-1})^2 m)$ the following holds for all $\x \in \M$ with one-bit measurements $\y = \sign(A\x)$ and GMRA constants $\tilde{C}_\x$ from property \eqref{3b} satisfying $\tilde{C}_\x < 2^{j - 1}$:  The approximations $\x^\ast$ obtained by \nameref{algorithm2} fulfill
	\begin{align*}
		\| \x - \x^\ast \|_2 \le   
		\left( 2\tilde{C}_\x 2^{-\frac{j}{2}} + \sqrt{\frac{C_1}{2}} \sqrt[4]{\log \left( \frac{4e}{\min \{C_1,1\}} \right)} + \sqrt{11 C'_\x} \sqrt[4]{\log \left( \frac{2e}{\min \{C'_\x, 1\}} \right)} \right) \sqrt[4]{j} 2^{-\frac{j}{2}}.
	\end{align*}
	Here $C'_\x := 2\tilde{C}_\x + C_1$.
\end{theorem}


\blue{
\begin{proof}[of Theorem \ref{GeneralApproximation}]
	As $j = \lceil \log(1/\eps) \rceil$, we know that $2^{-j} \le \eps \le 2^{-j+1}$. This implies
	\begin{align*}
		m 
		\ge E C_1^{-6} \eps^{-6} \max\left\{w(\M), \sqrt{d \left( \log \left( \frac{1}{\eps}\right) + 1 \right) } \right\}^2 
		\ge 16\max\{C',\bar{C}\} C_1^{-6} 2^{6(j+1)} (w(\M) + C\sqrt{dj})^2
	\end{align*}
	for $E > 0$ chosen appropriately. The result follows by applying Theorem \ref{ApproximationTheorem}.
\end{proof}
}

\begin{proof}[of Theorem \ref{ApproximationTheorem}]
	Recall that $k'$ is the index chosen by \nameref{algorithm2} in \eqref{MinCnew}. The proof consists of three steps. First, we apply Lemma \ref{cjk'Bound} in \textbf{(I)}. By the GMRA axioms this supplies an estimate for $\| \x - \P_{j,k'}(\x) \|_2$ with high probability. In \textbf{(II)} we use Theorem \ref{NoisyOneBit} to bound the distance between $\P_{j,k'}(\x)/\| \P_{j,k'}(\x) \|_2$ and the minimizer $\x^*$ given by
	\begin{align} \label{xbar}
	x^* = \argmin_\z \sum_{l = 1}^m (-y_l) \langle \av_l,\z \rangle, \quad \text{subject to } \z \in K := \conv ( \P_\St ( P_{j,k'} \cap \B(\0,2) ))
	\end{align}
	with high probability. By a union bound over all events Part \textbf{(III)} then concludes with an estimate of the distance $\| \x - \x^\ast \|_2$ combining \textbf{(I)} and \textbf{(II)}.\\
	
	\blue{\textbf{(I)}} Set $\delta := C_1 2^{-j-1}$. Observing that $C_1 2^{-j-2} < 1/2$ by assumption, GMRA property \eqref{tube} yields that all centers in $\C_j$ are closer to $\St^{D-1}$ than $1/2$, i.e., $1/2 \le \| \cj_{j,k} \|_2 \le 3/2$. Hence, by \eqref{WidthInequality}
	\begin{align} \label{ProjectedMeanWidth}
	0 \le w(\P_\St(\C_j)) \le \gamma(\P_\St(\C_j)) \le 2\gamma(\C_j) \le 4w(\C_j) + 2\sqrt{\frac{2}{\pi}} \dist(\0, \C_j) \le 4w(\C_j) + 4.
	\end{align}
	As $\C_j \subset \M_j$ we know by Lemma \ref{BoundOfGaussianWidth},  \eqref{ProjectedMeanWidth}, and Remark \ref{rem:BoundOfGaussianWidth} that
	\begin{align} \label{eq:mbound}
	\begin{split}
	m &\ge 4 \bar{C} \delta^{-6} (2w(\M) + 2C\sqrt{dj})^2 \ge 4\bar{C} \delta^{-6} (2w(\M) + 4w(\C_j) + 6 )^2 \\
	&\ge 4\bar{C} \delta^{-6} (2w(\M) + w(\P_\St (\C_j)) + 2 )^2 = \bar{C} \delta^{-6} (4w(\M) + 2w(\P_\St (\C_j)) + 4 )^2 \\
	&\ge \bar{C} \delta^{-6} (w(\M \cup \P_\St (\C_j)) + 1 )^2 \ge \bar{C} \delta^{-6} \max\{w(\M \cup \P_\St(\C_j))^2,2/\pi\}.
	\end{split}
	\end{align} 
	Hence, Lemma \ref{cjk'Bound} implies that
	\begin{align*}
	\| \x - \cj_{j,k'} \|_2 \le 16 \max \{ \| \x-\cj_{j,k_j(\x)} \|_2, C_1 2^{-j-1} \}.
	\end{align*}
	with probability at least $1-2\exp(-c\delta^2 m)$. By GMRA property \eqref{3b} we now get that
	\begin{align} \label{Approx}
	\| \x - \P_{j,k'}(\x) \|_2 \le \tilde{C}_\x 2^{-j}
	\end{align}
	for some constant $\tilde{C}_\x$.\\
	
	\blue{\textbf{(II)}} Define $\alpha := \| \P_{j,k'}(\x) \|_2$ and note that one has $1/2 \le \alpha \le 3/2$ as $\x \in \St^{D-1}$ and $\| \x - \P_{j,k'}(\x) \|_2 \le \tilde{C}_\x 2^{-j} \le 1/2$ by \eqref{Approx} and assumption. We now create the setting of Theorem \ref{NoisyOneBit}.
	Define $\tilde{\x} := \P_{j,k'}(\x)/\alpha \in \St^{D-1}$, $\tilde{\y} := \sign(A\tilde{\x}) = \sign(A\P_{j,k'}(\x))$, $K = \conv ( \P_\St ( P_{j,k'} \cap \B(\0,2) ))$, and $\tau := (2\tilde{C}_\x + C_1)2^{-j}$. If successfully applied with these quantities Theorem \ref{NoisyOneBit} will bound $\| \tilde{\x} - \x^* \|_2$ by
	\begin{align}
	\| \tilde{\x} - \x^* \|_2 &\leq \sqrt{\delta \sqrt{\log\left( \frac{e}{\delta} \right)} + 11\tau \sqrt{\log \left( \frac{e}{\tau} \right)}} \label{equ:PhaseIIbound}\\
	& \leq \left( \sqrt{\frac{C_1}{2}} \sqrt[4]{\log \left( \frac{4e}{\min \{C_1,1\}} \right)} + \sqrt{11 (2\tilde{C}_x + C_1)} \sqrt[4]{\log \left( \frac{2e}{\min \{(2\tilde{C}_x + C_1), 1\}} \right)} \right) \sqrt[4]{j} 2^{-\frac{j}{2}}. \nonumber
	\end{align}	
	All that remains is to verify that the conditions of Theorem \ref{NoisyOneBit} are met so that \eqref{equ:PhaseIIbound} is guaranteed with high probability.
	
	We first have to check $d_H(\tilde{\y},\y) \le \tau m$. Recall that $\frac 1 \alpha\leq 2$ and for $\alpha>0$ one has $\alpha w(K)=w(\alpha K)$. Applying Lemma \ref{BoundOfGaussianWidth} and \eqref{WidthInequality} we have, in analogy to \eqref{eq:mbound}, that
	\begin{align*}
	m &\ge \bar{C} \delta^{-6} (4w(\M) + 4w(\M_j) + 12)^2 \ge \bar{C} \delta^{-6} \left( 2w(\M) + 2w\left( \frac{\M_j}{\alpha} \right) + 12 \right)^2 \\ 
	&\ge \bar{C} \delta^{-6} \left( w \left(\M \cup \frac{\M_j}{\alpha} \right) + 7 \right)^2 \ge \bar{C} \delta^{-6} \left( w \left( \left( \M \cup \frac{\M_j}{\alpha} \right) \cap \B(\0,1) \right) + 7 \right)^2.
	\end{align*}
	Note that in the third inequality a slight modification of the second inequality in Lemma \ref{BoundOfGaussianWidth} is used. As $\M_j/\alpha \subset \B(0,4)$ one has $w(\M \cup \M_j/\alpha) \le 2w(\M) + 2w(\M_j/\alpha) + 5$ by adapting \eqref{equ:GenWidthUnionBound}. We can now use Theorem \ref{RUT}, Lemma \ref{Lem:NormalizedGeo}, and the fact that $|1 - \alpha| = |\| \x \|_2 - \| \P_{j,k'}(\x) \|_2 | \le \| \x - \P_{j,k'}(\x) \|_2$ to obtain
	\begin{align*}
	\frac{d_H(\tilde{\y},\y)}{m} &= d_A(\tilde{\x},\x) \le d_G(\tilde{\x},\x) + \delta \le \| \tilde{\x} - \x \|_2 + \delta \le \| \tilde{\x} - \P_{j,k'}(\x) \|_2 + \| \P_{j,k'}(\x) - \x \|_2 + \delta \\
	&= | 1 - \alpha | + \| \P_{j,k'}(\x) - \x \|_2 + \delta \le 2 \| \P_{j,k'}(\x) - \x \|_2 + \delta \\
	&\le (2\tilde{C}_\x + C_1) 2^{-j} = \tau
	\end{align*}
	with probability at least $1-2\exp(-c \delta^2 m)$. Furthermore, by a similar argumentation as in \eqref{ProjectedMeanWidth} one gets
	\begin{align} \label{eq:Kwidth}
	w(K) = w(\P_\St ( P_{j,k'} \cap \B(\0,2) )) \le 4w(\M_j) + 4
	\end{align}
	where one uses invariance of the Gaussian width under taking the convex hull (see \cite[Proposition 2.1]{Plan2013}), the fact that $P_{j,k'} \cap \B(\0,2) \subset \M_j$, and the assumption that $1/2 \le \dist(\M_j, \0) \le 2$.  In combination with Lemma \ref{BoundOfGaussianWidth} we have, in analogy to \eqref{eq:mbound}, that
	\begin{align*}
	m \ge 4C'\delta^{-6} (2w(\M) + 4w(\M_j) + 6 )^2 \ge 4C'\delta^{-6} (w(K) + 2)^2 \ge C' \delta^{-6} w(K)^2.
	\end{align*} 
	Hence, we can apply Theorem \ref{NoisyOneBit} to obtain with probability at least $1-8\exp(-c \delta^2 m)$ that
	\begin{equation*} 
	\| \bar{\x} - \tilde{\x} \|_2^2 \le \delta \sqrt{\log \left( \frac{e}{\delta} \right)} + 11 \tau \sqrt{\log \left( \frac{e}{\tau} \right)},
	\end{equation*}
	the estimate \eqref{equ:PhaseIIbound} now follows.\\
	
	\blue{\textbf{(III)}} To conclude the proof we apply a union bound and obtain with probability at least $1 - 12\exp(-c \delta^2 m)$ that
	\begin{align*}
	\| \x - \x^\ast \|_2 &\le \| \x - \P_{j,k'}(\x) \|_2 + \| \P_{j,k'}(\x) - \tilde{\x} \|_2 + \| \tilde{\x} - \x^* \|_2 \\
	&= \| \x - \P_{j,k'}(\x) \|_2 + |1-\alpha| + \| \tilde{\x} - \x^* \|_2 \\
	&\le  2\| \x - \P_{j,k'}(\x) \|_2 + \| \tilde{\x} - \x^* \|_2.
	\end{align*}
	GMRA property \eqref{3b} combined with \eqref{equ:PhaseIIbound} now yields the final desired error bound.
	
\end{proof}

\blue{
	\begin{remark}
		For obtaining the lower bounds on $m$ in \eqref{mBoundUniformRecoveryAxiomaticCase} and \eqref{mBoundUniformRecoveryCorollary} we made use of Lemma \ref{BoundOfGaussianWidth} leading to the influence of $j$ which is suboptimal for fine scales (i.e., $j$ large). To improve on this for large $j$ one can exploit the alternative versions of the lemma, namely, Lemma \ref{BoundOfGaussianWidthFine} and Lemma \ref{RiemannianGWidthBound}. Then, however, some minor modifications become necessary in the proof of Theorem \ref{ApproximationTheorem} as the lemmas only apply to $\Mr_j$:\\ 
		In proof step \textbf{(I)}, e.g., one has to guarantee that $\C_j \subset \Mr_j$, i.e., that each center $\cj_{j,k}$ is a best approximation for some part of the manifold. This is a reasonable assumption especially if the centers are constructed as means of small manifold patches which is a common approach in empirical applications (cf. Appendix \ref{EmpiricalGMRA}).  \\
		Also, when working with $\Mr_j$ it is essential in proof step \textbf{(II)} to have a near-best approximation subspace of $\x$, i.e., the $k'$ obtained in proof step \textbf{(I)} has to fulfill $k' \approx k_j(\x)$ as $\Mr_j$ does not include many near-optimal centers for each point on $\M$.  Here, one can exploit the minimal distance of centers $\cj_{j,k}$ to each other as described in GMRA property \eqref{GMRA2b} and choose $\delta$ slightly smaller (in combination with a correspondingly strengthened upper bound in Lemma \ref{cjk'Bound}) to obtain the necessary guarantees for proof step \textbf{(I)}.  As we are principally concerned with the case where $j = \mathcal{O}({\log(D)})$ in this paper, however, we will leave such variants to future work.
	\end{remark}
}

We are now prepared to explore the numerical performance of the proposed methods.

\section{Numerical Simulation} 
\label{Numerics}

\paragraph{} \blue{In this section we present various numerical experiments to benchmark OMS. The GMRAs we work with are constructed using the GMRA code provided by Maggioni\footnote{The code is available at \url{http://www.math.jhu.edu/\~mauro/\#tab\_code}.}. We compared the performance of OMS for three exemplary choices of $\M$, namely, a simple $2$-dim sphere embedded in $\R^{20}$ ($20000$ data points sampled from the $2$-dimensional sphere $\M$ embedded in $\St^{20-1}$), the MNIST data set \cite{MNIST2017} of handwritten digits "$1$" ($3000$ data points in $\R^{784}$), and the Fashion-MNIST data set \cite{xiao2017/online} of shirt images ($2000$ data points in $\R^{784}$). Both MNIST data sets have been projected to the unit sphere before taking measurements and computing the GMRA. In each of the experiments \ref{SIMPLEvsCONVEX}-\ref{TREEvsNOTREE} we first computed a GMRA up to refinement level $j_\text{max} = 10$ and then recovered $100$ randomly chosen $\x \in \M$ from their one-bit measurements by applying \nameref{algorithm2}. Depicted is the averaged relative error between $\x$ and its approximation $\x^*$, i.e., $\| \x - \x^\ast \|_2 / \| \x \|_2$ which is equal to the absolute error $\| \x - \x^\ast \|_2$ for $\M \subset \St^{D-1}$. Note the different approximation error ranges of the sphere and the MNIST experiments when comparing both settings. As a benchmark, the average error caused by the best GMRA approximation, i.e., projection of $\x$ onto the GMRA, is provided in all plots in the form of a horizontal dashed black line\footnote{\blue{To be precise for each $\x$ we picked the GMRA subspace minimizing the projection distance $\| \x - \P_{j,k}(\x) \|_2$ and averaged this error over all realizations of $\x$.}}. Let us mention that the error caused by best GMRA approximation -- though being a benchmark for the recovery performance one could expect -- is not a strict lower bound to OMS since the algorithm is not restricted to the GMRA but uses convex relaxations of the GMRA subspaces, cf. Figure \ref{Fig:SIMPLEvsCONVEXa}.
}

\subsection{OMS-simple vs. OMS} \label{SIMPLEvsCONVEX}

\paragraph{} \blue{The first test compares recovery performance of the two algorithms presented above, namely \nameref{algorithm} for $R \in \{0.5, 1, 1.5 \}$ and \nameref{algorithm2}. The results are depicted in Figure \ref{Fig:SIMPLEvsCONVEX}. Note that only $R = 1.5$ and, in the case of the $2$-sphere, $R = 1$ are depicted as in the respective other cases for each number of measurements most of the trials did not yield a feasible solution in \eqref{MinX} so the average was not well-defined. One can observe that for all data sets \nameref{algorithm2} outperforms \nameref{algorithm} which is not surprising as \nameref{algorithm2} does not rely on a suitable parameter choice. This observation is also the reason for us to restrict the theoretical analysis to \nameref{algorithm2}. The more detailed approximation of the toy example ($2$-dimensional sphere) is due to its simpler structure and lower dimensional setting and can also be observed in \ref{PLUSvsNOPLUS}-\ref{TREEvsNOTREE}. Figures \ref{fig:Image1} and \ref{fig:Image2} depict one specific reconstructed image of Fashion-MNIST, for four different numbers of measurements $m$. Obviously,  OMS shows a better performance having less quantization artifacts. Moreover, the good visual quality of the OMS reconstruction in Figure \ref{fig:Image2} for only $m=100$ suggests that the $\ell_2$-error used in Figure \ref{Fig:SIMPLEvsCONVEX} is a rather pessimistic performance measure. Considering that all MNIST-images could be fully coded in $8\cdot 28 \cdot 28 = 6272$ bits (gray-scale images contain only $8$-bit of information per pixel) it is important to point out the potentially overly pessimistic nature of the $\ell_2$-errors reported for $m = 10000$, as well as to note that the visual quality is already much better than one might expect for OMS at compressions of ratios of size $100/6272 < 0.02$.}\\


\begin{figure}[!ht]
	\centering
	\begin{subfigure}[b]{0.48\textwidth}
		\psfragfig*[width=0.50\textwidth]{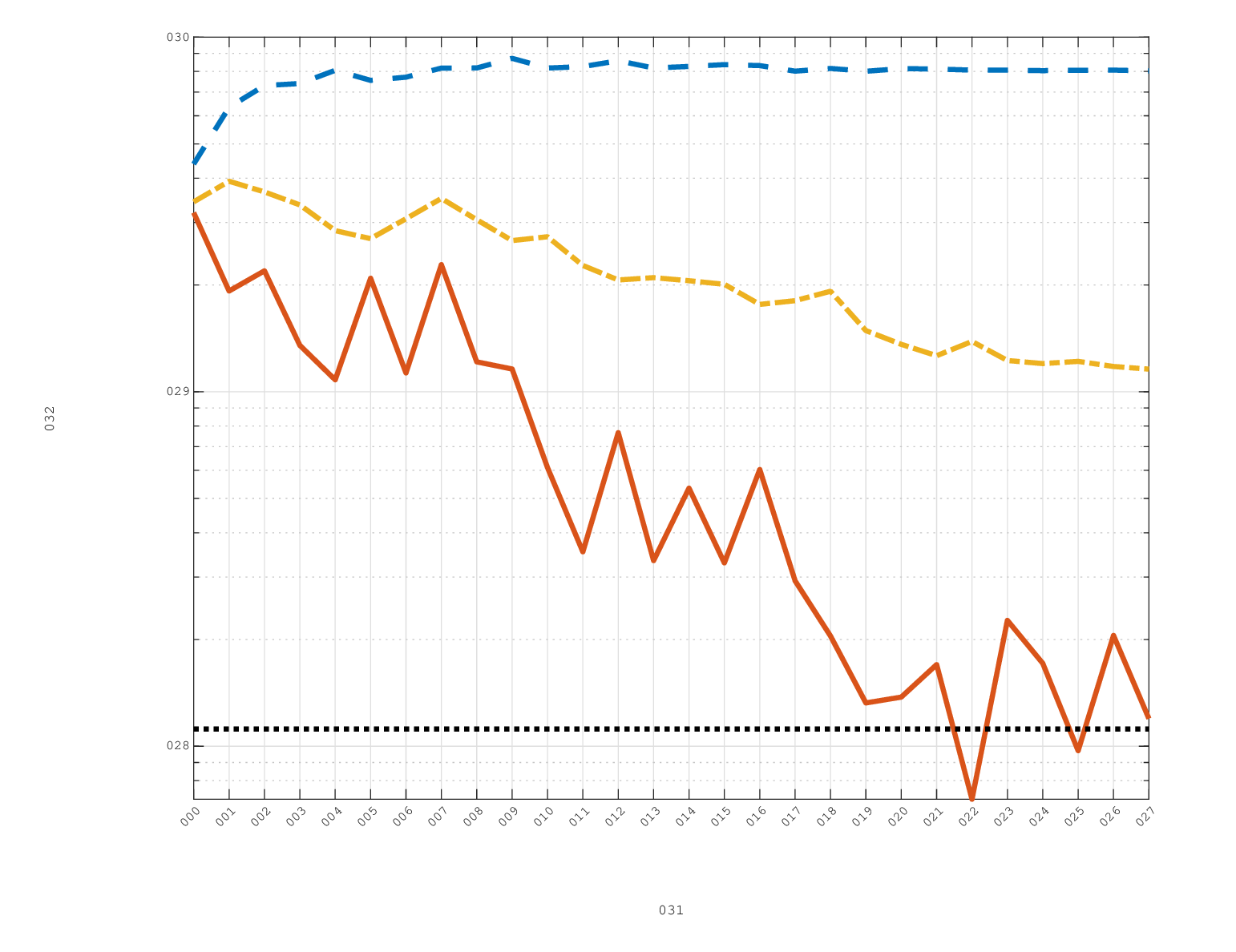}{}
		\subcaption{$2$-Sphere}
		\label{Fig:SIMPLEvsCONVEXa}
	\end{subfigure} \\
	\begin{subfigure}[b]{0.48\textwidth}
		\def\svgwidth{\linewidth}
		\psfragfig*[width=0.50\textwidth]{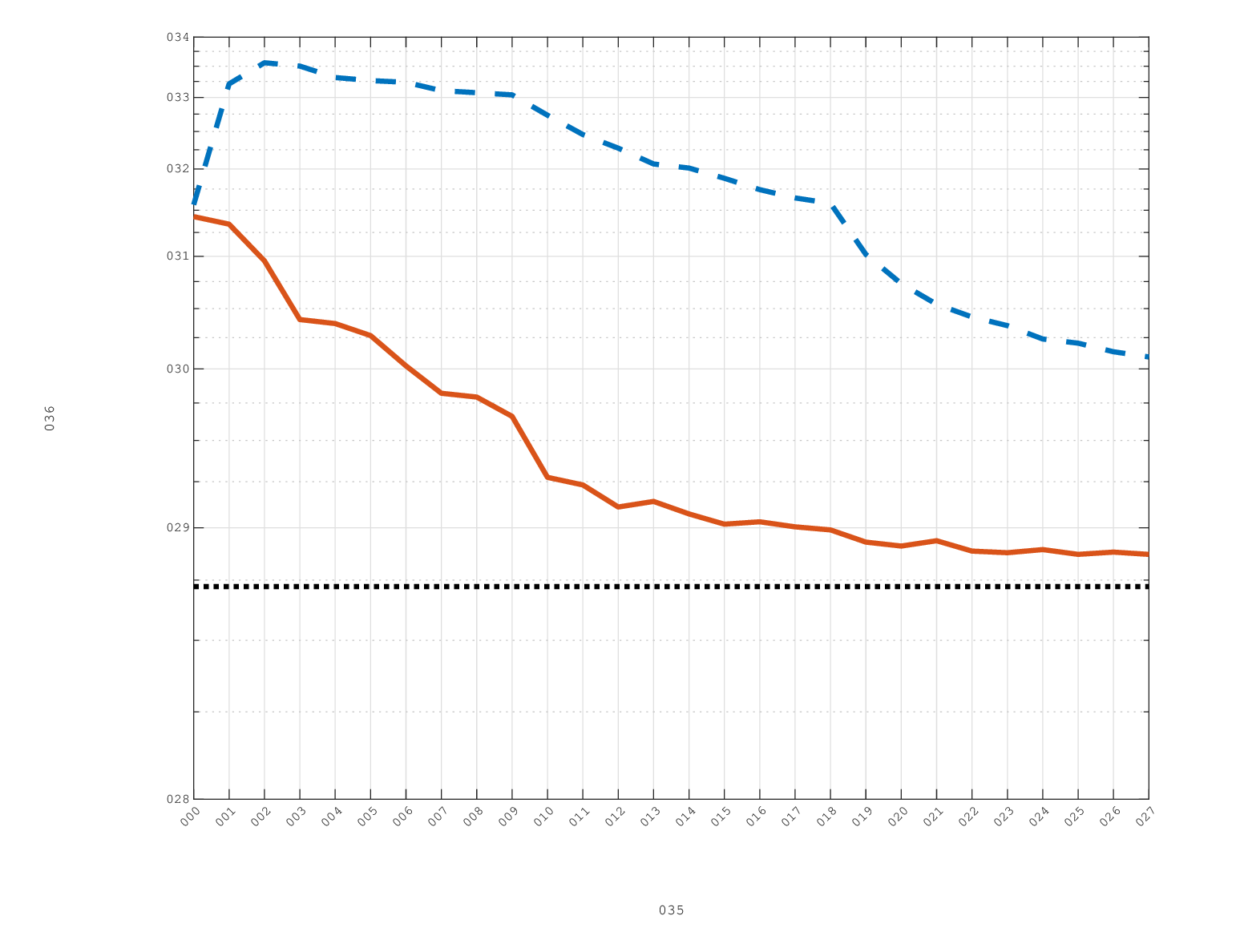}{}
		\subcaption{MNIST}
		\label{Fig:SIMPLEvsCONVEXb}
	\end{subfigure}
	\begin{subfigure}[b]{0.48\textwidth}
		\psfragfig*[width=0.50\textwidth]{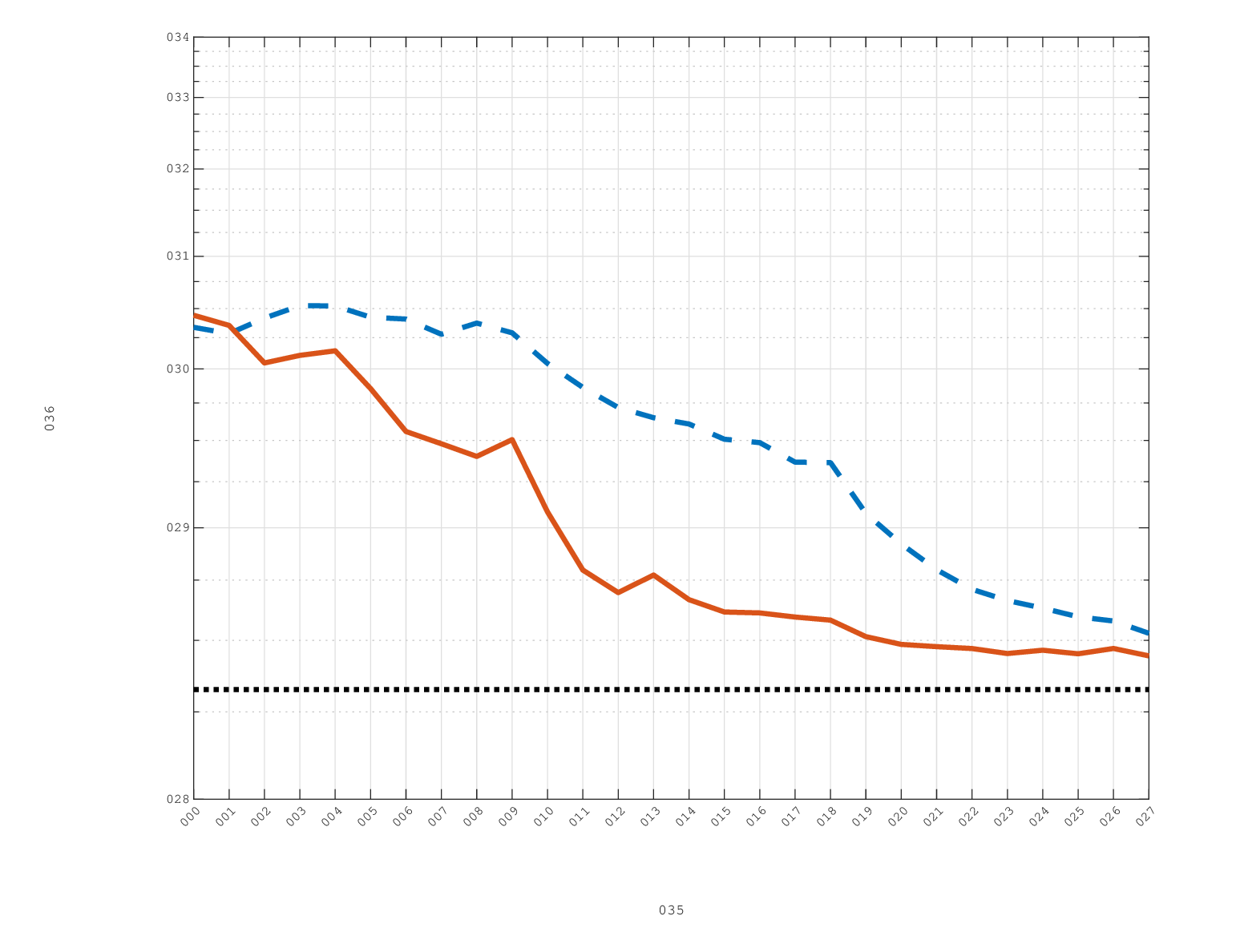}{}
		\subcaption{Fashion-MNIST}
		\label{Fig:SIMPLEvsCONVEXc}
	\end{subfigure}
	\caption{Comparison of OMS-Simple for $R = 1$ (dotted-dashed, yellow), $R = 1.5$ (dashed, blue) and OMS (solid, red). Recall from Section \ref{SIMPLEvsCONVEX} that OMS-Simple for $R = 1$ does not recover in Figures \ref{Fig:SIMPLEvsCONVEXb} and \ref{Fig:SIMPLEvsCONVEXc}. The black dotted line highlights average error caused by direct GMRA projection.} \label{Fig:SIMPLEvsCONVEX}
\end{figure}

\begin{figure}[!ht]
	\centering
	\begin{subfigure}{0.14\textwidth}
		\includegraphics[width=\linewidth]{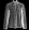}
		\caption{Ground truth.}
	\end{subfigure} \quad
    \begin{subfigure}{0.14\textwidth}
    	\includegraphics[width=\linewidth]{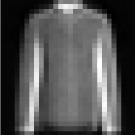}
    	\caption{GMRA.}
    \end{subfigure} \quad
    \begin{subfigure}{0.14\textwidth}
    	\includegraphics[width=\linewidth]{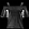}
    	\caption{$m=10$.}
    \end{subfigure} \quad
    \begin{subfigure}{0.14\textwidth}
    	\includegraphics[width=\linewidth]{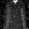}
    	\caption{$m=100$.}
    \end{subfigure} \quad
    \begin{subfigure}{0.14\textwidth}
    	\includegraphics[width=\linewidth]{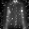}
    	\caption{$m=1000$.}
    \end{subfigure} \quad
    \begin{subfigure}{0.14\textwidth}
    	\includegraphics[width=\linewidth]{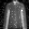}
    	\caption{$m=10000$.}
    \end{subfigure} 
    \caption{One Fashion-MNIST signal with its OMS-Simple $(R = 1.5)$ reconstructions. The best GMRA approximation is given as a benchmark. Note that the GMRA uses an $11$-dimensional subspace at this part of the manifold.}
    \label{fig:Image1}
\end{figure}

\begin{figure}[!ht]
	\centering
	\begin{subfigure}{0.14\textwidth}
		\includegraphics[width=\linewidth]{GroundTruth1.pdf}
		\caption{Ground truth.}
	\end{subfigure} \quad
	\begin{subfigure}{0.14\textwidth}
		\includegraphics[width=\linewidth]{GMRAProj1.pdf}
		\caption{GMRA.}
	\end{subfigure} \quad
	\begin{subfigure}{0.14\textwidth}
		\includegraphics[width=\linewidth]{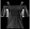}
		\caption{$m=10$.}
	\end{subfigure} \quad
	\begin{subfigure}{0.14\textwidth}
		\includegraphics[width=\linewidth]{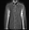}
		\caption{$m=100$.}
	\end{subfigure} \quad
	\begin{subfigure}{0.14\textwidth}
		\includegraphics[width=\linewidth]{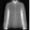}
		\caption{$m=1000$.}
	\end{subfigure} \quad
	\begin{subfigure}{0.14\textwidth}
		\includegraphics[width=\linewidth]{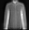}
		\caption{$m=10000$.}
	\end{subfigure} 
	\caption{One Fashion-MNIST signal with its OMS reconstructions. The best GMRA approximation is given as a benchmark. Note that the GMRA uses an $11$-dimensional subspace at this part of the manifold.}
	\label{fig:Image2}
\end{figure}


\subsection{Modifying \nameref{algorithm2}} \label{PLUSvsNOPLUS}

\paragraph{} \blue{Observations in \cite{krausesolberg2017} motivate to consider a modification of \nameref{algorithm2} in which
	 \eqref{MinXnew} is replaced by
\begin{align*}
	\x^\ast &= \argmin_{\z \in \R^D} \sum_{l=1}^{m} \left[(-y_l) \langle \av_l,\z \rangle \right]_+, \quad
	\text{subject to } \z \in \conv \left( \P_\St (P_{j,k'} \cap \B(0,2)) \right)
\end{align*} 
where $[t]_+ = \max \{0,t\}$ denotes the positive part of $t \in \R$.} 

\blue{We tested this approach in some initial numerical experiments, but found that the modification produced ambiguous results with no clear improvement. Let us mention that after the original version of this paper had been finished, the $[\cdot]_+$-formulation has been thoroughly analyzed for robust one- and multi-bit quantization in a dithered measurement setup in \cite{jung2019quantized}, but not in the context of GMRA and with dithering (this explains the ambiguous experimental outcomes in our case, as we do not use dithering in the one-bit measurements). Certainly one could transfer results from \cite{jung2019quantized} to our setting by changing the measurement model to include dithering, cf. \cite{dirksen2019}. }

\subsection{Are Two Steps Necessary?} \label{STEPS}

\paragraph{} One might wonder if the two steps in \nameref{algorithm} and \nameref{algorithm2} are necessary at all. Wouldn't it be sufficient to use the center $\cj_{j,k'}$ determined in step \ref{I} as an approximation for $\x$? If the GMRA is fine enough, this indeed is the case. If one only has access to a rather rough GMRA, the simulations in Figure \ref{Fig:TREEvsNOTREE} show that the second step makes a notable difference in approximation quality. \blue{This behavior is not surprising in view of Lemma \ref{cjk'Bound}.} The lemma guarantees a good approximation of $\x$ by $\cj_{j,k'}$ as long as $\x$ is well approximated by an optimal center. For \blue{both} MNIST data sets, one can observe that the second step only improves performance if the number of one-bit measurements is sufficiently high. For a small set of measurements the centers might yield better approximation as they lie close to $\M$ by GMRA property \eqref{tube}. On the other hand, only parts of the affine spaces are practical for approximation and a certain number of measurements is necessary to restrict \ref{II2} to the relevant parts.\\

%

\begin{figure}[!ht]
	\centering
	\begin{subfigure}[b]{0.48\textwidth}
		\centering
		\def\svgwidth{\linewidth}
		\psfragfig[width=0.50\textwidth]{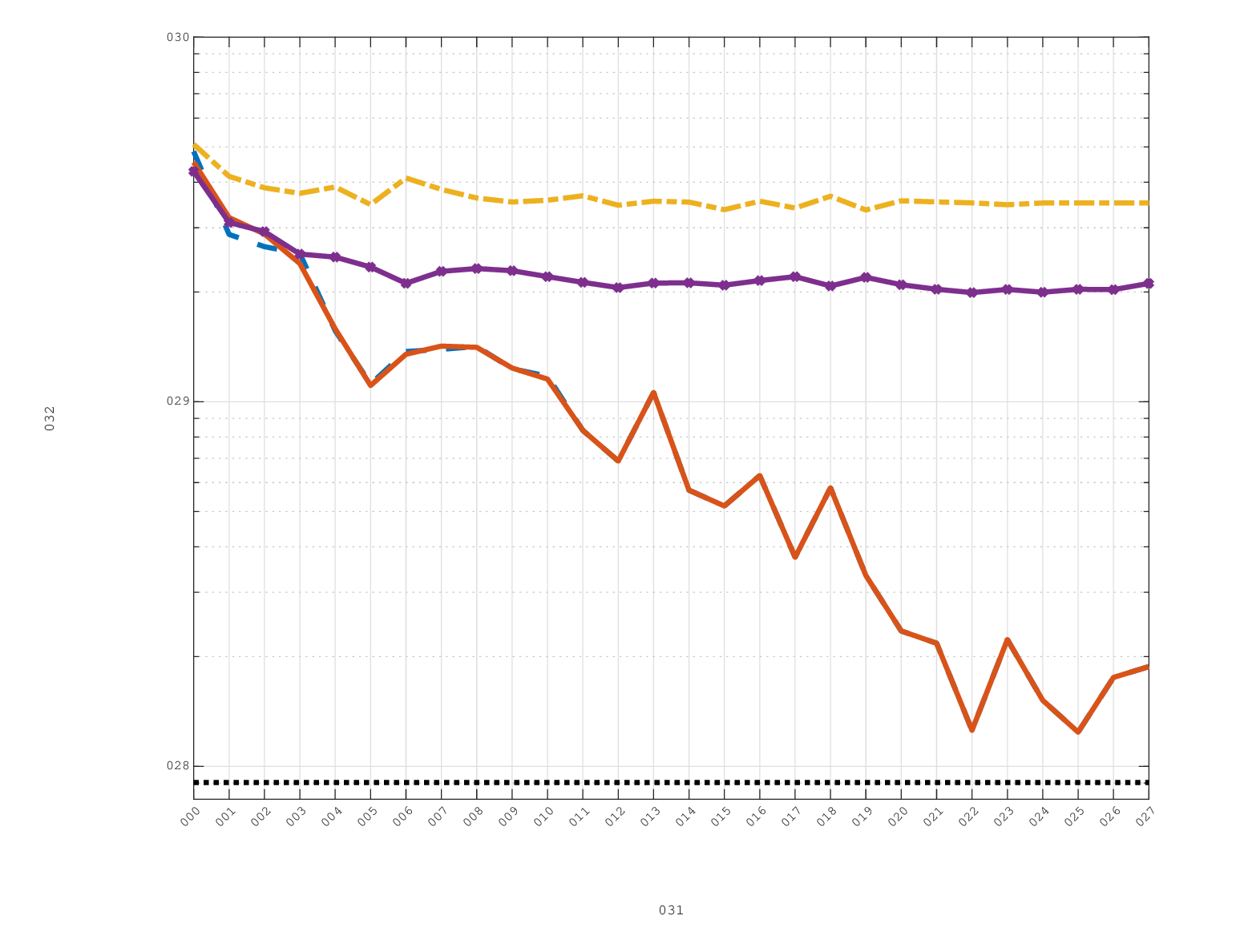}{}
		\subcaption{$2$-Sphere}
		\label{Fig:TREEvsNOTREEa}
	\end{subfigure}
	
	\begin{subfigure}[b]{0.48\textwidth}
		\centering
		\def\svgwidth{\linewidth}
		\psfragfig[width=0.50\textwidth]{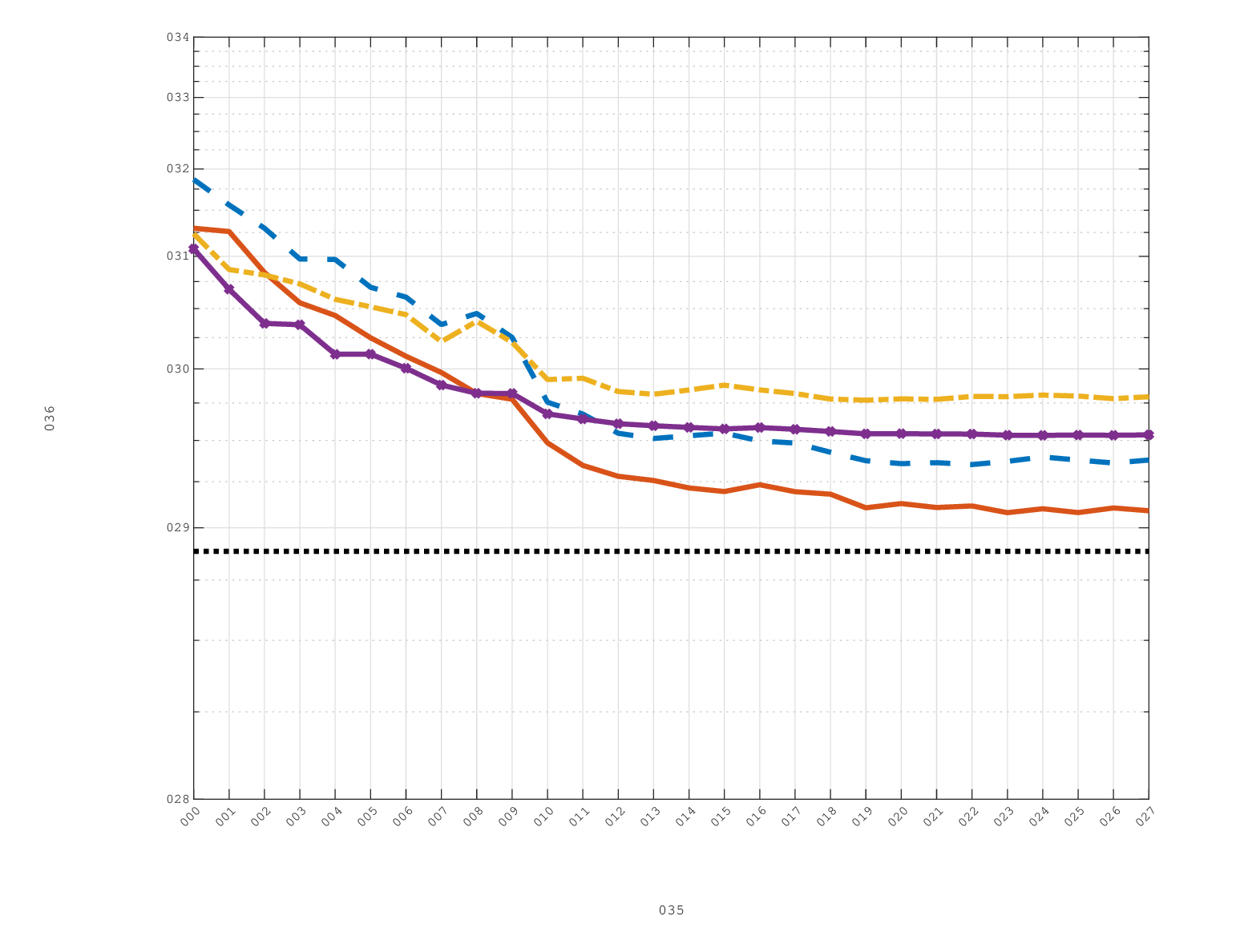}{}
		\subcaption{MNIST}
	\end{subfigure}
	\begin{subfigure}[b]{0.48\textwidth}
		\centering
		\def\svgwidth{\linewidth}
		\psfragfig[width=0.50\textwidth]{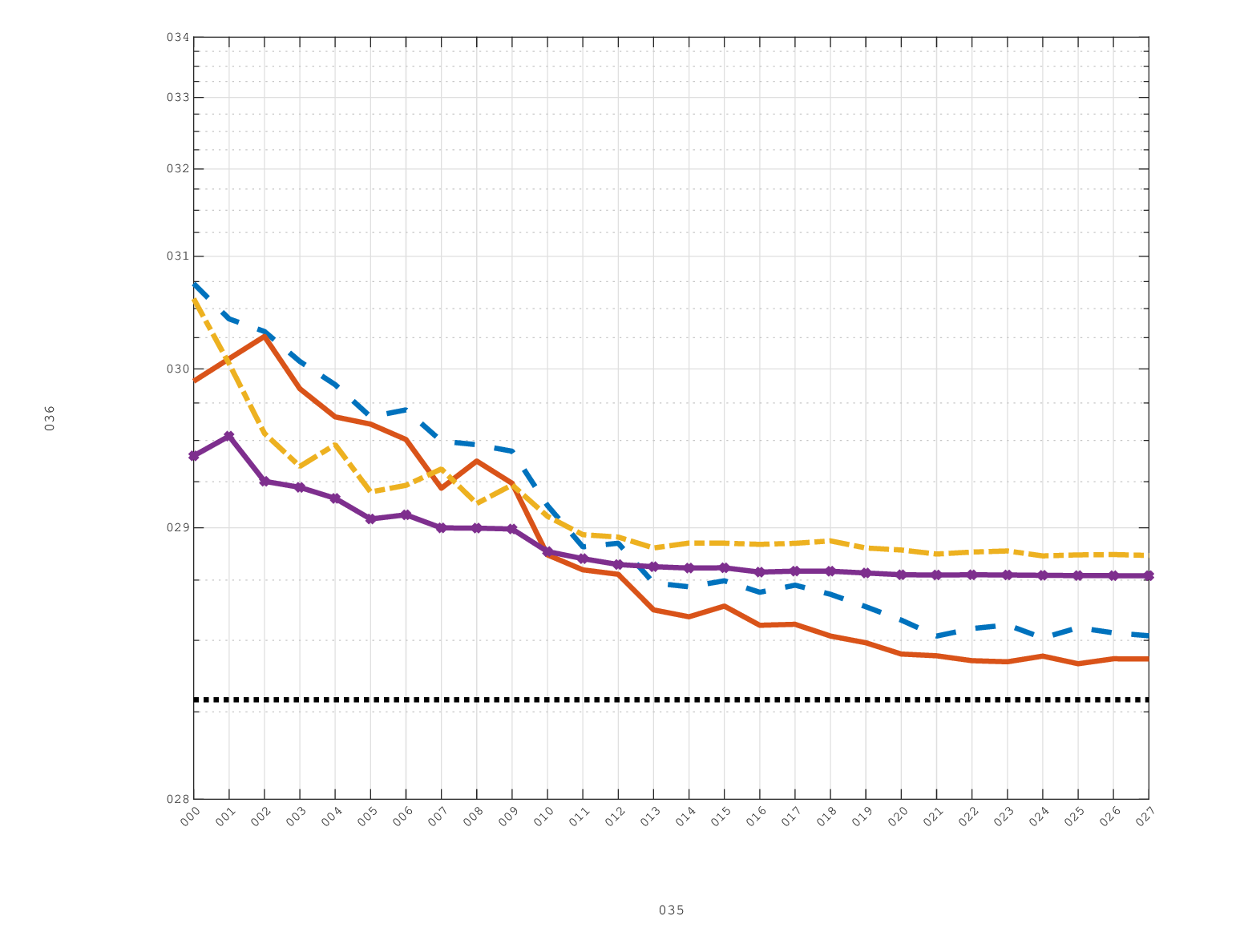}{}
		\subcaption{Fashion-MNIST}
	\end{subfigure}
	\caption{Comparison of the following: Approximation by step \ref{I2} of OMS when using tree structure (dashed with points, yellow) and when comparing all centers (solid with points, purple); approximation by step \ref{I2}+\ref{II2} of OMS when using tree structure (dashed, blue; \blue{ this line is mostly hidden behind the solid red curve in the first plot}) and when comparing all centers (solid, red). The black dotted line highlights average error caused by direct GMRA projection.} \label{Fig:TREEvsNOTREE}
\end{figure}

\begin{figure}[!ht]
	\centering
	\begin{subfigure}{0.14\textwidth}
		\includegraphics[width=\linewidth]{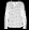}
		\caption{Ground truth. \\ \phantom{...}}
	\end{subfigure} \quad
	\begin{subfigure}{0.14\textwidth}
		\includegraphics[width=\linewidth]{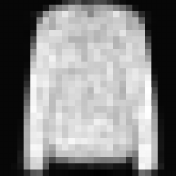}
		\caption{GMRA. \\ \phantom{...}}
	\end{subfigure} \quad
	\begin{subfigure}{0.14\textwidth}
		\includegraphics[width=\linewidth]{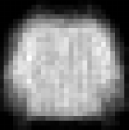}
		\caption{Tree structure\\ using step I.}
		\label{fig:Image4c}
	\end{subfigure} \quad
	\begin{subfigure}{0.14\textwidth}
		\includegraphics[width=\linewidth]{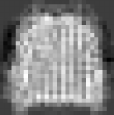}
		\caption{Tree structure\\ using steps I+II.}
		\label{fig:Image4d}
	\end{subfigure} \quad
	\begin{subfigure}{0.14\textwidth}
		\includegraphics[width=\linewidth]{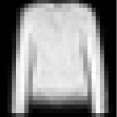}
		\caption{Full search\\ using step I.}
		\label{fig:Image4e}
	\end{subfigure} \quad
	\begin{subfigure}{0.14\textwidth}
		\includegraphics[width=\linewidth]{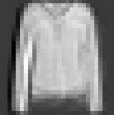}
		\caption{Full search\\ using steps I+II.}
		\label{fig:Image4f}
	\end{subfigure} 
	\caption{One Fashion-MNIST signal with OMS-reconstructions using tree search or full center search and using only step I or both steps, for $m=100$. The best GMRA approximation is given as a benchmark. Note that the GMRA uses a $7$-dimensional subspace at this part of the manifold.}
	\label{fig:Image4}
\end{figure}

\begin{figure}[!ht]
	\centering
	\begin{subfigure}{0.14\textwidth}
		\includegraphics[width=\linewidth]{GroundTruth2.pdf}
		\caption{Ground truth. \\ \phantom{...}}
	\end{subfigure} \quad
	\begin{subfigure}{0.14\textwidth}
		\includegraphics[width=\linewidth]{GMRAProj2.pdf}
		\caption{GMRA. \\ \phantom{...}}
	\end{subfigure} \quad
	\begin{subfigure}{0.14\textwidth}
		\includegraphics[width=\linewidth]{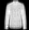}
		\caption{Tree structure\\ using step I.}
		\label{fig:Image3c}
	\end{subfigure} \quad
	\begin{subfigure}{0.14\textwidth}
		\includegraphics[width=\linewidth]{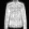}
		\caption{Tree structure\\ using steps I+II.}
		\label{fig:Image3d}
	\end{subfigure} \quad
	\begin{subfigure}{0.14\textwidth}
		\includegraphics[width=\linewidth]{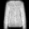}
		\caption{Full search\\ using step I.}
		\label{fig:Image3e}
	\end{subfigure} \quad
	\begin{subfigure}{0.14\textwidth}
		\includegraphics[width=\linewidth]{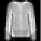}
		\caption{Full search\\ using steps I+II.}
		\label{fig:Image3f}
	\end{subfigure} 
	\caption{One Fashion-MNIST signal with OMS-reconstructions using tree search or full center search and using only step I or both steps, for $m = 1000$. The best GMRA approximation is given as a benchmark. Note that the GMRA uses a $7$-dimensional subspace at this part of the manifold.}
	\label{fig:Image3}
\end{figure}

\subsection{Tree vs. No Tree} \label{TREEvsNOTREE}

\paragraph{} \blue{In the fourth test we checked if approximation still works when not all possible centers are compared in step \ref{I} of \nameref{algorithm2} but their tree structure is used. More precisely, to find an optimal center one compares on the first refinement level all centers, and then continues in each subsequent level solely with the children of the $k$ best centers (in the presented experiments we chose $k = 10$). Of course, the chosen center will not be optimal as not all centers are compared (see Figure \ref{Fig:TREEvsNOTREE}). In the simple $2$-dimensional sphere setting, step \ref{II}, however, can compensate the worse approximation quality of \ref{I} with tree search. Figure \ref{Fig:TREEvsNOTREEa} hardly shows a difference in final approximation quality in both cases.
	In both MNIST settings, however, one can observe a considerable difference even when performing two steps.\\
Figures \ref{fig:Image4} and \ref{fig:Image3} illustrate the differences in reconstruction using tree search vs full center comparison respectively using only step I of OMS vs using both steps, for $m=100$ and $m = 1000$. This corresponds to a compression ratio of $100/6272 < 0.02$ resp.\ $1000/6272 < 0.2$, cf. the discussion in Section \ref{SIMPLEvsCONVEX}. Comparing \ref{fig:Image4c}-\ref{fig:Image4d} with \ref{fig:Image4e}-\ref{fig:Image4f} one sees that full center search reconstructs more detailed features (shape), comparing \ref{fig:Image3c} to \ref{fig:Image3d} and \ref{fig:Image3e} to \ref{fig:Image3f} one sees that the second step enhances reconstruction of details (neckline).
}\\

\subsection{A Change of Refinement Level}

\paragraph{} \blue{The last experiment (see Figure \ref{RefinementLevel}) examines the influence of the refinement level $j$ on the approximation error. For small $j$ (corresponding to a rough GMRA) a high number of measurements can hardly improve the approximation quality while for large $j$ (corresponding to a fine GMRA) the approximation error decreases with increasing measurement rates. This behavior is as expected. A rough GMRA cannot profit much from many measurements as the GMRA approximation itself yields an approximate lower bound on the obtainable approximation error.  For fine GMRAs the behavior along the measurement axis is similar to above experiments. Note that further increase of $j$ for the same range of measurements did not improve accuracy. Notably, the Fashion-MNIST reconstruction performs quite well even for small choices of $j$ suggesting that the manifold of shirt images is, at least in terms of GMRA approximation, of lower complexity than the manifold of digits "1". This is in line with the observation that the approximation error for Fashion-MNIST is smaller than for MNIST in the above experiments. The non-monotonous increase of the quality of reconstruction in Figure \ref{RefinementLevelc} -- the levels $j=4,8,10$ perform especially well -- is more difficult to explain; it	might relate to the fact that Fashion-MNIST is considered to be harder to classify than MNIST -- there seem to be features identifying the images (overall shape, detailed contours) on various levels of refinement.}


\begin{figure}[!ht]
	\centering
	\begin{subfigure}{0.48\textwidth}
		\centering
		\def\svgwidth{\linewidth}
		\psfragfig[width=0.48\textwidth]{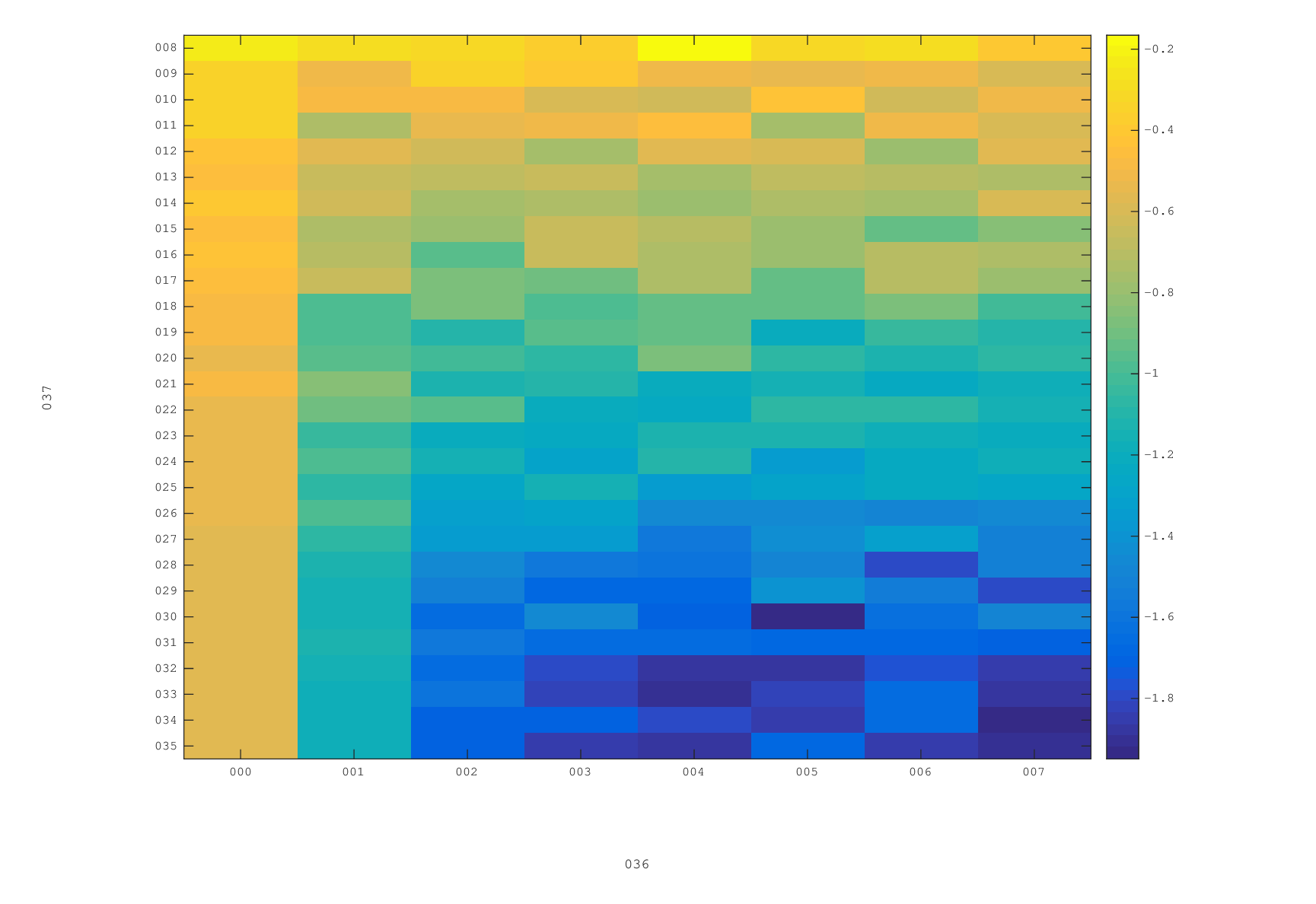}{}
		\subcaption{$2$-Sphere}
	\end{subfigure}

	\begin{subfigure}{0.48\textwidth}
		\centering
		\psfragfig[width=0.48\textwidth]{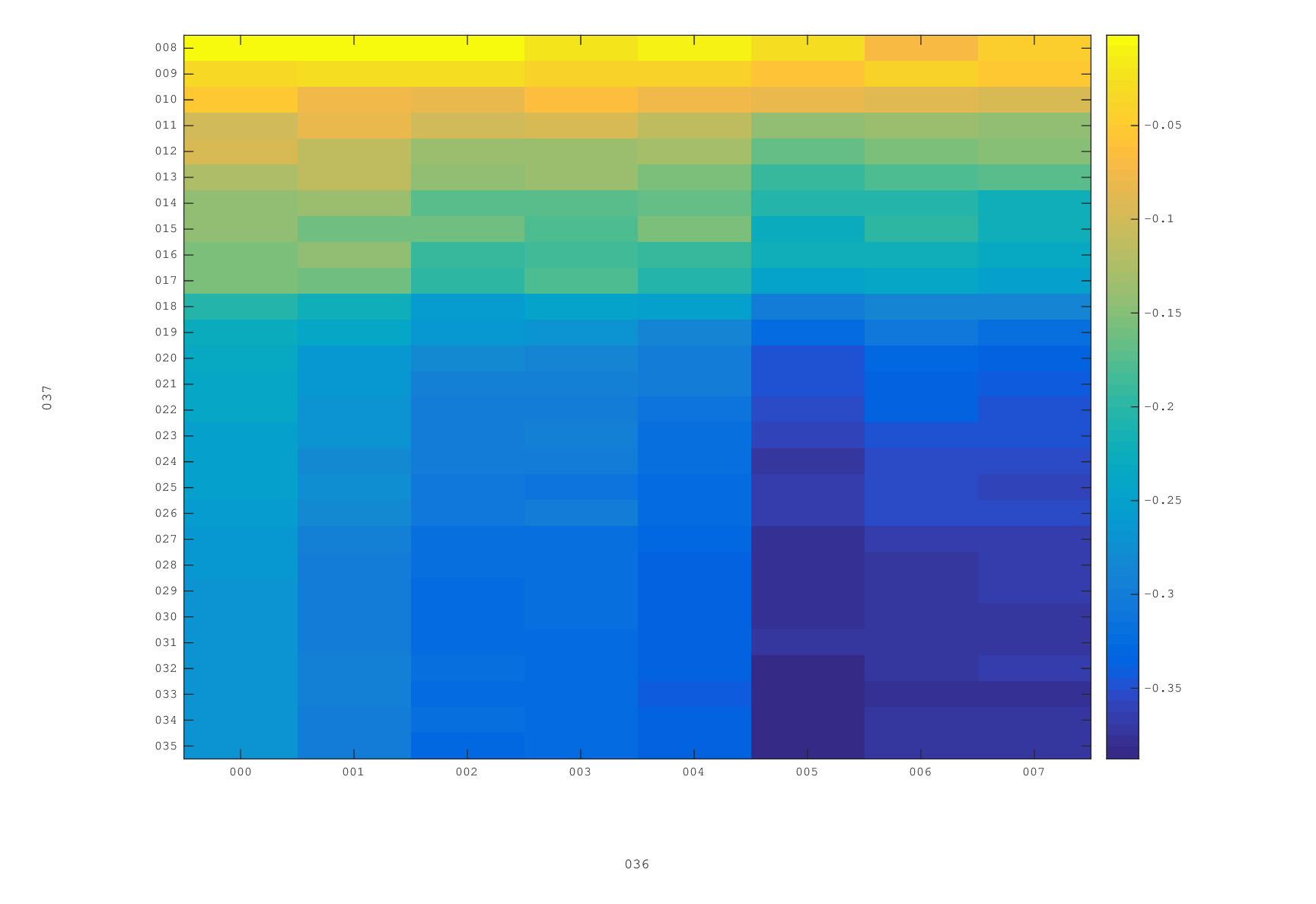}{}
		\subcaption{MNIST}
	\end{subfigure}
	\begin{subfigure}{0.48\textwidth}
	\centering
	\psfragfig[width=0.48\textwidth]{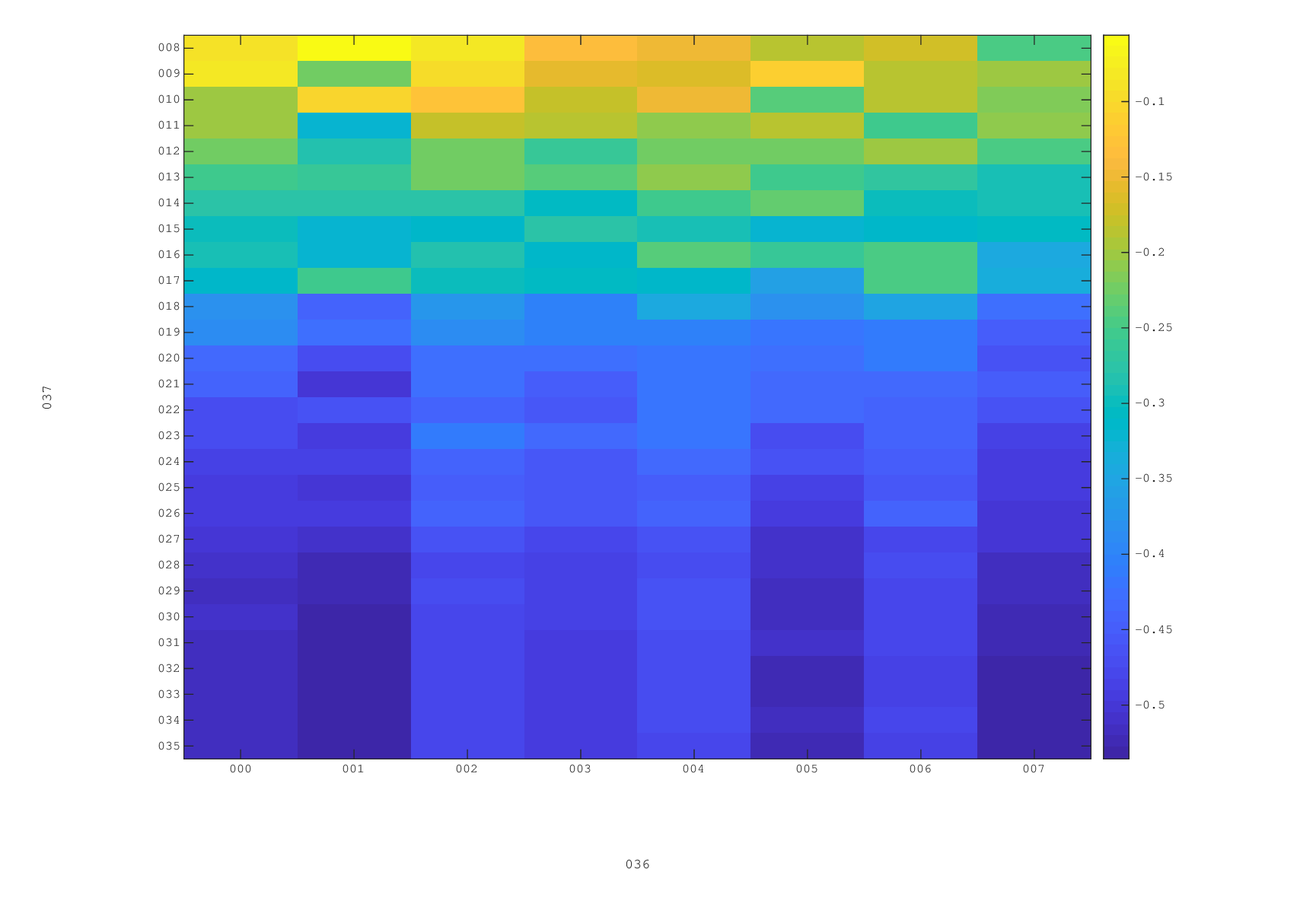}{}
	\subcaption{Fashion-MNIST}
	\label{RefinementLevelc}
\end{subfigure}
	\caption{Approximation error of OMS for different refinement levels $j$ and numbers of measurements.} \label{RefinementLevel}
\end{figure}

\section{Discussion}
\label{Discussion}

In this paper we proposed \nameref{algorithm2}, a tractable algorithm to approximate data lying on low-dimensional manifolds from compressive one-bit measurements, thereby  complementing the theoretical results of Plan and Vershynin on one-bit sensing for general sets in \cite{Plan2013} in this important setting. We then proved (uniform) worst-case error bounds for approximations computed by \nameref{algorithm2} under slightly stronger assumptions than \cite{Plan2013}, and also performed numerical experiments on both toy-examples and real-world data.  As a byproduct of our theoretical analysis (see, e.g., \S\ref{Proofs}) we have further linked the theoretical understanding of one-bit measurements as tessellations of the sphere \cite{Plan2014} to the GMRA techniques introduced in \cite{Allard2012} by analyzing the interplay between a given manifold and its GMRA approximation's complexity measured in terms of the Gaussian mean width. Finally, to indicate applicability of our results we show that they hold even if there are just random samples from the manifold at hand as opposed to the entire manifold (see, e.g., Appendix~\ref{EmpiricalGMRA}~and~\ref{EmpiricalGMRAproof}). Several interesting questions remain for future research however:

First, the experiments in Section \ref{TREEvsNOTREE} suggest a possible benefit from using the tree structure within $\C_j$. Indeed approximation of \nameref{algorithm2} does still yield comparable results if \ref{I} is restricted to a tree based search which has the advantage of being computable much faster than the minimization over all possible centers. It would be desirable to obtain theoretical error bounds even in this case, as well as to consider the use of other related fast nearest neighbor methods from computer science \cite{indyk1998approximate}.

Second, the attentive reader might have noticed in the empirical setting of Appendix \ref{EmpiricalGMRA} and \ref{EmpiricalGMRAproof} that \ref{A2} in combination with Lemma \ref{chjk'Bound} seems to imply that \ref{II} of \nameref{algorithm2} may be unnecessary.  As can be seen from Section \ref{STEPS} though, the second step of \nameref{algorithm2} yields a notable improvement even with an empirically constructed GMRA which hints that even with \ref{A2} not strictly fulfilled the empirical GMRA techniques remain valid, and \ref{II} of \nameref{algorithm2} of value.  Understanding this phenomenon might lead to more relaxed assumptions than \ref{A1}-\ref{A4}.

Third, it could be rewarding to also consider versions of \nameref{algorithm2} for additional empirical GMRA variants including, e.g., those which rely on adaptive constructions \cite{Liao2017}, GMRA constructions in which subspaces that minimize different criteria are used to approximate the data in each partition element (see, e.g., \cite{Iwen2015}), and distributed GMRA constructions which are built up across networks using distributed clustering \cite{BANDYOPADHYAY20061952} and SVD \cite{Ong2016} algorithms.  Such variants could prove valuable with respect to reducing the overall computational storage and/or runtime requirements of \nameref{algorithm2} in different practical situations.  

\blue{Fourth, it would be illustrative to compare our GMRA approach which can be interpreted as convexification (or at least reduction to a convex subproblem) with recent advances on non-convex ADMM \cite{latorre2019fast}. Two approaches seem of special interest: first, one could replace the GMRA manifold model by Generative Adversarial Networks (GANs) and apply ADMM to the resulting non-convex reconstruction problem or, second, stay in the GMRA setting but not restrict the reconstruction procedure to one single subspace.} 

\blue{Finally, it would also likely be fruitful to explore alternate quantization schemes to the one-bit measurements \eqref{one-bit-measurements} considered herein.  In particular, so called $\Sigma \Delta$ quantization schemes generally outperform the type of memoryless scalar quantization methods considered herein both theoretically and empirically in compressive sensing contexts \cite{KSY13,SAAB2018123,huynh2020fast}, and initial work suggests that they may provide similar improvements in the GMRA model setting considered here \cite{iwen2019new}. Nevertheless, one should be aware that feedback quantization schemes like $\Sigma \Delta$ quantization might be of limited practicability in recent large scale applications like massive MIMO \cite{haghighatshoar2018low} where the measurements are collected (and quantized) in a distributed fashion rendering the use of feedback information impossible. In such cases the just analyzed memoryless quantization model is of particular interest.}

\bibliographystyle{IEEEtranS}
\bibliography{IEEEabrv,mybibM}

\appendix

\section{Characterization of Convex Hull} 
\label{ConvexHull}

\begin{lemma}
    Let $P_{j,k'}$ be the affine subspace chosen in step \ref{I} of \nameref{algorithm} and define $\cj = \P_{j,k'}(\0)$. If $\0 \notin P_{j,k'}$, the following equivalence holds:
	\begin{align}
	\z \in \conv \left( \P_\St (P_{j,k'} \cap \B(\0,2)) \right) \Leftrightarrow \begin{cases}
	\| \z \|_2 \le 1, \\
	\Phi_{j,k'}^T \Phi_{j,k'} \z + \P_\cj(\z) = \z, \\
	\langle \z,\cj \rangle \ge \frac{1}{2} \| \cj \|_2^2.
	\end{cases}
	\end{align}
\end{lemma}

\begin{proof}
    First, assume $\z \in \conv \left( \P_\St (P_{j,k'} \cap \B(\0,2)) \right)$. Obviously, $\| \z \|_2 \le 1$. As projecting onto the sphere is a simple rescaling, $\conv \left( \P_\St (P_{j,k'} \cap \B(\0,2)) \right) \subset \Span ( P_{j,k'} )$ implying that $\Phi_{j,k'}^T \Phi_{j,k'} \z + \P_\cj(\z) = \z$. 
    For showing the third constraint note that any $\z' \in P_{j,k'}$ can be written as $\z' = \cj + (\z' - \cj)$ where $\z'-\cj$ is perpendicular to $\cj$. If in addition $\| \z' \|_2 \le 2$, we get
    \begin{align*}
    	\langle \P_\St(\z'),\cj \rangle = \left\langle \frac{\z'}{\| \z' \|_2}, \cj \right\rangle = \frac{\langle \cj, \cj \rangle}{\| \z' \|_2} \ge \frac{1}{2} \| \cj \|_2^2.
    \end{align*}
    As $\z$ is a convex combination of different $\P_\St (\z')$ the constraint also holds for $\z$.\\
    Let $\z$ fulfill the three constraints. Then $\z' = (\| \cj \|_2^2 / \langle \z, \cj \rangle) \cdot \z$ satisfies $\z' \in P_{j,k'}$ because of the second constraint and $\langle \z',\cj \rangle = \| \cj \|_2^2$. Furthermore, by the first and third constraint $\| \z' \|_2 \le (\| \cj \|_2^2 / \langle \z, \cj \rangle) \le 2$ and hence $\z' \in P_{j,k'}~\cap~\B(\0,\| \cj \|_2^2 / \langle \z, \cj \rangle) \- \subset P_{j,k'}~\cap~\B(\0,2)$. As $P_{j,k'} \cap \B(\0,\| \cj \|_2^2 / \langle \z, \cj \rangle)$ is the convex hull of $P_{j,k'} \cap (\| \cj \|_2^2 / \langle \z, \cj \rangle) \cdot \St^{D-1}$ , there are $\z_1,...,\z_n \in P_{j,k'}$ and $\lambda_1,...,\lambda_n \ge 0$ with $\| \z_k \|_2 = \| \cj \|_2^2 / \langle \z, \cj \rangle$ and $\sum \lambda_k = 1$ such that $(\| \cj \|_2^2 / \langle \z, \cj \rangle) \cdot \z = \sum \lambda_k \z_k$. Hence, $\z = \sum \lambda_k ( \langle \z, \cj \rangle / \| \cj \|_2^2) \cdot \z_k$. As $( \langle \z, \cj \rangle / \| \cj \|_2^2) \cdot \z_k \in \P_\St (P_{j,k'} \cap \B(\0,2))$ we get $\z \in \conv \left( \P_\St (P_{j,k'} \cap \B(\0,2)) \right)$.
\end{proof}

\section{Proof of Theorem~\ref{Riemann}} 
\label{ProofRieman}

	Denote by $\tau$ the reach of $\M$ and by $\rho$ the diameter $\diam(\M)$.
	First, note that for a set $K \subset \R^D$ by Dudley's inequality \added[id=josa]{\cite{Dudley2006}}
	\begin{align*}
	w(K) \le C' \int_0^{\diam(K)/2} \sqrt{\log(\mathcal{N}(K,\eps))} ~d\eps
	\end{align*}
	where $C'$ is an absolute constant. Second, \cite[Lemma 14]{Eftekhari2015} states that the covering number $\mathcal{N}(\M,\eps)$ of a $d$-dimensional Riemannian manifold $\M$ can be bounded by
	\begin{align*}
\mathcal N(\M,\eps)\leq\left(\frac{2}{\eps\sqrt{1-\left(\frac{\eps}{4\tau}\right)^2}} \right)^d\frac{\Vol(\M)}{\Vol(\mathcal B_d)}.
	\end{align*}
	for $ \eps \leq \frac\tau2$. After noting that $\Vol(\mathcal B_d)\geq \beta^{-1} \left(\frac{2\pi}{d} \right)^{\frac{d}{2}}$ for all $d\geq 1$ for an absolute constant $\beta > 1$, this expression may be simplified to
	\begin{equation*}
	\mathcal N(\M,\eps)\leq \beta \left(\frac{\sqrt{2d}}{\sqrt{\pi}\eps\sqrt{1-\left(\frac{\eps}{4\tau}\right)^2}} \right)^d \vol(\M) \leq \beta \left(\frac{\sqrt{d}}{\eps\sqrt{1-\left(\frac{\eps}{4\tau}\right)^2}} \right)^d \vol(\M).
	\end{equation*} We can combine these facts to obtain
	\begin{align*}
	w(\M)&\le C' \int_0^{\frac{\rho}{2}} \sqrt{\log(\mathcal{N}(\M,\eps))} ~d\eps\\
	&\le  C' \left( \frac{\rho}{2} \int_0^{\frac{\rho}{2}} {\log(\mathcal{N}(\M,\eps))} ~d\eps\right)^{\frac 12}\\
	&= C'\sqrt{\frac{\rho}{2}} \left(\int_0^{\frac\tau2} \log(\mathcal{N}(\M,\eps)) ~d\eps +\int_{\frac\tau2}^{\frac{\rho}{2}} \log(\mathcal{N}(\M,\eps)) ~d\eps \right)^{\frac 12},
	\end{align*}
	by using Cauchy-Schwarz inequality for the second inequality. We now bound the first integral by
	\begin{align*}
	\int_0^{\frac\tau2}\log(\mathcal{N}(\M,\eps)) ~d\eps&\leq\int_0^{\frac\tau2}  -d\log\left( \frac{\eps}{\beta \sqrt{d}}\underbrace{\sqrt{1-\left(\frac{\eps}{4\tau}\right)^2}}_{\geq\frac 12,\text{ as }\eps\leq\frac\tau 2}\right) + \log(\Vol(\M)) \mathrm ~d\eps\\
	&\leq \int_0^{\frac\tau2}  -d\log\left( \frac{\eps}{2\beta\sqrt{d}}\right)\mathrm d\eps + \frac{\tau}{2} \log(\Vol(\M))\\
	&=-d\left[\eps\log\left( \frac{\eps}{2\beta\sqrt{d}}\right)-\eps \right]_{0}^{\frac\tau2} + \frac{\tau}{2} \log(\Vol(\M))\\
	&=\frac{d\tau}2\left(\log\left(4 \frac{\beta\sqrt{d}}{\tau}\right)+1 \right) + \frac{\tau}{2} \log(\Vol(\M)).
	\end{align*}
	As the covering number is decreasing with increasing $\eps$, the second integral can be bounded as follows.
	\begin{align*}
	\int_{\frac\tau2}^{\frac{\rho}{2}} \log(\mathcal{N}(\M,\eps)) d\eps&\le \int_{\frac\tau2}^{\frac{\rho}{2}} \log\left(\mathcal{N} \left(\M,\frac{\tau}{2} \right) \right) d\eps\\
	&=\left(\frac{\rho}{2}-\frac \tau 2\right)\left[ -d\log\left( \frac{\tau}{c\sqrt{d}}\right) + \log(\Vol(\M)) \right]\\
	&=d\left(\frac{\rho}{2}-\frac \tau 2\right)\log\left( c\frac{\sqrt{d}}{\tau}\right) + \left(\frac{\rho}{2}-\frac \tau 2\right) \log(\Vol(\M)).
\end{align*}	
Both together yield
\begin{align*}
w(\M)&\leq C\sqrt{\frac{\rho}{2}} \left( \frac{d\tau}2\left(\log\left( c'\frac{\sqrt{d}}{\tau}\right)+1 \right) + d\left(\frac{\rho}{2}-\frac \tau 2\right)\log\left( c'\frac{\sqrt{d}}{\tau}\right) + \frac{\rho}{2} \log(\Vol(\M)) \right)^{\frac 12}\\
&\le C\sqrt{\frac{\rho}{2}} \left( d\tau \cdot \max \left\{ \log\left( c'\frac{\sqrt{d}}{\tau}\right), 1 \right\} + d\left( \rho- \tau \right) \cdot \max \left\{ \log\left( c'\frac{\sqrt{d}}{\tau}\right), 1 \right\} + \rho \log(\Vol(\M)) \right)^{\frac 12}\\
&=C\sqrt{\frac{\rho}{2}} \left( d\rho \cdot \max \left\{ \log\left( c'\frac{\sqrt{d}}{\tau}\right), 1 \right\} + \rho \log(\Vol(\M)) \right)^{\frac 12} \\
&\leq \frac{C}{\sqrt{2}} \rho \sqrt{ d \cdot \max \left\{ \log\left( c'\frac{\sqrt{d}}{\tau}\right), 1 \right\} + \log(\Vol(\M))}.
\end{align*}

\section{Proof of Lemmas~\ref{BoundOfGaussianWidthFine}~and~\ref{RiemannianGWidthBound}} 
\label{ProofWidthBounds}

  Recall that $\Mr_j := \{\P_{j,k_j(\z)}(\z)\colon \z\in \M \} \cap B(\0,2)$.  We will begin by establishing some additional technical lemmas.
  
\begin{lemma}
 Set $C_{\M}:=\sup_{\z\in\M}C_\z$ (cf.\  Property \eqref{3b}).  Then, $\Mr_j \subseteq \tube_{C_{\M} 2^{-2j}} (\M)$.
\label{lem:Tech1}
\end{lemma}
  
\begin{proof}
  If $\x \in \Mr_j$ there exists $\z_\x \in \M$ such that $\x = \P_{j,k_j(\z_\x)}(\z_\x)$.  The Euclidean distance $d(\x, \M)$ therefore satisfies
  $$d(\x, \M) = \inf_{\z \in \M} \| \x - \z \|_2 \leq \| \P_{j,k_j(\z_\x)}(\z_\x) - \z_\x \|_2 \leq C_{\M} 2^{-2j}$$
  by property \eqref{3b}.
\end{proof}
  

Given a subset $S \subset \mathbbm{R}^D$ we will let $\mathcal{N}(S,\eps)$ denote the cardinality of a minimal $\eps$-cover of $S$ by $D$-dimensional Euclidean balls of radius $\eps > 0$ each centered in $S$.  Similarly, we will let $\mathcal{P}(S,\eps)$ denote the maximal packing number of $S$ (i.e., the maximum cardinality of a subset of $S$ that contains points all of which are at least Euclidean distance $\eps > 0$ from one another.)  The following lemmas bound $\mathcal{N}(\Mr_j,\eps)$ for various ranges of $j$ and $\eps$.
   
\begin{lemma}
Set $C_{\M}:=\sup_{\z\in\M}C_\z$.  Then $\mathcal{N}(\Mr_j,\eps) \leq \mathcal{N}(\M,\eps / 2)$ for all $\eps \geq 2C_{\M} 2^{-2j}$.
\label{lem:Tech2}
\end{lemma}

\begin{proof}
First note that for all $\eta\ge\rho:= C_{\M} 2^{-2j}$ Lemma~\ref{lem:Tech1} implies that
    \begin{equation*}
    \Mr_j \subseteq \tube_\rho(\M) \subset \bigcup_{\p \in C_{\eta}}\B(\p,2\eta),
     \end{equation*}
    where $C_{\eta}$ is an $\eta$-cover of $\M$.  Thus, for all $\eps\ge2\eta \ge 2\rho$ 
    \begin{equation*}
    \mathcal{N}(\Mr_j,\eps) \leq \mathcal{N}\left(\bigcup_{\p \in C_{\eta}}\B(\p,2\eta),\eps \right) \leq \N(\M,\eta)=\N\left(\M,\frac{\eps}{2}\right).
     \end{equation*}    
\end{proof}

\begin{lemma}
$\mathcal{N}(\Mr_j,\eps) \leq (6 / \eps)^d \mathcal{N}(\M,\eps)$ for all $\eps \leq \frac{1}{4} C_1 2^{-j}$ as long as $j > j_0$ (see properties \eqref{tube}~and~\eqref{GMRA2b}).
\label{lem:Tech3}
\end{lemma}

\begin{proof}
By properties \eqref{tube}~and~\eqref{GMRA2b} every center $\cj_{j,k}$ has an associated $\p_{j,k} \in \M$ such that both $\B \left(\p_{j,k}, C_1 2^{-j-2} \right) \subset \B \left(\cj_{j,k}, C_1 2^{-j-1} \right)$ and $\B \left(\p_{j,k},  C_1 2^{-j-2} \right) \bigcap \B \left(\cj_{j,k'}, C_1 2^{-j-1} \right) = \emptyset$ for all $k \neq k'$.  Let $\tilde{P}_j := \left\{ \p_{j,k} ~|~k \in [K_j] \right\}$.  Consequently, we have that  $K_j = |\tilde{P}_j|$ and $\| \p_{j,k} - \p_{j,k'} \|_2 \geq  C_1 2^{-j-1}$ for all $k \neq k'$.  Since $\tilde{P}_j$ is a $C_1 2^{-j-1}$-packing of $\M$ we can further see that
$$K_j \leq \mathcal{P} \left( \M, C_1 2^{-j-1} \right) \leq \mathcal{N} \left( \M, C_1 2^{-j-2} \right) \leq \mathcal{N}(\M,\eps)$$
for all $\eps \leq  C_1 2^{-j-2}$.
Now, $\Mr_j \subseteq L_j$, where $L_j$ is defined as in the proof of Lemma~\ref{BoundOfGaussianWidth} (this proof also discusses its covering numbers).  As a result we have that 
$$\mathcal{N}(\Mr_j,\eps) \leq \mathcal{N}(L_j,\eps) \le K_j (6/\eps)^d \leq \mathcal{N}(\M,\eps) \cdot (6/\eps)^d$$
holds for all $\eps \leq  C_1 2^{-j-2}$.
\end{proof}  

\blue{We furthermore will use the two-sided Sudakov inequality as stated in \cite{Vershynin2017}.
\begin{lemma}[Two-sided Sudakov's inequality, \cite{Vershynin2017}] \label{lem:Sudakov}
	There exist absolute constants $c,C > 0$ such that the following holds. For $T \in \R^n$, we have that
	\begin{align*}
		c \sup_{\eps\ge 0} \eps \sqrt{\log(\N(T,\eps))} \le w(T) \le C \log(D) \sup_{\eps\ge 0} \eps \sqrt{\log(\N(T,\eps))}.
	\end{align*}
\end{lemma}}
   
 \subsection{Proof of Lemma~\ref{BoundOfGaussianWidthFine}}  
  
    \paragraph{} We aim to bound $w(\Mr_j)$ in terms of $w(\M)$.  
        \blue{By Lemmas \ref{lem:Tech1} and \ref{lem:Sudakov}}, we get that
    \begin{align*}
    w(\Mr_j)&\le C\log(D)\sup_{\eps\ge0}\eps\sqrt{\log\N(\Mr_j,\eps))}\\
    & \le C \log(D)\left( \sup_{0\le\eps\le 2C_{\M}2^{-2j}}\eps\sqrt{\log\N(\tube_{C_{\M} 2^{-2j}} (\M),\eps)} + \sup_{\eps\ge 2C_{\M}2^{-2j}}\eps\sqrt{\log\N(\Mr_j,\eps)}\right)\\
    & \le C \log(D)\left( \sup_{0\le\eps\le 2C_{\M}2^{-2j}}\eps\sqrt{\log\N(\B(\0,1+C_{\M}),\eps))} +  \sup_{\eps\ge 2C_{\M}2^{-2j}}\eps\sqrt{\log\mathcal{N}(\M,\eps / 2)} \right),
    \end{align*}
    where the last inequality follows from $\tube_{C_{\M} 2^{-2j}} (\M) \subseteq \B(\0,1+C_{\M})$ and Lemma~\ref{lem:Tech2}.  \blue{Appealing to Lemma \ref{lem:Sudakov} once more} to bound the second term
    above we learn that
    \begin{align*}
    w(\Mr_j) &\le C \log(D)\left( \sup_{0\le\eps\le 2C_{\M}2^{-2j}}\eps\sqrt{\log\N(\B(\0,1+C_{\M}),\eps))} +  2\sup_{\eps\ge 0}\frac{\eps}{2}\sqrt{\log\mathcal{N}(\M,\eps / 2)} \right)\\
    &\le C \log(D)\left( \sup_{0\le\eps\le 2C_{\M}2^{-2j}}\eps\sqrt{\log\N(\B(\0,1+C_{\M}),\eps))} +  2c ~w(\M) \right).
    \end{align*}
    
    To bound the first term above we note that using the covering number of $\B(\0,1+C_{\M})$ can be bounded as follows
    \begin{align*}
    \N \left(\B(\0,1+C_{\M}),\eps \right)=\N \left( \B(\0,1),\frac\eps {1+C_{\M}} \right) \le \left(1+\frac {2+2C_{\M}}\eps \right)^D\le \left( \frac{4C_{\M}+4}{\eps} \right)^D.
    \end{align*}
    As $\eps\mapsto\eps\sqrt{\log(\frac {4C_{\M}+4}\eps)}$ is non decreasing for $\eps \in (0,2C_{\M} 2^{-2j})$, we obtain by assuming that $\log_2(D) \le j$
    \begin{align*}
    \sup_{0\le\eps\le 2C_{\M}2^{-2j}}\eps\sqrt{\log\N(\B(\0,1+C_{\M}),\eps))} & \le\sup_{0\le\eps\le 2C_{\M}2^{-2j}} \eps \sqrt{D\log \left( \frac{4C_{\M}+4}\eps \right) }\\
    &\le 2C_{\M} 2^{-2j}\sqrt{D\log\left(\left(4+\frac 4{C_{\M}}\right) \cdot 2^{2j-1}\right)}\\
    &\le C_A \left( \frac{\sqrt{2j-1}}{2^j} \right) \left( \frac{\sqrt{D}}{2^{j}} \right) \le C'
    \end{align*}
    where $C'$ is an absolute constant.  Appealing to \eqref{equ:GenWidthUnionBound} now finishes the proof.

     \subsection{Proof of Lemma~\ref{RiemannianGWidthBound}}

 Let  $2C_{\M} 2^{-2j} \leq \tilde{\eps} \leq \frac{1}{4} C_1 2^{-j}$.  We aim to bound $w(\Mr_j)$ in terms of covering numbers for $\M$.  To do this we will use Dudley's inequality in combination with the knowledge that $\Mr_j \subset B(\0,2)$ (by definition).  By Dudley's inequality
	\begin{align*}
	w(\Mr_j) &\le C' \int_0^{2} \sqrt{\log(\mathcal{N}(\Mr_j,\eps))} ~d\eps \\
	&\le C' \left( \int_0^{ \tilde{\eps}} \sqrt{\log(\mathcal{N}(\Mr_j,\eps))} ~d\eps +  \int_{\tilde{\eps}}^{2} \sqrt{\log(\mathcal{N}(\Mr_j,\eps))} ~d\eps \right)
	\end{align*}
	where $C'$ is an absolute constant.  
	
	Appealing now to Lemmas~\ref{lem:Tech3} and~\ref{lem:Tech2} for the first and second terms above, respectively, we can see that
	\begin{align*}
	w(\Mr_j) &\le C' \left( \int_0^{ \tilde{\eps}} \sqrt{\log( (6 / \eps)^d \mathcal{N}(\M,\eps)  )} ~d\eps +  \int_{\tilde{\eps}}^{2} \sqrt{\log(  \mathcal{N}(\M,\eps / 2)  )} ~d\eps \right)\\
	&\le C' \int_0^2 \sqrt{\log( (6 / \eps)^d \mathcal{N}(\M,\eps / 2)  )} ~d\eps\\
	&= 2C' \int_0^1 \sqrt{d \log(3 / \eta) + \log( \mathcal{N}(\M,\eta)  )} ~d\eta\\
	&\leq 2C' \sqrt{ \int_0^1d \log(3 / \eta)~d\eta + \int_0^1 \log( \mathcal{N}(\M,\eta)  ) ~d\eta}
	\end{align*}
where the last bound follows from Jensen's inequality.  

We can now bound the second term as in the proof of Theorem~\ref{Riemann} in Appendix \ref{ProofRieman}.  Doing so we obtain
\begin{align*}
	w(\Mr_j) &\leq C'' \sqrt{ \int_0^1d \log(3 / \eta)~d\eta + d\left(1+\log\left( c'\frac{\sqrt{d}}{\tau}\right)\right) + \log(\Vol(\M))}\\
	&\leq C''' \sqrt{ d\left(1+\log\left( c'\frac{\sqrt{d}}{\tau}\right)\right) + \log(\Vol(\M))}
	\end{align*}
	where $\tau$ is the reach of $\M$, and $C''',c'$ are an absolute constants.  Appealing to \eqref{equ:GenWidthUnionBound} together with Theorem~\ref{Riemann} now finishes the proof.

\section{Data-Driven GMRA} 
\label{EmpiricalGMRA}

\paragraph{} The axiomatic definition of GMRA proves useful in deducing theoretical results but lacks connection to concrete applications where the structure of $\M$ is not known a priori.  Hence, in the following we first describe a probabilistic definition of GMRA which can be well approximated by empirical data (see \cite{Allard2012,GChen2012,Maggioni2016}) and is connected to the above axioms by applying results from \cite{Maggioni2016}. In fact, we will see that under suitable assumptions the probabilistic GMRA fulfills the axiomatic requirements and its empirical approximation allows one to obtain a version of Theorem \ref{GeneralApproximation} even when only samples from $\M$ are known.

\subsection{Probabilistic GMRA}

\paragraph{} 
A probabilistic GMRA of $\M$ with respect to a Borel probability measure $\Pi$, as introduced in  \cite{Maggioni2016},  is a family of (piecewise linear) operators $\{ \P_j \colon \R^D \rightarrow \R^D\}_{j \ge 0}$ of the form
\begin{align*}
\P_j(\x)=\sum_{k=1}^{K_j}\eins_{ C_{j,k}}(\x) \P_{j,k}(\x).
\end{align*}
Here, $\eins_{M}$ denotes the indicator function of a set $M$ and, for each refinement level $j \ge 0$, the collection of pairs of measurable subsets and affine projections $\{ (C_{j,k},\P_{j,k}) \}_{k = 1}^{K_j}$ has the following structure.

The subsets $C_{j,k}\subset\mathbb R^D$ for $k=1,\dots, K_j$ form a partition of $\mathbb R^D$, i.e., they are pairwise disjoint and their union is $\mathbb R^D$. The affine projectors are defined by \begin{align*}
\P_{j,k}(\x)={\bf c}'_{j,k}+\P_{V_{j,k}}(\x-{\bf c}'_{j,k}),
\end{align*}
where, for $X\sim \Pi$, ${\bf c}'_{j,k} =\mathbb{E}[X| X \in C_{j,k}] =: \Ejk{j,k}{X}\in \R^D$ and 
\begin{align*}
V_{j,k} &:= \argmin_{\dim(V)=d} \Ejk{j,k}{\|X-({\bf c}'_{j,k}+\mathrm{Proj}_V(X-{\bf c}'_{j,k}))\|_2^2},
\end{align*}
where the minimum is taken over all linear spaces $V$ of dimension $d$. From now on we will assume uniqueness of these subspaces $V_{j,k}$. To point out parallels to the axiomatic GMRA definition, think of $\Pi$ being supported on the tube of a $d$-dimensional manifold.
The axiomatic centers ${\bf c}_{j,k}$ are then considered to be approximately equal to the conditional means ${\bf c}'_{j,k}$ of some cells $C_{j,k}$ partitioning the space, and the corresponding affine projection spaces $P_{j,k}$ are spanned by eigenvectors of the $d$ leading eigenvalues of the conditional covariance matrix
\begin{align*}
	\Sigma_{j,k}=\Ejk{j,k}{(X-{\bf c}'_{j,k})(X-{\bf c}'_{j,k})^T}.
\end{align*}
Defined in this way, the $\P_j$ correspond to projectors onto the GMRA approximations $\M_j$ introduced above if ${\bf c}_{j,k} = {\bf c}'_{j,k}$.  From \cite{Maggioni2016} we adopt the following assumptions on the entities defined above, and hence, on the distribution $\Pi$. From now on we suppose that for all integers $j_{min}\leq j\leq j_{max}$ \ref{A1}-\ref{A4} (see Table \ref{assumptions}) hold true.

\begin{table}
\fbox{
	\begin{minipage}{16cm}
    \centering
	\begin{minipage}{15cm}
		\begin{enumerate}[label=\textbf{(A\arabic*)}]
			\item\label{A1} There exists an integer $1\leq d\leq D$ and a positive constant $\theta_1=\theta_1(\Pi)$ such that for all $k=1,\dots,K_j$, \begin{equation*}
			\Pi(C_{j,k})\geq\theta_1 2^{-jd}.
			\end{equation*}
			\item\label{A2} Define the restricted measure $\Pi_{j,k}$ by $\Pi_{j,k}(S) := \Pi(S \cap C_{j,k}) / \Pi(C_{j,k})$ for measurable $S$.  There is a positive constant $\theta_2=\theta_2(\Pi)$ such that for all $k=1,\dots, K_j$, if $X$ is drawn from $\Pi_{j,k}$ then, $\Pi_{j,k}$-almost surely,
			\begin{equation*}
			\|X-{\bf c}'_{j,k}\|_2\leq \theta_2 2^{-j}.
			\end{equation*}
			\item\label{A3} Denote the eigenvalues of the covariance matrix $\Sigma_{j,k}$ by $\lambda_1^{j,k}\geq\dots\geq\lambda_D^{j,k}\geq 0$. Then there exists $\sigma=\sigma(\Pi)\geq 0$, $\theta_3=\theta_3(\Pi)$, $\theta_4=\theta_4(\Pi)>0$, and some $\alpha>0$ such that for all $k=1,\dots,K_j$,
			\begin{equation*}
			\lambda_d^{k,j}\geq\theta_3\frac{2^{-2j}}{d}\text{ and } \sum_{l=d+1}^D\lambda_l^{j,k}\leq\theta_4(\sigma^2+2^{-2(1+\alpha)j})\leq\frac 12\lambda_d^{j,k}.
			\end{equation*}
			\item\label{A4} There exists $\theta_5=\theta_5(\Pi)$ such that
			\begin{equation*}
			\|\id-\P_j\|_{\infty,\Pi}\leq\theta_5(\sigma+2^{-(1+\alpha)j}),
			\end{equation*}
			where $\| T \|_{\infty,\Pi} = \sup_{x \in \supp(\Pi)} \| T(x) \|_2$, for $T \colon \R^D \rightarrow \R^D$.
		\end{enumerate}
	\end{minipage}
\end{minipage}
}\caption{The assumption set on $\Pi$.}\label{assumptions}
\end{table}

\begin{remark}
	Assumption \ref{A1} ensures that each partition element contains a reasonable amount of $\Pi$-mass.  Assumption \ref{A2} guarantees that all samples from $\Pi_{j,k}$ will lie close to its expection/center.  As a result, each ${\bf c}'_{j,k}$ must be somewhat geometrically central within $C_{j,k}$. Together, \ref{A1} and \ref{A2} have the combined effect of ensuring that the probability mass of $\Pi$ is somewhat equally distributed onto the different sets $C_{j,k}$, i.e., the number of points in each set $C_{j,k}$ is approximately the same, at each scale $j$. The third and fourth assumptions \ref{A3} and \ref{A4} essentially constrain the geometry of the support of $\Pi$ to being effectively $d$-dimensional and somewhat regular (e.g., close to a smooth $d$-dimensional submanifold of $\R^D$).  We refer the reader to \cite{Maggioni2016} for more detailed information regarding these assumptions. \end{remark}

\paragraph{} An important class of probability measures $\Pi$ fulfilling \ref{A1}-\ref{A4} is presented in \cite{Maggioni2016}. For the sake of completeness we repeat it here and also discuss a method of constructing the partitions $\{C_{jk}\}_{k=1}^{K_j}$ from such probabilities measures. 
From here on let $\M$ be a smooth $d$-dimensional submanifold of $\mathbb S^{D-1} \subset \R^D$. 
Let $\mathcal U_{ K}$ denote the uniform distribution on a given set $K$.  We have the following definition.

\begin{definition}[{\cite[Definition 3]{Maggioni2016}}]
	Assume that $0\leq\sigma<\tau$. The distribution $\Pi$ is said to satisfy the \emph{$(\tau,\sigma)$-model assumption} if  $(i)$ there exists a smooth, compact submanifold $\M\hookrightarrow\mathbb R^D$ with reach $\tau$ such that $\supp(\Pi)=\tube_\sigma(\M)$, $(ii)$ the distributions $\Pi$ and $\mathcal U_{\tube_\sigma(\M)}$ are absolutely continuous with respect to each other so the Radon-Nikodym derivative $\frac{\mathrm d \Pi}{\mathrm d \mathcal U_{\tube_\sigma(\M)}}$ exists and satisfies
	\begin{equation*}
	0<\phi_1\leq\frac{\mathrm d \Pi}{\mathrm d\mathcal U_{\tube_\sigma(\M)}}\leq\phi_2<\infty\qquad\mathcal U_{\tube_\sigma(\M)}-\text{almost surely}.
	\end{equation*}
	The constants $\phi_1$ and $\phi_2$ are implicitly assumed to only depend on a slowly growing function of $D$, compare \cite[Remark 4]{Maggioni2016}.
\end{definition}

\paragraph{} Let us now discuss the construction of suitable partitions $\{C_{jk}\}$ by making use of cover trees. A \emph{cover tree} $T$ on a finite set of samples $S \subset \M$ is a hierarchy of levels with the starting level containing the root point and the last level containing every point in $S$. To every level a set of nodes is assigned which is associated with a subset of points in $S$. To be precise, given a set $S$ of $n$ distinct points in some metric space $(\mathbb X,d_{\mathbb X})$. A cover tree $T$ on $S$ is a sequence of subsets $T_i\subset S, i=0,1,\dots$ that satisfies the following, see \cite{Beygelzimer2006}:
\begin{enumerate}
	\item[(i)]\textbf{Nesting:} $T_i\subseteq T_{i+1}$, i.e., once a point appears in $T_i$ it is in every $T_j$ for $j\geq i$.
	\item[(ii)] \textbf{Covering:} For every $\x\in T_{i+1}$ there exists exactly one $\y\in T_{i}$ such that $d_{\mathbb X}(\x,\y)\leq2^{-i}$.  Here $\y$ is called the {\it parent} of $\x$.
	\item[(iii)] \textbf{Separation:} For all distinct points $\x,\y\in T_i$, $d_{\mathbb X}(\x,\y)>2^{-i}$.
\end{enumerate}
The set $T_i$ denotes the set of points in $S$ associated with nodes at level $i$. Note that there exists $N\in\mathbb N$ such that $T_i=S$ for all $i\geq N$.  Herein we will presume that $S$ is large enough to contain an $\epsilon$-cover of $\M$ for $\epsilon > 0$ sufficiently small.\\

Moreover, the axioms characterizing cover trees are strongly connected to the dyadic structure of GMRA. For a given cover tree (for construction see \cite{Beygelzimer2006}) on a set $\mathcal X_n=\{X_1,\dots X_n\}$ of i.i.d.\ samples from the distribution $\Pi$ with respect to the Euclidean distance let $\aj_{j,k}$ for $k=1,\dots,K_j$ be the elements of the $j$th level of the cover tree, i.e.\ $T_j=\{\aj_{j,k}\}_{k=1}^{K_j}$ and define
\begin{equation*}
\kappa_j(x)=\argmin_{1\leq k\leq K_j} \|\x-\aj_{j,k}\|_2.
\end{equation*}
With this a partition of $\mathbb R^D$ into \emph{Voronoi regions} \begin{equation}\label{Eq:Voronoi}
C_{j,k}=\{ \x\in\mathbb R^D\colon \kappa_j(\x)=k\}
\end{equation}
can be defined. Maggioni et.\ al.\ showed in \cite[Theorem 7]{Maggioni2016} that by this construction all assumptions \ref{A1}-\ref{A4} can be fulfilled.


The question arises if the properties of the axiomatic definition of GMRA in Definition \ref{GMRA} are equally met. As only parts of the axioms are relevant for our analysis, we refrain from giving rigorous justification for all properties.\\
\begin{enumerate}
\item GMRA property \eqref{first} holds by construction if the matrices $\Phi_{j,k}$ are defined, s.t. $\Phi_{j,k}^T \Phi_{j,k} = \P_{V_{j,k}}$ along with any reasonable choice of centers ${\bf c}_{j,k}$.
\item The dyadic structure axioms \eqref{GMRA2a} -- \eqref{GMRA2c} also hold as a trivial consequence of the cover tree properties $(i)$ -- $(iii)$ above if the axiomatic centers ${\bf c}_{j,k}$ are chosen to be the elements of the cover tree set $T_j$ (i.e., the $\aj_{j,k}$ elements).
By the $(\rho,\sigma)$-model assumption samples drawn from $\Pi$ will have a quite uniform distribution all over $\supp(\Pi)$. Hence, the probabilistic centers ${\bf c}'_{j,k}$ of each $C_{j,k}$-set will also tend to be close to the axiomatic centers ${\bf c}_{j,k} = \aj_{j,k}$ proposed here for small $\sigma$ (see, e.g., assumption \ref{A2} above).
\item One can deduce GMRA property \eqref{tube} from the fact that our chosen centers $\aj_{j,k}$ belong to $\M$ if $\supp(\Pi) = \M$ (or to a small tube around $\M$ if $\sigma$ is small).
\item The first part of \eqref{3b} is implied by \ref{A4} with the uniform constant $\theta_5$ for all $\x \in \M$ if $\aj_{j,k}$ is sufficiently close to ${\bf c}'_{j,k}$. To show the second part of \eqref{3b} note that
\begin{align*}
\|\x-\P_{j,k'}(\x)\|_2 & \leq\| \x-{\bf c}_{j,k'}\|_2 + \|{\bf c}_{j,k'}-\P_{j,k'}(\x)\|_2 = \| \x - {\bf c}_{j,k'}\|_2 + \|\P_{V_{j,k'}}( \x- {\bf c}_{j,k'})\|_2 \\
&\leq2\| \x- {\bf c}_{j,k'}\|_2 \leq 32\max\{\| \x - {\bf c}_{j,k_j(\x)}\|_2,C_1 2^{-j-1}\} \\
&\leq 32\max\{C_\epsilon 2^{-j},C_1 2^{-j-1}\} \le C \cdot 2^{-j}
\end{align*}
where in the second last step we used our cover tree properties (recall that ${\bf c}_{j,k} = \aj_{j,k}$). Again, the constants $C, C_\epsilon > 0$ do not depend on the chosen $x \in \M$ as long as $S$ is well chosen (e.g., contains a sufficiently fine cover of $\M$).
\end{enumerate}

Considering the GMRA axioms above we can now see that only the first part of \eqref{3b} may not hold in a satisfactory manner if we choose to set $\Phi_{j,k}^T \Phi_{j,k} = \P_{V_{j,k}}$ and ${\bf c}_{j,k} = \aj_{j,k}$.  And, even when it doesn't hold with $C_z$ being independent of $j$ it will still at least still hold with a worse $j$ dependence due to assumption \ref{A2}.

\subsection{Empirical GMRA}

The axiomatic properties only hold above, of course, if the GMRA is constructed with knowledge of the true $\P_{V_{j,k}}$-subspaces. In reality, however, this won't be the case and we are rather given some training data consisting of $n$ samples from near/on $\M$, $\mathcal{X}_n = \{ X_1,...,X_n \}$, which we assume to be i.i.d. with distribution $\Pi$. These samples are used to approximate the real GMRA subspaces based on $\Pi$ such that the operators $\P_j$ can be replaced by their estimators
\begin{align*}
\Ph_j(\x)&=\sum_{k=1}^{K_j}\eins_{\{\x\in C_{j,k}\}}\Ph_{j,k}(\x)\\
\intertext{where $\{ C_{j,k}\}_{k=1}^{K_j}$ is a suitable partition of $\mathbb R^d$ obtained from the data,}
\Ph_{j,k}(\x) &= \ch_{j,k}+\P_{\Vh_{j,k}}(\x-\ch_{j,k}),\\
\ch_{j,k} &= \frac{1}{|\X_{k,j}|}\sum_{\x\in \X_{j,k}}\x, \\
\Vh_{j,k} &= \argmin_{\dim(V)=d}\frac{1}{|\X_{j,k}|}\sum_{\x\in \X_{j,k}}\|\x-\ch_{j,k}-\P_{V}(\x-\ch_{j,k})\|_2^2
\end{align*}
and $\X_{j,k} = C_{j,k} \cap \X_n$.
In other words, working with above model we have one perfect GMRA that cannot be computed (unless $\Pi$ is known) but fulfills all important axiomatic properties, and an estimated GMRA that is at hand but that is only an approximation to the perfect one. Thankfully, the main results of \cite{Maggioni2016} stated in Appendix \ref{EmpiricalGMRAproof} give error bounds on the difference between perfect and estimated GMRA with ${\bf c}_{j,k} = \ch_{j,k} \approx {\bf c}'_{j,k} \approx \aj_{j,k}$ that only depend on the number of samples from $\Pi$ one can acquire.  Following their notational convention we will denote the empirical GMRA approximation at level $j$, i.e., the set $\Ph_j$ projects onto, by $\Mh_j = \{ \Ph_j(\z) \colon \z \in \B(\0,2) \} \cap \B(\0,2)$ and the affine subspaces by $\widehat{P}_{j,k} = \{ \Ph_{j,k}(\z) \colon \z \in \R^D \}$. We again restrict the approximation to $\B(\0,2)$. The single affine spaces will be non-empty as all $\ch_{j,k}$ lie by definition close to $\B(\0,1)$ if $\supp(\Pi)$ is close to $\M$, which we assume. 

\paragraph{} In the empirical setting \nameref{algorithm2} has to be slightly modified to conform to our empirical GMRA notation. Hence, \eqref{MinCnew} and \eqref{MinXnew} become
\begin{align}
&\phantom{xx}\ch_{j,k'} \in \argmin_{\ch_{j,k} \in \Ch_j} \; d_H( \sign(A\ch_{j,k}),\y ). \label{MinCh}\\
&\begin{cases} \label{MinXh}
\x^\ast = \argmin_{\z \in \R^D} \sum_{l=1}^{m} (-y_l) \langle \aj_l, \z \rangle , \\
\text{subject to } \z \in \conv \left( \P_\St (\widehat{P}_{j,k'} \cap \B(\0,2)) \right).
\end{cases}
\end{align}
\nameref{algorithm2} can be adapted in a similar way by changing \eqref{MinCnew} and \eqref{MinXnew}. To stay consistent with the axiomatic notation we denote the sets containing the centers ${\bf c}'_{j,k}$ and $\ch_{j,k}$ by $\C'_j$ and $\Ch_j$ respectively. As shown in Appendix \ref{EmpiricalGMRAproof} the main result also holds in this setting. There is only an additional influence of sample size on the probability.

\section{Proof of \Cref{GeneralApproximation} with Empirical GMRA}
\label{EmpiricalGMRAproof}

\paragraph{} Recall the definitions of probabilistic GMRA, empirical GMRA and the modifications of \eqref{MinCnew} resp. \eqref{MinXnew} to become \eqref{MinCh} resp. \eqref{MinXh}. We start with the central result of \cite{Maggioni2016}. 

\begin{theorem}[{\cite[Theorem 2]{Maggioni2016}}] \label{Theorem2}
	Suppose that assumptions \ref{A1}-\ref{A3}
	are satisfied (see Table \ref{assumptions}). Let $X,X_1,\dots X_n$ be an i.i.d.\ sample from $\Pi$ and set $\bar d=4d^2\frac{\theta_2^4}{\theta_3^2}$. Then for any $j_{\min}\leq j\leq j_{\max}$ and any $t\geq 1$ such that $t+\log(\max\{\bar d,8\})\leq\frac 1 2\theta_1n 2^{-jd}$,
	\begin{equation*}
	\E{\|X-\Ph_j(X)\|_2^2}\leq 2\theta_4\left( \sigma^2+2^{-2j(1+\alpha)}\right)+c_1 2^{-2j}\frac{(t+\log(\max\{\bar d,8\}))d^2}{n2^{-jd}},
	\end{equation*}
	and if in addition \ref{A4} is satisfied,
	\begin{equation*}
	\left\| \id -\Ph_j\right\|_{\infty,\Pi}\leq\theta_5\left(\sigma+2^{-(1+\alpha)j}\right)+\sqrt{ \frac{c_1}{2} 2^{-2j}\frac{(t+\log(\max\{\bar d,8\}))d^2}{n2^{-jd}}}
	\end{equation*}
	with probability $\geq 1-\frac{2^{jd+1}}{\theta_1}\left(e^{-t}+e^{-\frac{\theta_1}{16}n2^{-jd}}\right)$, where $c_1=2\left(12\sqrt{2}\frac{\theta_2^3}{\theta_3\sqrt{\theta_1}}+4\sqrt{2}\frac{\theta_2}{d\sqrt{\theta_1}}\right)^2$.
\end{theorem}

\paragraph{} Theorem \ref{Theorem2} states that under assumptions \ref{A1}-\ref{A4} the empirical GMRA approximates $\M$ as well as the perfect probabilistic one as long as the number of samples $n$ is sufficiently large. For the proof of our main theorem we only need the following two bounds which can be deduced from (20) and (21) in \cite{Maggioni2016} by setting $t = 2^{jd}$. As both appear in the proof of Theorem \ref{Theorem2}, we state them as a corollary. The interested reader may note that $n_{j,k}$ appearing in the original statements can be lower bounded by $\theta_1 n 2^{-jd}$.
\begin{corollary} \label{ImportantBounds}
	Under the assumptions of Theorem \ref{Theorem2} the following holds for any $C_1 > 0$ as long as $j, \alpha$ are sufficiently large and $\sigma$ is sufficiently small:
	\begin{align*}
	\Pr{\max_{k \in K_j} \left\| \P_{V_{j,k}}-\P_{\Vh_{j,k}}\right\| \ge \frac{C_1}{12} 2^{-j-2}}{} &\le \frac{2}{\theta_2} 2^{jd} e^{-2^{jd} \min\left\{ 1, \frac{32 \theta_2^2 d^2}{C_1^2} \right\}}\\
	\Pr{\max_{k \in K_j} \left\| {\bf c}'_{j,k}-\ch_{j,k} \right\|_2 \ge \frac{C_1}{12} 2^{-j-2}}{} &\le \frac{2}{\theta_2} 2^{jd} e^{-2^{jd} \min\left\{ 1, \frac{32 \theta_2^2 d^2}{C_1^2} \right\}}
	\end{align*}
	if $n \ge n_\text{min} = \left( 2^{jd} + \log (\max\{\bar d,8\}) \right) \min \left\{ 144 \frac{\theta_2^2 d}{C_1 \theta_1 \theta_3} 2^{(d+1)j + 3}, 96 \frac{\theta_2}{C_1 \theta_1} 2^{dj+1} \right\}^2$.
\end{corollary}
\begin{remark}
	By Corollary \ref{ImportantBounds} with probability of at least $1 - \mathcal{O}(2^{jd} \exp (-2^{jd}))$ the empirical centers $\ch_{j,k}$ of one level $j$ have a worst case distance to the perfect centers ${\bf c}'_{j,k}$ of at most $\mathcal{O}(2^{-j-2})$ if $n \gtrsim \mathcal{O}(2^{3jd})$.  As a result, the empirical centers $\ch_{j,k}$ will also be at most $\mathcal{O}(2^{-j-2})$ distance from their associated cover tree centers $\aj_{j,k}$ if $n \gtrsim \mathcal{O}(2^{3jd})$ by assumption \ref{A2}.  The same holds true for the projectors $\P_{V_{j,k}}$ and $\P_{\Vh_{j,k}}$ in operator norm.
\end{remark}

\paragraph{} The proof of Theorem \ref{GeneralApproximation} in this setting follows the same steps as in the axiomatic one. First, we give an empirical version of Lemma \ref{cjk'Bound}. Then we link $x$ and $x^\ast$ as described in Section \ref{ProofAxiomatic} while controlling the difference between empirical and axiomatic but unknown GMRA by Corollary \ref{ImportantBounds}. The following extension of Lemma \ref{BoundOfGaussianWidth} will be regularly used.

\begin{corollary}[Bound of Gaussian width] \label{BoundOfGaussianWidth2}
	The Gaussian width of $\M\cup\M_j\cup\widehat\M_j$ can be bounded from above by
	\begin{align*}
		\max\{ w(\M),w(\M_j),w(\widehat{\M}_j) \} &\le  w(\M\cup\M_j\cup\widehat\M_j) \leq 2w(\M) + 2w(\M_j) + 2w(\widehat \M_j) + 5 \\ 
		&\le 2w(\M)+ C\sqrt{dj}
	\end{align*} 
	where $\widehat M_j$ is defined as at the end of Appendix \ref{EmpiricalGMRA}.
\end{corollary}

\begin{proof}
	The proof follows directly the lines of the proof of Lemma \ref{BoundOfGaussianWidth}. The additional term $w(\widehat \M_j)$ can be bounded in the same way as $w(\M_j)$. 
\end{proof}
\begin{remark}
	By structure of the proof one can easily obtain several subversions of the inequalities, e.g., $w(\M \cup \Mh_j) \le 2w(\M) + 2w(\Mh_j) + 5$. We will use them while only referring to Corollary \ref{BoundOfGaussianWidth2}. Moreover, similar generalizations as in Lemma \ref{BoundOfGaussianWidth} apply (cf.\ Remark \ref{rem:BoundOfGaussianWidth})
\end{remark}

Note that we are now setting our empirical GMRA centers ${\bf c}_{j,k}$ to be the associated mean estimates $\ch_{j,k}$ as a means of approximating the axiomatic GMRA structure we would have if we had instead chosen our centers to be the true expectations ${\bf c}'_{j,k}$ (recall Appendix \ref{EmpiricalGMRA}).  We also implicitly assume below that there exists a constant $C_1 > 0$ for which the associated axiomatic GMRA properties in Section \ref{ProblemFormulation} hold when the centers $\cj_{j,k}$ are chosen as these true expectations ${\bf c}'_{j,k}$ and the $\Phi_{j,k}^T \Phi_{j,k}$ as $\P_{V_{j,k}}$.

\begin{lemma}\label{chjk'Bound}
	Fix $j$ sufficiently large.  Under the assumptions of Theorem \ref{Theorem2} and $n \ge n_\text{min}$ if $m \ge \bar{C} C_1^{-6} 2^{6(j+1)} w(\M \cup \P_\St(\Ch_j))^2$ the index $k'$ of the center $\ch_{j,k'}$ chosen in step $I$ of the algorithm fulfills
	\begin{align*}
	\| \x - {\bf c}'_{j,k'} \|_2 \le 16 \max \{ \| \x - {\bf c}'_{j,k_j(\x)} \|_2, C_1 2^{-j-1} \}.
	\end{align*}
	for all $\x \in \M \subset \St^{D-1}$ with probability at least $1 - \mathcal{O}\left( 2^{jd} \exp (-2^{jd}) + \exp(\delta^2 m) \right)$.
\end{lemma}

\begin{proof}
	The proof will be similar to the one of Lemma \ref{cjk'Bound}. By definition we have
	\begin{align*}
	d_H ( \sign(A\ch_{j,k'}),\y ) \le d_H ( \sign(A\ch_{j,k_j(\x)}),\y ).
	\end{align*}
	As, for all $\z,\z' \in \St^{D-1}$, $d_H(\sign(A \z),\sign(A \z')) = m \cdot d_A( \z, \z')$, this is equivalent to
	\begin{align*}
	d_A ( \P_\St (\ch_{j,k'}), \x ) \le d_A ( \P_\St (\ch_{j,k_j( \x)}), \x ).
	\end{align*}
	Theorem \ref{RUT} transfers the bound to normalized geodesic distance, namely
	\begin{align*}
	d_G ( \P_\St (\ch_{j,k'}),\x ) \le d_G ( \P_\St (\ch_{j,k_j(\x)}),\x ) + 2\delta
	\end{align*}
	with probability at least $1-2\exp(-c\delta^2m)$ where $\delta = C_1 2^{-j-1}$. Observe $d_G(\z,\z') \le \| \z-\z' \|_2 \le \pi d_G(\z,\z')$ for all $\z,\z' \in \St^{D-1}$ (see Lemma \ref{Lem:NormalizedGeo}) which leads to
	\begin{align*}
	\| \P_\St (\ch_{j,k'}) - \x \|_2 &\le \pi d_G( \P_\St (\ch_{j,k_j(\x)}),\x) + 2\pi \delta \\
	&\le \pi \| \P_\St (\ch_{j,k_j(\x)}) - \x \|_2 + 2\pi \delta.
	\end{align*}
	We will now use the fact that by Corollary \ref{ImportantBounds}
\begin{align*}
	\| \ch_{j,k} - {\bf c}'_{j,k} \|_2 \le \frac{C_1}{12} 2^{-j-2}
	\end{align*}
for all $k \in K_j$ with probability at least $1 - \mathcal{O}(2^{jd} \exp (-2^{jd}))$. From this we first deduce by GMRA property \eqref{tube} that $\| \ch_{j,k} - \P_\St (\ch_{j,k}) \|_2 \le \| \ch_{j,k} - \P_\St ({\bf c}'_{j,k}) \|_2 \le \| \ch_{j,k} - {\bf c}'_{j,k} \|_2 + \| {\bf c}'_{j,k} - \P_\St ({\bf c}'_{j,k}) \|_2 < (C_1 + C_1/2) 2^{-j-2}$ for all $\ch_{j,k} \in \Ch_j$. Combining above estimates and using triangle inequality we obtain
	\begin{align*}
	\| {\bf c}_{j,k'} - \x \|_2 &\le \| \P_\St (\ch_{j,k'}) - \x \|_2 + \| \ch_{j,k'} - \P_\St (\ch_{j,k'}) \|_2 + \| \ch_{j,k'} - {\bf c}'_{j,k'} \|_2 \\
	&< \pi \| \P_\St (\ch_{j,k_j(x)}) - \x \|_2 + 2\pi \delta + 2C_1 2^{-j-2} \\
	&\le \pi (\| \ch_{j,k_j(\x)} - \P_\St (\ch_{j,k_j(\x)}) \|_2 + \| \ch_{j,k_j(\x)} - {\bf c}'_{j,k_j(\x)} \|_2 + \| {\bf c}'_{j,k_j(\x)} - \x \|_2) + 2\pi \delta + C_1 2^{-j-1} \\
	&< \pi \| {\bf c}_{j,k_j(\x)} - \x \|_2 + 2\pi \delta + (1+\pi)C_1 2^{-j-1} \\
	&\le (4\pi + 1) \max \{ \| c_{j,k_j(\x)} - \x \|_2, C_1 2^{-j-1} \}\\
	&\le 16 \max \{ \| c_{j,k_j(\x)} - \x \|_2, C_1 2^{-j-1} \}.
	\end{align*}
	A union bound over both probabilities yields the result.
\end{proof}

\paragraph{} Having Lemma~\ref{chjk'Bound} at hand we can now show a detailed version of Theorem \ref{GeneralApproximation} in this case. For convenience please first read the proof of  Theorem \ref{ApproximationTheorem}. As above choosing $\eps = \sqrt[2]{j} \; 2^{-j}$ yields Theorem \ref{GeneralApproximation} for \nameref{algorithm} with a slightly modified probability of success and slightly different dependencies on $C_1$ and $\tilde{C}_\x$ in \eqref{equ:GeneralApprox}.

\begin{theorem} \label{ApproximationTheorem2}
	Let $\M \subset \St^{D-1}$ be given by its empirical GMRA for some levels $j_0 \le j \le J$ from samples $X_1,...,X_n$ for $n \ge n_\text{min}$ (defined in Corollary \ref{ImportantBounds}), such that $0<C_1 < 2^{j_0 + 1}$ where $C_1$ is the constant from GMRA properties \eqref{GMRA2b} and \eqref{tube} for a GMRA structure constructed with centers ${\bf c}'_{j,k}$ and with the $\Phi_{j,k}^T \Phi_{j,k}$ as $\P_{V_{j,k}}$. Fix $j$ and assume that $\dist(\0, \Mh_j) \ge 1/2$. Further let
	\begin{align*}
	m \ge 64\max\{C',\bar{C}\} C_1^{-6} 2^{6(j+1)} (w(\M) + C\sqrt{dj})^2.
	\end{align*}
	where $C'$ is the constant from Theorem \ref{NoisyOneBit}, $\bar{C}$ from Theorem \ref{RUT} and $C$ from Lemma \ref{BoundOfGaussianWidth}. Then, with probability at least $1 - \mathcal{O}\left( 2^{jd} \exp (-2^{jd}) + \exp(\delta^2 m) \right)$ the following holds for all $\x \in \M$ with one-bit measurements $y = \sign(A \x)$ and GMRA constants $\tilde{C}_\x$ from property \eqref{3b} satisfying $\tilde{C}_\x < 2^{j-1}$: The approximations $\x^\ast$ obtained by \nameref{algorithm2} fulfill
	\begin{align*}
	\| \x - \x^\ast \|_2 &\le \left( {2}\left(\tilde{C}_\x + \frac{C_1}{8} \right) 2^{-\frac{j}{2}} + \sqrt{C_1} \sqrt[4]{\log \left( \frac{4e}{C_1} \right)}
	+ \sqrt{22\tilde{C}_\x + \frac{55}{4}C_1} \sqrt[4]{\log \left( \frac{2e}{(2\tilde{C}_\x + \frac{5}{4} C_1)} \right)} \right) \sqrt[4]{j} 2^{-\frac{j}{2}}.
	\end{align*}
\end{theorem}

\begin{proof}
	\renewcommand{\alph}{\widehat{\alpha}}
	The proof consists of the same three steps as the one of Theorem \ref{ApproximationTheorem}. First, we apply Lemma \ref{chjk'Bound} in \textbf{(I)}. By the GMRA axioms this supplies an estimate for $\| \x - \P_{j,k'}(\x) \|_2$ with high probability (recall that $\P_{j,k'}(\x)$ will be $\P_{V_{j,k'}} (\x - {\bf c}'_{j,k'}) + {\bf c}'_{j,k'}$ in this case). In \textbf{(II)} we use \textbf{(I)} to deduce a bound on $\| \x - \Ph_{j,k'}(\x) \|_2$, and then use Theorem \ref{NoisyOneBit} to bound the distance between $\Ph_{j,k'}(\x)/\| \Ph_{j,k'}(\x) \|_2$ and the minimum point $\x^*$ of
	\begin{align} \label{xbar2}
	\x^* = \argmin_\z \sum_{l = 1}^m (-y_l) \langle \av_l,\z \rangle, \quad \text{subject to } \z \in K := \conv \left( \P_\St (\widehat{P}_{j,k'} \cap \B(\0,2)) \right)
	\end{align}
	with high probability.  Taking the union bound over all events, part \textbf{(III)} then concludes with an estimate of the distance $\| \x - \x^\ast \|_2$ by combining \textbf{(I)} and \textbf{(II)}. \\
	
	\blue{\textbf{(I)}} Set $\delta = C_1 2^{-j-1}$ and recall that $C_1 2^{-j-2} < 1/2$ by assumption which implies by GMRA property \eqref{tube} that all centers in $\C'_j$ are closer to $\St^{D-1}$ than $1/2$, i.e.\ $1/2 \le \| \cj'_{j,k} \|_2 \le 3/2$. Moreover, Corollary \ref{ImportantBounds} holds with probability at least $1 - \mathcal{O}(2^{jd} \exp (-2^{jd}))$ and implies $\| \ch_{j,k} - \cj'_{j,k} \|_2 \le (C_1/12) 2^{-j-2} \le 1/4$. Hence, by triangle inequality $1/4 \le \| \ch_{j,k} \|_2 \le 7/4$. From this and \eqref{WidthInequality} we deduce
	\begin{align} \label{ProjectedMeanWidth2}
	w(\P_\St(\Ch_j)) \le \gamma(\P_\St(\Ch_j)) \le 4\gamma(\Ch_j) \le 8w(\Ch_j) + 4\sqrt{\frac{2}{\pi}} \dist(\0, \Ch_j) \le 8w(\Ch_j) + 8.
	\end{align}
	As $\Ch_j \subset \Mh_j$ we know by Corollary \ref{BoundOfGaussianWidth2} and \eqref{ProjectedMeanWidth2} that
	\begin{align}
	\begin{split} \label{eq:anothermbound}
	     m &\ge 16 \bar{C} \delta^{-6} (2w(\M) + 2C\sqrt{dj})^2 \ge 16\bar{C} \delta^{-6} (2w(\M) + 4w(\Ch_j) + 10)^2 \\
	     &\ge 4\bar{C} \delta^{-6} (4w(\M) + 8w(\Ch_j) + 20)^2 \ge 4\bar{C} \delta^{-6} (4w(\M) + w(\P_\St(\Ch_j)) + 12)^2 \\
	     &\ge \bar{C} \delta^{-6} (8w(\M) + 2w(\P_\St(\Ch_j)) + 24)^2 \ge \bar{C} \delta^{-6} (w(\M \cup \P_\St(\Ch_j)) + 19)^2 \\
	     &\ge \bar{C} \delta^{-6} \max \left\{ w(\M \cup \P_\St(\Ch_j))^2, \frac{2}{\pi} \right\}.
	\end{split}
	\end{align}
	Hence, Lemma \ref{cjk'Bound} implies
	\begin{align*}
	\| \x - \cj'_{j,k'} \|_2 \le 16 \max \left\{ \left\| \x- \cj'_{j,k_j(\x)} \right\|_2, C_1 2^{-j-1} \right\}.
	\end{align*}
	with probability at least $1 - \mathcal{O}\left( 2^{jd} \exp (-2^{jd}) + \exp(\delta^2 m) \right)$. By GMRA property \eqref{3b} we get
	\begin{align} \label{ApproxHat}
	\| \x - \P_{j,k'}(\x) \|_2 \le \tilde{C}_\x 2^{-j}
	\end{align}
	for some constant $\tilde{C}_\x$.\\
	
	\blue{\textbf{(II)}} Define $\alph = \left\| \Ph_{j,k'}(\x) \right\|_2$. Note that $\Vert \x - \cj'_{j,k'} \Vert_2 \le 4$ as $\x \in \St^{D-1}$ and all $\cj'_{j,k}$ are close to the sphere by assumption. Hence,
	\begin{align*}
	\| \P_{j,k'}(\x) - \Ph_{j,k'}(\x) \|_2 &\le \| \cj'_{j,k'} + \P_{V_{j,k'}}(\x - \cj'_{j,k'}) - \ch_{j,k'} - \P_{\Vh_{j,k'}}(\x - \ch_{j,k'}) \|_2 \\
	&\le \| \cj'_{j,k'} - \ch_{j,k'} \|_2 + \Vert \P_{V_{j,k'}} - \P_{\Vh_{j,k'}} \Vert \Vert \x - \cj'_{j,k'} \Vert_2 + \Vert \cj'_{j,k'} - \ch_{j,k'} \Vert_2 \\
	&\le \frac{2}{12} C_1 2^{-j-2} + \frac{1}{12} C_1 2^{-j-2} \Vert \x - \cj'_{j,k'} \Vert_2 \le \frac{1}{2}C_1 2^{-j-2}
	\end{align*}
	by application of Corollary \ref{ImportantBounds}. This implies $1/4 \le \alph \le 7/4$ as $\x \in \St^{D-1}$ and
	\begin{align} \label{x-Phat}
	\| \x - \Ph_{j,k'}(\x) \|_2 \le \| \x - \P_{j,k'}(\x) \|_2 + \| \P_{j,k'}(\x) - \Ph_{j,k'}(\x) \|_2 \le \tilde{C}_\x 2^{-j} + \frac{1}{2}C_1 2^{-j-2} \le \frac{3}{4}
	\end{align}  
	by \eqref{ApproxHat} and the assumption that $\max\{ \tilde{C}_\x, C_1/4 \} \cdot 2^{-j} \le 1/2$. As before we create the setting of Theorem \ref{NoisyOneBit}.
	
	Define $\tilde{\x} := \Ph_{j,k'}(\x)/\alph \in \St^{D-1}$, $\tilde{\y} := \sign(A\tilde{\x}) = \sign(A\Ph_{j,k'}(\x))$, $K = \conv ( \P_\St (\widehat{P}_{j,k'} \cap \B(\0,2)))$ and $\tau := (2\tilde{C}_\x + \frac{5}{4} C_1) 2^{-j}$. If applied to this, Theorem \ref{NoisyOneBit} would give the desired bound on $\| \tilde{\x} - \x^* \|_2$. We first have to check $d_H(\tilde{\y},\y) \le \tau m$. Recall that $\frac 1 \alph\leq 4$ and as $\alph>0$ one has $\alph w(K)=w(\alph K)$. By applying Corollary \ref{BoundOfGaussianWidth2} again we have that
	\begin{align*}
	m &\ge 64\max\{ C',\bar{C} \} \delta^{-6} (w(\M) + C\sqrt{dj})^2 \ge 4 \bar{C} \delta^{-6} (2w(\M) + 2w(\Mh_j) + 5)^2 \\
	&\ge \bar{C} \delta^{-6} (8w(\M) + 8w(\Mh_j) + 20)^2 \ge \bar{C} \delta^{-6} w \left(\M \cup \frac{\Mh_j}{\alph} \right)^2 \\
	&\ge \bar{C} \delta^{-6} w \left( \left( \M \cup \frac{\Mh_j}{\alph} \right) \cap B(\0,1) \right)^2.
	\end{align*}
	We can now use Theorem \ref{RUT}, Lemma \ref{Lem:NormalizedGeo} and $\| \tilde{\x} - \Ph_{j,k'}(\x) \|_2 = |1 - \alph| \le \| \x - \Ph_{j,k'}(\x) \|_2$ to obtain
	\begin{align*}
	\frac{d_H(\tilde{\y},\y)}{m} &= d_A(\tilde{\x},\x) \le d_G(\tilde{\x},\x) + 2\delta \le \| \tilde{\x} - \x \|_2 + 2\delta \le \| \tilde{\x} - \Ph_{j,k'}(\x) \|_2 + \| \Ph_{j,k'}(\x) - \x \|_2 + 2\delta \\
	&\le 2 \| \Ph_{j,k'}(\x) - \x \|_2 + 2\delta \le 2\tilde{C}_\x 2^{-j} + C_1 2^{-j-2} + 2\delta\qquad\\
	&\le (2\tilde{C}_\x + \frac{5}{4} C_1) 2^{-j} = \tau
	\end{align*}
	with probability at least $1-2\exp(-c \delta^2 m)$. Assuming the above events hold true we can apply Theorem \ref{NoisyOneBit} as by Corollary \ref{BoundOfGaussianWidth2}, in analogy to \eqref{eq:anothermbound} and \eqref{eq:Kwidth}, that
	\begin{align*}
	m &\ge 4 C' \delta^{-6} (2w(\M) + 4w(\M_j) + 4w(\Mh_j) + 10)^2 \\
	&\ge C' \delta^{-6} w(\P_\St (\widehat{P}_{j,k'} \cap \B(\0,2))) \\
	&\ge C' \delta^{-6} w(K)^2
	\end{align*}
	and obtain with probability at least $1-8\exp(-c \delta^2 m)$
	\begin{equation} \label{ApproxHat2}
	\| \tilde{\x} - \x^* \|_2^2 \le \delta \sqrt{\log \left( \frac{e}{\delta} \right)} + 11\tau \sqrt{\log \left( \frac{e}{\tau} \right)}.
	\end{equation}
	
	\blue{\textbf{(III)}} We conclude as in Theorem \ref{ApproximationTheorem}. Recall that $\| \tilde{\x} - \Ph_{j,k'}(\x) \|_2 = |1 - \alpha| \le \| \x - \P_{j,k'}(\x) \|_2 \le (\tilde{C}_\x + \frac{C_1}{8}) 2^{-j}$. By union bound we obtain with probability at least $1 - \mathcal{O}\left( 2^{jd} \exp (-2^{jd}) + \exp(\delta^2 m) \right)$
	\begin{align*}
	\| \x - \x^\ast \|_2 &\le \| \x - \Ph_{j,k'}(\x) \|_2 + \| \Ph_{j,k'}(\x) - \tilde{\x} \|_2 + \| \tilde{\x} - \x^* \|_2 \\
	&\le  2\| \x - \Ph_{j,k'}(\x) \|_2 + \sqrt{ \delta \sqrt{\log \left( \frac{e}{\delta} \right)} + 11 \tau \sqrt{\log \left( \frac{e}{\tau} \right)} }\\
	&\le  {2} \left(\tilde{C}_\x + \frac{C_1}{8} \right) 2^{-j} + \sqrt{C_1} 2^{\frac{-j-1}{2}} \sqrt[4]{\log \left( \frac{e}{C_1 2^{-j-1}} \right)}
	+ \sqrt{22\tilde{C}_\x + \frac{55}{4}C_1} 2^{-\frac{j}{2}} \sqrt[4]{\log \left( \frac{e}{(2\tilde{C}_\x + \frac{5}{4} C_1) 2^{-j}} \right)}\\
	&\le \left( {2} \left(\tilde{C}_\x + \frac{C_1}{8} \right) 2^{-\frac{j}{2}} + \sqrt{C_1} \sqrt[4]{\log \left( \frac{4e}{C_1} \right)}
	+ \sqrt{22\tilde{C}_\x + \frac{55}{4}C_1} \sqrt[4]{\log \left( \frac{2e}{(2\tilde{C}_\x + \frac{5}{4} C_1)} \right)} \right) \sqrt[4]{j} 2^{-\frac{j}{2}}.
	\end{align*}
	As explained in the proof of Theorem \ref{ApproximationTheorem} the last step was simplified for notational reasons.
\end{proof}

\end{document}